%% file: main.tex
\newif\ifdraft\drafttrue  
\newif\iffull\fulltrue   
\newif\ifmuchlater\muchlaterfalse 
\newif\ifever\everfalse   
\newif\ifbackref\backreffalse 
\newif\ifapp\appfalse     
\newif\ifscr\scrfalse     
\newif\ifpump\pumptrue    
\newif\iftechrep\techrepfalse  

\makeatletter \@input{texdirectives} \makeatother

\iftechrep
\documentclass[a4paper]{report}
\usepackage[
  margin=3cm,
  includefoot,
  footskip=30pt,
]{geometry}
\usepackage{inconsolata}
\else
\documentclass[jcs]{iosart1c}
\fi

\usepackage[T1]{fontenc}
\usepackage{times}
\iftechrep\else\fi
\usepackage[sort,numbers]{natbib}
\let\cite=\citep

\usepackage{amsmath}
\usepackage{url}
\usepackage{xspace}
\usepackage{multirow}
\usepackage{array}
\usepackage{amsmath,amssymb,amsthm}
\usepackage{stmaryrd}
\usepackage{xcolor}
\usepackage{listings}
\usepackage{paralist}
\usepackage{fancyvrb}
\usepackage{latexsym}
\usepackage{bm}
\usepackage{morefloats}

\usepackage{tikz}
\usetikzlibrary{matrix,arrows,decorations.pathmorphing}

\input{temp/defns}

\usepackage{preamble}

\newcommand{\rulesabsifc}{\mathcal{R}^{\sf abs}}
\newcommand{\tiniFH}{\ensuremath \phi_{\rulesabsifc}}

\begin{document}

\iftechrep

\makeatletter
\renewenvironment{titlepage}
    {%
      \if@twocolumn
        \@restonecoltrue\onecolumn
      \else
        \@restonecolfalse\newpage
      \fi
      \thispagestyle{empty}%
    }%
    {\if@restonecol\twocolumn \else \newpage \fi
    }
\makeatother

\title{{\huge A Verified Information-Flow Architecture}}

\newcommand{\authorcontent}{%
  {\large Arthur Azevedo de Amorim}\\
    {\em University of Pennsylvania, Philadelphia, PA, USA}\\[2ex]
  {\large Nathan Collins}\\
    {\em Galois Inc, Portland, OR, USA}\\[2ex]
  {\large Andr\'e DeHon}\\
    {\em University of Pennsylvania, Philadelphia, PA, USA}\\[2ex]
  {\large Delphine Demange}\\
    {\em Universit\'e Rennes 1 / IRISA, Rennes, France}\\[2ex]
  {\large C\u{a}t\u{a}lin Hri\c{t}cu}\\
    {\em INRIA, Paris, France}\\[2ex]
  {\large David Pichardie}\\
    {\em ENS Rennes / IRISA, Rennes, France}\\[2ex]
  {\large Benjamin C. Pierce}\\
    {\em University of Pennsylvania, Philadelphia, PA, USA}\\[2ex]
  {\large Randy Pollack}\\
    {\em Harvard University, Boston, MA, USA}\\[2ex]
  {\large Andrew Tolmach}\\
    {\em Portland State University, Portland, OR, USA}\\[2ex]
}

\author{\authorcontent}

\maketitle

\else

\begin{frontmatter}

\title{A Verified Information-Flow Architecture}
\runningtitle{A Verified Information-Flow Architecture}
\subtitle{Submission to special issue on Verified Information Flow
  Security}

\maketitle

\author[A]{\fnms{Arthur} \snm{Azevedo de Amorim}},
\author[G]{\fnms{Nathan} \snm{Collins}},
\author[A]{\fnms{Andr\'e} \snm{DeHon}},
\author[E]{\fnms{Delphine} \snm{Demange}},
\author[C]{\fnms{C\u{a}t\u{a}lin} \snm{Hri\c{t}cu}},
\author[F]{\fnms{David} \snm{Pichardie}},
\author[A]{\fnms{Benjamin C.} \snm{Pierce}},
\author[D]{\fnms{Randy} \snm{Pollack}} and
\author[B]{\fnms{Andrew} \snm{Tolmach}}
\runningauthor{A. Azevedo de Amorim et al.}
\address[A]{University of Pennsylvania, Philadelphia, PA, USA}
\address[B]{Portland State University, Portland, OR, USA}
\address[C]{INRIA, Paris, France}
\address[D]{Harvard University, Boston, MA, USA}
\address[E]{Universit\'e Rennes 1 / IRISA, Rennes, France}
\address[F]{ENS Rennes / IRISA, Rennes, France}
\address[G]{Galois Inc, Portland, OR, USA}
\fi

\begin{abstract}
SAFE is a clean-slate design for a highly secure computer
system, with pervasive mechanisms for tracking and limiting
information flows.  At the lowest level, the SAFE hardware supports
fine-grained programmable tags, with efficient and flexible
propagation and combination of tags as instructions are executed.
The operating system virtualizes these generic facilities to present
an information-flow abstract machine that allows user programs to
label sensitive data with rich confidentiality policies.  We present
a formal, machine-checked model of the key hardware and software
mechanisms used to dynamically control information flow in SAFE and
an end-to-end proof of noninterference for this model.

We use a refinement proof methodology to propagate the noninterference
property of the abstract machine down to the concrete machine
level. We use an intermediate layer in the refinement chain that
factors out the details of the information-flow control policy and
devise a code generator for compiling such information-flow policies
into low-level monitor code. Finally, we verify the correctness of
this generator using a dedicated Hoare logic that abstracts from
low-level machine instructions into a reusable set of verified
structured code generators.
\end{abstract}

\iftechrep
\renewcommand\thesection{\arabic{section}}
\tableofcontents
\newpage
\else

\end{frontmatter}
\fi

\section{Introduction}

The SAFE design is motivated by the conviction that the insecurity
of present-day computer systems is due in large part to legacy design
decisions left over from an era of scarce hardware resources.  The
time is ripe for a complete rethink of the entire system stack with
security as the central focus.  In particular, designers should be
willing to spend more of the abundant processing power available on
today's chips to improve security.

\ifmuchlater
\bcp{Integrate some of the discussion of specific legacy decisions that we
  don't like... (there's a nice set of keywords on an early slide in the CSF
talk)}\ch{In iffalse below}\ch{Maybe only the subset that's
  relevant to PL people (as opposed to systems people)?}
\fi

A key feature of SAFE is that every piece of data, down to the word
level, is annotated with a {\em tag} representing policies that govern
its use.
While the tagging mechanism is very general~\cite{pump_asplos2015,micropolicies2015},
one particularly interesting use of tags is for representing {\em
  information-flow control (IFC)} policies.
For example, an individual record might be tagged
``This information should only be seen by principals {\tt Alice} or {\tt
  Bob},'' a function pointer might be tagged ``This code is trusted to work
with {\tt Carol}'s secrets,'' or a string might be tagged ``This came from
the network and has not been sanitized yet.''
Such tags representing IFC policies can involve arbitrary
sets of principals, and principals themselves can be dynamically allocated
to represent an unbounded number of entities within and outside the system.

At the programming-language level, rich IFC policies have been
extensively explored, with many proposed designs for
static~\citeEtcShort{sabelfeld03:lang_based_security,HedinS11%
\iffull,%
  myers99:thesis,%
  Myers99,%
  zdancewic02:thesis,%
  pottier03:flowcaml%
\fi} and
dynamic~\citeEtcShort{AustinF09,
RussoS10, StefanRMM11, HedinS12%
\iffull,%
  SabelfeldR09,%
  Guernic07,%
  Pistoia:2007,%
  AustinF:2010,%
  AustinF12,%
AustinFA12\fi}
enforcement
mechanisms and a huge literature on their formal properties%
~\cite[etc.]{sabelfeld03:lang_based_security,HedinS11}. Similarly,
operating systems with information-flow tracking have been a staple of
the OS literature for over a decade%
~\citeEtcShort{KrohnT09%
\iffull,%
  efstathopoulos05:asbestos,%
  ZeldovichBKM11,%
  krohn07:flume,%
  seL4:Oakland2013,
  ZeldovichBKM11%
\fi}.
But progress at the hardware level has been more
limited, with most proposals concentrating on hardware acceleration for
taint-tracking schemes%
~\citeEtcShort{secure_flow_track_asplos2004, raksha_isca2007, flexitaint,
deng_dsn2012%
\iffull,%
  deng_micro2010,%
  CrandallC04,%
  pointer_taintedness_dsn2005%
\fi}.
SAFE extends the state of the art in two significant ways.
First, the SAFE machine offers hardware support for sound and
efficient purely-dynamic tracking of both explicit and implicit flows
(i.e., information leaks through both data and control flow) for
\emph{arbitrary} machine code programs---not just programs accepted by
static analysis, or produced by translation or transformation.
Moreover, rather than using just a few ``taint bits,'' SAFE associates
a word-sized tag to every word of data in the machine---both memory
and registers.
In particular, SAFE tags can be pointers to arbitrary data structures
in memory.
The interpretation of these tags is left entirely to software: the
hardware just propagates tags from operands to results as each
instruction is executed, following software-defined rules.
Second, the SAFE design has been informed from the start by an intensive
effort to formalize critical properties of its key mechanisms
and produce machine-checked proofs,
in parallel with the design and implementation of its hardware
and system software.
Though some prior work (surveyed in~\autoref{sec:relwork}) shares some of
these aims, to the best of our knowledge no project has attempted this
combination of innovations.

Abstractly, the tag propagation rules in SAFE can be viewed as a
partial function from argument tuples of the form ({\em opcode}, {\em
pc tag}, {\em argument$_1$ tag}, {\em argument$_2$ tag}, \ldots) to
result tuples of the form ({\em new pc tag}, {\em result tag}),
meaning ``if the next instruction to be executed is {\em opcode}, the
current tag of the program counter (PC) is {\em pc tag}, and the
arguments expected by this opcode are tagged {\em argument$_1$ tag},
etc., then executing the instruction is allowed and, in the new state
of the machine, the PC should be tagged {\em new pc tag} and any new
data created by the instruction should be tagged {\em result tag}.''
(The individual argument-result pairs in this function's graph are
called {\em rule instances}, to distinguish them from the symbolic
{\em rules} used at the software level.)
In general, the graph of this function {\em in extenso} will be huge;
so, concretely, the hardware maintains a {\em cache} of recently-used
rule instances.
On each instruction dispatch (in parallel with the logic implementing
the usual behavior of the instruction---\EG addition), the hardware
forms an argument tuple as described above and looks it up in the rule
cache.  If the lookup is successful, the result tuple includes a new
tag for the PC and a tag for the result of the instruction (if any);
these are combined with the ordinary results of instruction execution
to yield the next machine state.  Otherwise, if the lookup is
unsuccessful, the hardware invokes a {\em cache fault handler}---a
trusted piece of system software with the job of checking whether the
faulting combination of tags corresponds to a policy violation or
whether it should be allowed.
In the latter case, an appropriate rule
instance specifying tags for the instruction's results is added to the
cache, and the faulting instruction is restarted.  \ifmuchlater\amd{potentially
  draw analogy with TLB for VM page table} \jms{Yes, that's the right
  analogy. The locality of reference behavior so far looks very good.}\fi
Thus, the hardware is generic and the interpretation of policies
(e.g., IFC, memory safety or control flow
integrity~\cite{pump_asplos2015,micropolicies2015})
is programmed in software, with the results cached in hardware for
common-case efficiency.

The first contribution of this paper is to explain and formalize, in
the Coq proof assistant~\cite{coq_manual}, the key ideas in this
design via a simplified model of the SAFE machine, embodying its
tagging mechanisms in a distilled form and focusing on enforcing IFC
using these general mechanisms.
In \autoref{sec:fullSAFE}, we outline the features of the full SAFE
system and enumerate the most significant simplifications in our
model.
In \autoref{sec:abstract}, we present the high-level programming
interface of our model, embodied by an {\em abstract IFC machine} with
a built-in, purely dynamic IFC enforcement mechanism and an abstract
lattice of IFC labels.
We then show, in three steps, how this abstract machine can be implemented
using the low-level mechanisms we propose.
The first step introduces a {\em symbolic IFC rule machine} that
reorganizes the semantics of the abstract machine, splitting out the
IFC enforcement mechanism into a separate judgment parameterized by a
{\em symbolic IFC rule table} (\autoref{sec:quasi}).
The second step defines a generic {\em concrete machine}
(\autoref{sec:conc}) that provides low-level support for efficiently
implementing many different high-level policies (IFC and others) with
a combination of a hardware {\em rule cache} and a software {\em fault
  handler}.
The final step instantiates the concrete machine with a concrete fault
handler enforcing IFC.  We do this using an {\em IFC fault handler
  generator} (\autoref{sec:fault-handler}), which compiles the
symbolic IFC rule table into a sequence of machine instructions
implementing the IFC enforcement judgment.

Our second contribution is a machine-checked proof that this
simplified SAFE system is {\em correct} and {\em secure}, in the sense
that user code running on the concrete machine equipped with the IFC
fault handler behaves the same way as on the abstract machine and
enjoys the standard {\em noninterference} property that ``high inputs
do not influence low outputs.''
\ifmuchlater\bcp{Say something more about why this is hard / unexpected?  Or is
  what comes below enough to make the point?}\amd{could
  we cite some papers on the historical challenges of getting VM fault
  handlers completely correct?}\bcp{Sounds like a good idea.  Is there a
  canonical one?}  \jms{The paper I think addresses this is by Alkassar,
  Schirmer and Starostin at TACAS 08}\fi{}
The interplay of the concrete machine and fault handler is
complex, so some proof abstraction is essential.
\iffull (Previous projects such as the
  \mbox{CompCert} compiler~\cite{Leroy09},
  the seL4~\cite{Klein09sel4:formal,seL4:Oakland2013} and
  CertiKOS~\cite{GuKRSWWZG15, Shao15} microkernels,
  and the RockSalt SFI checker~\cite{MorrisettTTTG12} have demonstrated
  the need for significant attention to organization in similar proofs.) \fi
In our proof architecture, a first abstraction layer is based on {\em
  refinement}.
This allows us to reason in terms of a high-level view of memory, ignoring
the concrete implementation of IFC labels, while setting up the
intricate indistinguishability relation used in the noninterference proof.
A second layer of abstraction is required for reasoning about the
correctness of the fault handler.  Here, we rely on a verified custom
Hoare logic that abstracts from low-level machine instructions into a
reusable set of verified structured code generators.

In \autoref{sec:fault-handler-correct} we prove that the IFC fault
handler generator correctly compiles a symbolic IFC rule table and a
concrete representation of an abstract label lattice into an
appropriate sequence of machine instructions.
We then introduce a standard notion of refinement
(\autoref{sec:refinement}) and show that the concrete machine running
the generated IFC fault handler refines the abstract IFC machine and
vice-versa, using the symbolic IFC rule machine as an intermediate
refinement point in each direction of the proof
(\autoref{sec:qa-c-ref}).
In our deterministic setting, showing refinement in both directions
guarantees that the concrete machine does not diverge or get stuck when
handling a fault.
We next introduce a standard {\em termination-insensitive
  noninterference (TINI)} property (\autoref{sec:ni}) and show that it
holds for the abstract machine.
Since deterministic TINI is preserved by refinement, we conclude that
the concrete machine running the generated IFC fault handler also
satisfies TINI.
In~\autoref{sec:extensions}, we explain how the programming model and
formal development of the first sections can be extended to
accommodate two important features: dynamic memory allocation and tags
representing sets of principals. This extension, carried out after the
development of the basic model, gives us confidence in the robustness
of our methodology.
We close with a survey of related work~(\autoref{sec:relwork}) and a
discussion of future directions~(\autoref{sec:concl}).
Our Coq formalization is available at
\url{https://github.com/micro-policies/verified-ifc}.


\ifmuchlater\bcp{One point about genericity is that we intend to support a wide
  variety of label models, building on the foundation
  in~\cite{GenLabels}...  Work this in someplace...}

\bcp{Maybe incorporate more text from DP and DD...}
\ch{They are in iffalse below}
\fi

\iffull
A preliminary abridged version of this work appeared in the
proceedings of the POPL 2014 conference~\cite{PicoCoq2013}.
This extended and improved version includes:
\begin{itemize}
\item more examples and clarifying explanations in the formal sections;
\item a more detailed technical description of the formalization:
      the semantics of the abstract, symbolic and concrete machines,
      the language for expressing symbolic IFC rules, our verified
      structured code generators, and TINI-preserving refinements;
\item more details of the proofs; \dd{only unwinding conditions imply TINI}
\item a more extensive discussion of related work, including more
  recent work on transplanting the tagging mechanism of SAFE onto a
  mainstream RISC processor~\cite{Dover16} and using it to enforce properties
  beyond IFC~\cite{micropolicies2015, pump_asplos2015}.
\end{itemize}

\section{Overview of SAFE}
\label{sec:fullSAFE}

To establish context, we begin with a brief overview of the full SAFE
system, concentrating on its OS- and hardware-level features.  More
detailed descriptions can be found elsewhere~\cite{SAFEPLOS11,
  interlocks_ahns2012, Exceptional, GenLabels,
  near_assoc_cache_fpga_2013, TestingNI, LowFat2013,pump_asplos2015}.
SAFE's system software performs process scheduling, stream-based
interprocess communication, storage allocation and garbage collection, and
management of the low-level tagging hardware (the focus of this paper).  The
goal is to organize these services as a collection of mutually suspicious
compartments following the principle of least privilege (a {\em
  zero-kernel OS}~\cite{shrobe09:tiara_nicecap_report}), so that an attacker
would need to compromise multiple compartments to gain complete control of
the machine\ifmuchlater\bcp{I'm not 100\% comfortable with claiming that
  this is something that has already been achieved.  Indeed, I remain a
  little skeptical that we will ever get all the way there, though I do
  believe that even a reasonable approximation of a ZKOS would be a good
  thing.  Have tried to weaken the claim to an aspiration, but I'm not sure
  it's weak enough yet.}\fi.  It is programmed in a combination of assembly
and {\em Tempest}, a new low-level systems
programming language.

The SAFE hardware integrates a number of mechanisms for eliminating
common vulnerabilities and supporting higher-level security
primitives.  To begin with, SAFE is (dynamically) typed at the
hardware level: each data word is indelibly marked as a number, an
instruction, a pointer, etc.  Next, the hardware is memory safe: every
pointer consists of a triple of base, bounds, and offset (compactly
encoded into 64 bits~\cite{interlocks_ahns2012, LowFat2013}), and
every pointer operation includes a hardware bounds
check~\cite{LowFat2013}.
Finally, the hardware associates each word in the registers and
memory, as well as the PC, with a large (59-bit) tag.  The hardware
rule cache, enabling software-specified propagation of tags from
operands to result on each machine step, is implemented using a
combination of multiple hash functions to approximate a
fully-associative cache~\cite{near_assoc_cache_fpga_2013}.

An unusual feature of the SAFE design is that formal modeling and
verification of its core mechanisms have played a central role in the design
process since the beginning. The original goal---formally specifying and
verifying the entire set of critical runtime services---proved to be
too ambitious, but key security properties of simplified models have been verified
both at the level of {\em Breeze}~\cite{Exceptional} (a mostly functional,
security-oriented, dynamic language used for user-level programming on SAFE)
and, in the present work, at the hardware and abstract machine level.
We also used random testing of properties like noninterference as
a means to speed the design process~\cite{TestingNI}.

Our goal in this paper is to develop a clear, precise, and mathematically
tractable model of one of the main innovations in the SAFE design:
its scheme for efficiently supporting high-level data use policies
using a combination of hardware and low-level system software.
To make the model
easy to work with, we simplify away many important facets of the real SAFE
system.  In particular,
\begin{inparaenum}[(i)]
\item we focus only on IFC and noninterference, although the tagging
  facilities of the SAFE machine are generic and can be applied to
  other policies (more recent work illustrates this
  point~\cite{pump_asplos2015, PicoCoq2013}; we return to it
  at the end of~\autoref{sec:relwork});
\item we ignore the Breeze and Tempest programming languages and concentrate
on the hardware and runtime services;
\item we use a stack instead of registers, and we distill the instruction
set to just a handful of opcodes;
\item we drop SAFE's fine-grained privilege separation in favor
of a more conventional user-mode / kernel-mode dichotomy;
\item we shrink the rule cache to a single entry (avoiding issues of
replacement and eviction) and maintain it in kernel memory, accessed by
ordinary loads and stores, rather than in specialized cache hardware;
\item we focus on termination-insensitive noninterference and omit a
  large number of more advanced IFC-related concepts that are
  supported by the real SAFE system (dynamic principals, downgrading,
  public labels, integrity, clearance, etc.);
\item we handle exceptional conditions, including potential security
violations, by simply halting the whole machine;
and
\item most importantly, we ignore concurrency, process scheduling, and
  interprocess communication, assuming instead that the whole machine
  has a single, deterministic thread of control.
\end{inparaenum}
We believe that most of these restrictions can be lifted without
fundamentally changing the structure of the model or of the
proofs. For instance, recent follow-on work by some of the
authors~\cite{jfpsubmission2014} discusses a mechanized proof of
noninterference for a similar abstract machine featuring registers and
a richer IFC policy.
The absence of concurrency is a particularly significant
simplification, given that we are talking about an operating
system that offers IFC as a service.
However, we conjecture that it may be possible to add concurrency to our formalization,
while maintaining a high degree of determinism, by adapting the approach
used in the proof of noninterference for the seL4
microkernel~\cite{seL4:Oakland2013, MurrayMBGK12}.
We return to this point in \autoref{sec:concl}.

\section{Abstract IFC Machine}
\label{sec:abstract}

We begin the technical development by defining a very simple
stack-and-pointer machine with ``hard-wired'' dynamic IFC.
This machine concisely embodies the IFC mechanism we want to provide
to higher-level software and serves as a specification for the
symbolic IFC rule machine (\autoref{sec:quasi}) and for the concrete
machine (\autoref{sec:conc}) running our IFC fault handler
(\autoref{sec:fault-handler}).
The three machines share a tiny instruction set (\autoref{fig:instrs})
designed to be a convenient target for compiling the symbolic IFC rule
table\iffull{} into machine instructions\fi{} (the Coq development
formalizes several other instructions\iffull{}, including {\sf Sub},
{\sf Pop}, a variant of {\sf Call} that takes a variable number of
arguments and a variant of {\sf Ret} that allows returning a result on
the stack). 
\else). \fi{}%
All three machines use a fixed {\em instruction memory} $\iota$, a
partial function from \iffull (non-negative) integer \fi addresses to
instructions.

\begin{figure}[tbp!]
\ottgrammartabular{
\ottinstr\ottinterrule
}
\vspace*{-3ex}
\caption{Instruction set}
\label{fig:instrs}
\end{figure}

\newcommand{\secret}{ \top }
\newcommand{\public}{ \bot }

The machine manipulates integers (ranged over by $\ottnt{n}$, $\ottnt{m}$, and
$\ottnt{p}$); unlike the real SAFE machine, we make no distinction between
raw integers and pointers (we re-introduce this distinction in
\autoref{sec:extensions}).  Each integer is marked with an individual
IFC {\em label} (ranged over by $\ottnt{L}$) that denotes its security
level. We call a pair of an integer $n$ and its corresponding label
$L$ an {\em atom}, written $\ottnt{n}  \mathord{\scriptstyle @}  \ottnt{L}$ and ranged over by $\ottnt{a}$.
We assume that IFC labels $L$ form a set \(\mathcal{L}\) equipped with
a partial order ($ \le $), a least upper bound operation ($ \mathord{\vee} $),
and a bottom element ($ \bot $), but do not place further
requirements on them. This generality allows us to model many
different kinds of labels present in existing IFC
systems~\cite{GenLabels}.
For instance we might take $\mathcal{L}$ to be the set of
levels $\{\public,\secret\}$ with $\public \mathrel{ \le } \secret$ and
$\public \mathbin{ \mathord{\vee} } \secret = \secret$.
Alternatively, we could consider a richer set of labels, such as
finite sets of principals ordered by set inclusion, as discussed
in~\autoref{sec:extensions}.

An {\em abstract machine state} $\langle \mu~[ \sigma ]~\ottnt{pc}\rangle$
consists of a data memory $\mu$, a stack $\sigma$, and a program
counter $\ottnt{pc}$.  (We sometimes drop the outer brackets.)
The {\em data memory} $\mu$ is a partial function from integer
addresses to atoms. We write $ \mu ( \ottnt{p} ) \leftarrow  \ottnt{a} $ for the memory that
coincides with $\mu$ everywhere except at $p$, where its value is $a$.
The {\em stack} $\sigma$ is essentially a list of atoms, but we
distinguish stacks beginning with return addresses (written $\ottnt{pc}  \mathord{;}\,  \sigma$) from
ones beginning with regular atoms
(written $\ottnt{a}  \mathord{,}\,  \sigma$).
\iffull Formally, stacks are lists with two ``cons'' constructors,
written ``$,$'' and ``$;$''.  This distinction is needed so that
stack-manipulating instructions treat frame markers specially;
for example, a program that $ \text{\sf Push} $es an integer and then attempts to
return to it is treated as erroneous by the operational semantics.
\fi
The {\em program counter} (PC) $\ottnt{pc}$ is an atom whose label is used to
track implicit flows, as explained below.

\begin{figure}[t]
\[
\begin{array}{lr}
\ottdruleStepXXAdd{} & \ottdruleStepXXPush{} \\[4em]
\ottdruleStepXXLoad{} & \ottdruleStepXXStore{} \\[4em]
\ottdruleStepXXJump{} & \ottdruleStepXXBnz{} \\[4em]
\ottdruleStepXXCall{} & \ottdruleStepXXRet{} \\[4em]
\multicolumn{2}{c}{\ottdruleStepXXOutput{}}
\end{array}
\]
\caption{Semantics of abstract IFC machine}
\label{fig:abstractSteps}
\end{figure}

The step relation of the abstract machine, written $\iota \vdash  { \aStep{ \mu_{{\mathrm{1}}} }{[  \sigma_{{\mathrm{1}}}  ]}{ \ottnt{pc_{{\mathrm{1}}}} }{ \alpha }{ \mu_{{\mathrm{2}}} }{[  \sigma_{{\mathrm{2}}}  ]}{ \ottnt{pc_{{\mathrm{2}}}} }{  } } $, is a partial function
taking a machine state to a machine state plus an output action
$\alpha$, which can be either an atom or the silent action
$ \tau $.  We generally omit the instruction memory
$\iota$ from transitions because it
is fixed.  Throughout the paper we consistently refer to
non-silent actions as \emph{events} (ranged over by $e$).

The stepping rules in \autoref{fig:abstractSteps} adapt a standard
purely dynamic IFC enforcement mechanism~\cite{AustinF09, RussoS10} to
a low-level machine, following recent work by
Hri\c{t}cu~\ETAL\cite{TestingNI}.
(Readers less familiar with the intricacies of dynamic IFC may find
some of these side conditions a bit mysterious.  A longer explanation
can be found in~\cite{TestingNI}, but the details are not critical for
present purposes.)
The rule for $ \text{\sf Add} $ joins ($ \mathord{\vee} $) the labels of the two operands
to produce the label of the result, which ensures that the result is
at least as classified as each of the operands.
\iffull
For example, suppose $\iota = [...,  \text{\sf Add} , ...]$ and $n$ is the index of
this $ \text{\sf Add} $ instruction.  Then
$ { \aStep{ \mu }{[  \ottsym{7}  \mathord{\scriptstyle @}  \bot  \mathord{,}\,  \ottsym{5}  \mathord{\scriptstyle @}  \top  ]}{ \ottnt{n}  \mathord{\scriptstyle @}  \bot }{ \tau }{ \mu }{[  \ottsym{12}  \mathord{\scriptstyle @}  \top  ]}{ \ottsym{(}  \ottnt{n}  \mathord{+}  \ottsym{1}  \ottsym{)}  \mathord{\scriptstyle @}  \bot }{  } } .$
\fi
The rule for $ \text{\sf Push} $ labels the integer constant added to the
stack as public ($ \bot $).
The rule for $ \text{\sf Jump} $ uses join to raise the label of the PC
by the label of the target address of the jump.
Similarly, $ \text{\sf Bnz} $ raises the label of the PC by the label
of the tested integer.
In both cases the value of the PC after the instruction depends
on data that could be secret, and we use the label of the PC to
track the label of data that has influenced control flow.
In order to prevent {\em implicit flows} (leaks exploiting the control
  flow of the program), the $ \text{\sf Store} $ rule joins the PC label
with the original label of the written integer and with the label of
the pointer through which the write happens.
Additionally, since the labels of memory locations are allowed to vary
during execution, we prevent leaking information via labels using a
``no-sensitive-upgrade'' check~\cite{zdancewic02:thesis, AustinF09} (the
$ \le $ precondition in the rule for $ \text{\sf Store} $).\iffull\footnote{
  More recent work further improves precision compared to
  the no-sensitive-upgrades policy~\cite{HedinS12,AustinF:2010,
    TestingNI,BichhawatRGH14b}.
  We adopted no-sensitive-upgrades in this work because it is simpler and
  requires less bookkeeping.}\fi{}
\ifever \ch{We could consider an extension in which we implement
  permissive upgrades and report on how easy/hard it is to change IFC
  mechanisms while adapting all proofs} \fi
This check prevents memory locations labeled public from being
overwritten when either the PC or the pointer through which the store
happens has been influenced by secrets.
The $ \text{\sf Output} $ rule labels the emitted integer with the
join of its original label and the current PC label.%
\footnote{We assume the observer of the events generated by $ \text{\sf Output} $ is
  constrained by the rules of information flow---i.e., cannot freely ``look
  inside'' bare events.  \ifmuchlater\bcp{Maybe add: This
    will also be the case even for our concrete machine, where atoms are
    labeled with integer tags {\em representing} information-flow
    labels.}\fi In the real SAFE machine, atoms being sent to the outside
  world need to be protected cryptographically; we are abstracting this
  away.}
Finally, because of the structured control flow imposed by the stack
discipline, the rule for $ \text{\sf Ret} $ can soundly restore the PC
label to whatever it was at the time of the $ \text{\sf Call} $.
This feature allows programmers to avoid \emph{label creep}---\IE
having the current PC label inadvertently go up when branching on
secrets unknowingly---by making judicious use of $ \text{\sf Call} $ and
$ \text{\sf Ret} $, but may require careful thought to be used correctly. Many
other solutions have been proposed to this problem, each with their
own strengths and weaknesses. Some systems, such as
LIO~\cite{StefanRBLMM12}, prevent label creep by maintaining a
\emph{clearance level} that serves as an upper bound on the PC label;
this, however, may lead to dynamic errors if a computation tries to
inspect a secret above its clearance.

All data in the machine's state are labelled, and this simple machine manages labels
to ensure noninterference as defined and proved
in \autoref{sec:ni}.  There are no instructions that dynamically raise
the label (classification) of an atom.  Such an instruction,
$ \text{\sf joinP} $, is added to the machine in \autoref{sec:extensions}.

\section{Symbolic IFC Rule Machine}
\label{sec:quasi}


In the abstract machine described above, IFC is tightly integrated
into the step relation in the form of side conditions on each
instruction.
In contrast, the concrete machine (\IE the ``hardware'')
described in~\autoref{sec:conc}
is generic, designed to support a wide range of software-defined
policies (IFC and other).
The machine introduced in this section serves as a bridge between these two
models.  It is closer to the abstract machine---indeed, its machine states
and the behavior of the step relation are identical.  The important
difference lies in the {\em definition} of the step relation, where all the
IFC-related aspects are factored out into a separate judgment.
\iffull
We can think
of the IFC mechanism as being implemented in a separate
``IFC rule processor'' distinct from the main ``CPU.''  In the concrete
machine, the CPU part will remain unchanged, but the IFC rule processor will
be implemented mostly in software (by the fault handler), with the hardware
only providing caching of rule instances.
\fi
While factoring out IFC enforcement into a separate reference
monitor\iffull~\cite{Schneider00}\fi{} is
commonplace~\cite{AskarovS09b,RussoS10,SabelfeldR09}, our approach
goes further. We define a small DSL for describing symbolic IFC rules
and obtain actual monitors by interpreting this DSL (in this section)
and by compiling it into machine instructions using verified
structured code generators (in \autoref{sec:fault-handler} and
\autoref{sec:fault-handler-correct}). This architecture makes it
easier to implement other IFC mechanisms (\EG permissive
upgrades~\cite{AustinF:2010}), beyond the simple one
in~\autoref{sec:abstract}.
Since the DSL compilation is verified, we
prove that the concrete machine of~\autoref{sec:conc} is
noninterfering when given \emph{any} correct monitor written in the
DSL. Showing that a monitor is correct, on the other hand, involves a
simple refinement proof (\autoref{lem:refAbsQuasi}), and a
noninterference proof for the abstract machine
(\autoref{thm:abstractIFCtini}), but is independent of the code
generation infrastructure and corresponding proofs.

\ifmuchlater
\bcp{As I'm reading over this text carefully, I'm finding it a bit awkward
  that we make such a big deal (notationally and in the text) about the IFC
  evaluation judgment being parameterized by $\mathcal{R}$ but then we don't
  bother parameterizing the step relation itself by $\mathcal{R}$.  Shouldn't
  we do this?  (If we don't do it, we should at least put a note in this
  section that $\mathcal{R}$ is also implicitly a parameter to the step
  relation.)}
\fi

More formally, each stepping rule of the new machine
(see~\autoref{fig:qa_step}) includes a uniform call to an {\em IFC
  enforcement} relation, which itself is parameterized by a {\em
  symbolic IFC rule table} $\mathcal{R}$.  Given the labels of the values
relevant to an instruction, the IFC enforcement
relation \begin{inparaenum}[(i)] \item checks whether the execution of
  that instruction is allowed in the current configuration, and \item
  if so, yields the labels to put on the resulting PC and on any
  resulting value.
\end{inparaenum}
This judgment has the form $ \vdash_{ \mathcal{R} } \ruleeval{ \ottsym{(}  L_{pc}  \mathord{,}\,  \ell_{{\mathrm{1}}}  \mathord{,}\,  \ell_{{\mathrm{2}}}  \mathord{,}\,  \ell_{{\mathrm{3}}}  \ottsym{)} }{ \ottnt{op} }{ L_{rpc} }{ L_r } $, where the 4-tuple on the left-hand side represents the input
PC label and three additional input labels (more precisely, optional
labels, as the number of relevant labels depends on the opcode but the
tuple is of fixed size), $op$ is an opcode, and $L_{rpc}$ and
$L_r$ are the resulting output labels (of which the second might be
ignored).

\iffull Let us illustrate, for
a few cases, how this new judgment is used in the stepping
relation (\autoref{fig:qa_step}).\fi{}
\iffull{T}\else{For example, t}\fi{}he stepping rule for $ \text{\sf Add} $
\iffull\else\ottusedrule{\ottdruleSAdd{}}\fi
passes three inputs to the
IFC enforcement judgment: $L_{pc}$, the label of
the current PC, and $\ottnt{L_{{\mathrm{1}}}}$ and $\ottnt{L_{{\mathrm{2}}}}$, the labels
of the two operands at the top of the stack. (The fourth element
of the input tuple is written as $\dummytag$ because it
is not needed for $ \text{\sf Add} $.)
The IFC enforcement judgment produces two labels:
$L_{rpc}$ is used to label the next
program counter ($n+1$) and $L_r$ is used to label the result value.
All the other stepping rules follow a similar scheme.  (The one for
$ \text{\sf Store} $ uses all four input labels.%
\iffull{} In this stepping rule the resulting
label $L_r$ is used to label the new value $m$ to be stored at
location $p$.\fi)

\iffull
\begin{figure}[t]
\[
\begin{array}{lr}
\ottdruleSAdd{} &
\ottdrulePush{} \\[4em]
\ottdruleLoad{} &
\ottdruleStore{} \\[4em]
\ottdruleJump{} &
\ottdruleBnz{} \\[4em]
\ottdruleCall{} &
\ottdruleRet{} \\[4em]
\multicolumn{2}{c}{\ottdruleOutput{}}
\end{array}
\]
\caption{Semantics of symbolic rule machine,
  parameterized by $\mathcal{R}$}
\label{fig:qa_step}
\end{figure}
\fi

A symbolic IFC rule table $\mathcal{R}$ describes a particular IFC
enforcement mechanism.
For instance, the rule table $\rulesabsifc$ corresponding to the IFC
mechanism of the abstract machine is shown in
\autoref{fig:rule_table}.  In general, a table $\mathcal{R}$ associates a
{\em symbolic IFC rule} to each instruction opcode (formally,
$\mathcal{R}$ is a total function).
Each of these rules is formed of three symbolic expressions:
\begin{inparaenum}[(i)]
\item a boolean expression indicating whether the execution of the
  instruction is allowed or not (\IE whether it violates the IFC
  enforcement mechanism);
\item a label-valued expression for $L_{rpc}$, the label of the next PC; and
\item a label-valued expression for $L_r$, the label of the result value, if there is one.
\iffull
In cases where $L_r$ is not
used by the corresponding opcode, we write $\ottsym{\_\_}$ to mean
``don't care,'' which is a synonym for $ \mathtt{BOT} $
(the symbolic representation of the $ \bot $ label).
\fi
\end{inparaenum}

\newenvironment{trivarray}{\begin{array}[t]{@{}l@{}}}{\vspace{-2.5ex}\end{array}}

\begin{figure}[tbp!]
{
\[
\setlength\arraycolsep{2pt}
\begin{array}{@{}l| p{100pt} p{60pt} @{\quad\;}l}
opcode     & {\it allow} & $e_{rpc}$ & e_{r} \\
\hline
\ottkw{add}    & $ \mathtt{TRUE} $ & $ \mathtt{LAB}_\mathit{pc} $ &   \mathtt{LAB}_1  ~ \sqcup ~  \mathtt{LAB}_2   \\
\ottkw{output}    & $ \mathtt{TRUE} $ & $ \mathtt{LAB}_\mathit{pc} $ &   \mathtt{LAB}_1  ~ \sqcup ~  \mathtt{LAB}_\mathit{pc}   \\
\ottkw{push}   & $ \mathtt{TRUE} $ & $ \mathtt{LAB}_\mathit{pc} $ &  \mathtt{BOT}  \\
\ottkw{load}   & $ \mathtt{TRUE} $ & $ \mathtt{LAB}_\mathit{pc} $ &   \mathtt{LAB}_1  ~ \sqcup ~  \mathtt{LAB}_2  \\
\ottkw{store}  & $ \mathtt{LAB}_1  \! \sqcup   \mathtt{LAB}_\mathit{pc}   \sqsubseteq  \mathtt{LAB}_3 $
           & $ \mathtt{LAB}_\mathit{pc} $ &
           \begin{trivarray}   \mathtt{LAB}_1  ~ \sqcup ~  \mathtt{LAB}_2    \sqcup ~  \mathtt{LAB}_\mathit{pc} 
           \end{trivarray} \\
\ottkw{jump}   & $ \mathtt{TRUE} $ & $  \mathtt{LAB}_1  ~ \sqcup ~  \mathtt{LAB}_\mathit{pc}  $ & \ottsym{\_\_}  \\
\ottkw{bnz}    & $ \mathtt{TRUE} $ & $  \mathtt{LAB}_1  ~ \sqcup ~  \mathtt{LAB}_\mathit{pc}  $ & \ottsym{\_\_}  \\
\ottkw{call}   & $ \mathtt{TRUE} $ & $  \mathtt{LAB}_1  ~ \sqcup ~  \mathtt{LAB}_\mathit{pc}  $ &  \mathtt{LAB}_\mathit{pc}  \\
\ottkw{ret}    & $ \mathtt{TRUE} $ & $ \mathtt{LAB}_1 $ & \ottsym{\_\_}
\end{array}
\]
\vspace*{-2ex}
}
\caption{Rule table $\rulesabsifc$ corresponding to abstract IFC
  machine}
\label{fig:rule_table}
\end{figure}

These symbolic expressions are written in a simple domain-specific
language (DSL) of operations over an IFC lattice.
The grammar of this DSL\iffull{} (\autoref{fig:rules_syn})\fi{}
includes label variables
$ \mathtt{LAB}_\mathit{pc} ,\ldots, \mathtt{LAB}_3 $, which correspond to the input labels
$L_{pc},\ldots,\ottnt{L_{{\mathrm{3}}}}$; the constant $ \mathtt{BOT} $; and the
lattice operators $ \sqcup $ (join) and $ \sqsubseteq $ (flows).
\iffull
\begin{figure}[tb!]
\ottgrammartabular{\ottLE\ottinterrule}
\ottgrammartabular{\ottBE\ottinterrule}
\caption{Symbolic IFC rule language syntax
}
\label{fig:rules_syn}
\end{figure}
\fi

\label{ruleeval}
The IFC enforcement judgment looks up the corresponding symbolic IFC
rule in the table and directly {\em evaluates} the symbolic
expressions in terms of the corresponding lattice operations.
\iffull\else The definition of this interpreter is completely
straightforward; we omit it for brevity.\fi{}
In contrast, in \autoref{sec:fault-handler} we {\em compile} this rule
table into the IFC fault handler for the concrete machine.
\iffull
Formally, the IFC enforcement judgment is defined by the two following
cases, depending on whether the second output label is relevant or not:
\[
\ottdruleevalXXrule{} \qquad \ottdruleevalXXruleXXNoRes{}
\]
Here $\rho$ is a 4-tuple of labels, {\it Rule}$_\mathcal{R}$ looks up
the relevant opcode in rule table $\mathcal{R}$, and the expression
evaluation judgment $\rho \vdash \dots$ is defined
in~\autoref{fig:rules_eval}.
\fi

\iffull
\begin{figure}[t]
\small
\[ \ottdruletrue \ottinterrule
   \quad \ottdruleflows\ottinterrule
   \quad \ottdruleand\ottinterrule \]
\[ \ottdruleorOne\ottinterrule
   \quad \ottdruleorTwo\ottinterrule
   \quad \ottdrulelbot \ottinterrule
   \quad \ottdrulejoin \ottinterrule  \]
\[  \ottdrulelpc \ottinterrule
   \quad \ottdrulevarTwo \ottinterrule \]
\[ \ottdrulevarOne \ottinterrule
   \quad \ottdrulevarThree \ottinterrule \]
\caption{Symbolic IFC rule language semantics%
}
\label{fig:rules_eval}
\end{figure}
\fi
\section{Concrete Machine}
\label{sec:conc}

\newcommand{\intercolspace}{\hspace{50pt}}

The concrete machine provides low-level support for efficiently
implementing many different high-level policies (IFC and others) with
a combination of a hardware rule cache and a software cache fault
handler.
In this section we focus on the concrete machine's hardware,
which is completely generic, while
in \autoref{sec:fault-handler} we describe a specific fault handler
corresponding to the IFC rules of the symbolic rule machine.

The concrete machine has the same general structure as the more
abstract ones, but differs in several important respects.  One is that
it annotates data values with integer {\em tags} $\mathtt{T}$, rather than
with {\em labels} $\ottnt{L}$ from an abstract lattice; thus, the {\em
  concrete atoms} $\concretefont{a}$ in the data memories and the stack have the
form $\ottnt{n}  \mathord{\scriptstyle @}  \mathtt{T}$.  Similarly, a {\em concrete action} $\alpha$ is
either a concrete atom or the silent action $ \tau $.  \iffull We
consistently use the word {\em label} and variable $\ottnt{L}$ to refer to
the (abstract, lattice-structured) labels of the abstract and symbolic
rule machines and the word {\em tag} and variable $\mathtt{T}$ for concrete
integers representing labels.\fi{}
Using plain integers as tags allows us to delegate their interpretation
entirely to software.
In this paper we focus solely on using tags to implement IFC labels,
although they could also be used for enforcing other policies, such as
type and memory safety or control-flow
integrity\ifpump~\cite{pump_asplos2015,micropolicies2015}\fi.
For instance, to implement the two-point abstract lattice with
$\public \mathrel{ \le } \secret$, we could use $0$ to represent
$\public$ and $1$ to represent $\secret$, making the operations
$ \mathord{\vee} $ and $ \le $ easy to implement
(see \autoref{sec:fault-handler}). For richer abstract lattices, a
more complex concrete representation might be needed; for example, a
label containing an arbitrary set of principals might be represented
concretely by a pointer to an array data structure
(see \autoref{sec:extensions}).  In places where a tag is needed but
its value is irrelevant, the concrete machine uses a specific but
arbitrary {\em default tag} value (e.g., -1), which we write
$ \mathtt{T}_\mathtt{D} $.

A second important difference is that the concrete machine has two
modes: {\em user mode} ($ \text{\sf u} $), for executing the ordinary user
program, and {\em kernel mode} ($ \text{\sf k} $), for handling rule cache
faults.
To support these two modes, the concrete machine's state contains a
\emph{privilege bit} $\pi$, a separate \emph{kernel instruction
  memory} $\phi$, and a separate \emph{kernel data memory} $\kappa$, in
addition to the user instruction memory $\iota$, the user data
memory $\concretesymbol{\mu}$, the stack $\concretesymbol{\sigma}$, and the PC.
When the machine is operating in user mode ($\pi =  \text{\sf u} $), instructions are
looked up using the PC as an index into $\iota$, and loads and stores use $\concretesymbol{\mu}$;
when in kernel mode ($\pi =  \text{\sf k} $), the PC is treated as an index into
$\phi$, and loads and stores use $\kappa$.
\iffull
The {\em concrete step relation} has the form
$ \iota,\phi \vdash \cStep{ \pi_{{\mathrm{1}}} }{ \kappa_{{\mathrm{1}}} }{ \concretesymbol{\mu}_{{\mathrm{1}}} }{ [ \concretesymbol{\sigma}_{{\mathrm{1}}} ] }{ \concretefont{pc}_{{\mathrm{1}}} }{ \alpha
}{ \pi_{{\mathrm{2}}} }{ \kappa_{{\mathrm{2}}} }{ \concretesymbol{\mu}_{{\mathrm{2}}} }{ [ \concretesymbol{\sigma}_{{\mathrm{2}}} ] }{ \concretefont{pc}_{{\mathrm{2}}} }{}$.
\fi
As before, since $\iota$ and $\phi$ are fixed, we normally leave them implicit\iffull{} when
writing down machine transitions\fi.

\iffull
\amd{I would appreciate a diagram for this---but maybe your
  intended audience are verbal rather than visual learners?}\bcp{There
  are lots of diagrams in my CSF talk.  Let's think about which one(s)
  to include, at least in the long version.}  \dd{the one page 43 in
  CSF.pdf seems useful, combined with the example at the end of the
  section. The one page 26 could be used earlier in the paper, to
  explain the big picture.}
\fi

The concrete machine has the same instruction set as the previous
ones, allowing user programs to be run on all three machines
unchanged.  But the tag-related semantics of instructions depends on
the privilege mode, and in user mode the semantics further depends on
the state of the {\em rule cache}.  In the real SAFE machine, the rule
cache may contain thousands of entries and is implemented as a
separate near-associative memory~\cite{near_assoc_cache_fpga_2013}
accessed by special instructions.  Here, for simplicity, we use a
cache with just one entry, located at the start of kernel memory, and
use $ \text{\sf Load} $ and $ \text{\sf Store} $ instructions to manipulate it. When
implementing simple IFC labels such as the two-point lattice defined
above, the rule cache is all that needs to live in $\kappa$. More
complex label models, on the other hand, such as those of
\autoref{sec:extensions}, may require additional memory to store
internal data structures.

The rule cache holds a single rule instance, represented graphically
like this:
\[  \begin{array}{|@{\;}l@{\;}|@{\;}l@{\;}|@{\;}l@{\;}|@{\;}l@{\;}|@{\;}l@{\;}||@{\;}l@{\;}|@{\;}l@{\;}|}
                       \hline
                           \ottnt{opcode}  &  \mathtt{T}_{pc}  &  \mathtt{T}_{{\mathrm{1}}}  &  \mathtt{T}_{{\mathrm{2}}}  &  \mathtt{T}_{{\mathrm{3}}}  &  \mathtt{T}_{rpc}  &  \mathtt{T}_{r}  \\
                       \hline
                       \end{array}  \]
Location 0 holds an integer representing an opcode.  \iffull (Since
the exact choice of representation doesn't matter, we will denote each
opcode with a lowercase identifier---for example, we might define
$\ottkw{add} = 0$, $\ottkw{output} = 1$, etc.)\fi{} Location 1 holds the PC
tag. Locations 2 to 4 hold the tags of any other arguments needed by
this particular opcode.  Location 5 holds the tag that should go on
the PC after this instruction executes, and location 6 holds the tag
for the instruction's result value, if needed.  For example, suppose
the cache contains\iffull{} this:
\[\else{} $\fi
 \begin{array}{|@{\;}l@{\;}|@{\;}l@{\;}|@{\;}l@{\;}|@{\;}l@{\;}|@{\;}l@{\;}||@{\;}l@{\;}|@{\;}l@{\;}|}
                       \hline
                           \ottkw{add}  &  \ottsym{0}  &  \ottsym{1}  &  \ottsym{1}  &   \text{-1}   &  \ottsym{0}  &  \ottsym{1}  \\
                       \hline
                       \end{array} 
\iffull\]\else.$ \fi
(Note that we are showing just the ``payload'' part of these seven atoms; by
convention, the tag part is always $ \mathtt{T}_\mathtt{D} $, and we do not display
it.)
\iffull
This one-line rule cache should be thought of as implementing a (very)
partial function: when the input is $ \begin{array}{|@{\;}l@{\;}|@{\;}l@{\;}|@{\;}l@{\;}|@{\;}l@{\;}|@{\;}l@{\;}|}
                       \hline
                           \ottkw{add}  &  \ottsym{0}  &  \ottsym{1}  &  \ottsym{1}  &   \text{-1}   \\
                       \hline
                       \end{array} $\ , the
output is $ \begin{array}{|@{\;}l@{\;}|@{\;}l@{\;}|}
                    \hline
                           \ottsym{0}  &  \ottsym{1}  \\
                    \hline
                       \end{array} $; otherwise it is undefined.
\fi
If $0$ is the tag representing the label $\public$, $1$ represents $\secret$,
and -1 is the default tag $ \mathtt{T}_\mathtt{D} $,
this can be interpreted abstractly as follows: ``If the next instruction is
$ \text{\sf Add} $, the PC is labeled $\public$, and the two relevant arguments are both
labeled $\secret$, then the instruction should be allowed,
the label on the new PC should be $\public$, and the
label on the result of the operation is $\secret$.''

There are two sets of stepping rules \iffull governing the behavior of
\else for \fi the concrete machine in user mode; which set applies
depends on whether the current machine state matches the current
contents of the rule cache.  In the ``cache hit'' case
\iffull (\autoref{fig:stepuserhit})\else{}\fi, the instruction
executes normally, with the cache's output determining the new PC tag
and result tag (if any).

\iffull
\begin{figure*}[t]
\centering
\[
\begin{array}{lr}
\ottdruleCAdd{} & \ottdruleCPsh{} \\[4em]
\ottdruleCLd{} & \ottdruleCSt{} \\[4em]
\ottdruleCJmp{} & \ottdruleCBnz{} \\[4em]
\ottdruleCCll{} & \ottdruleCRet{} \\[4em]
\multicolumn{2}{c}{\ottdruleCOut{}}
\end{array}
\]
\caption{Concrete step relation: user mode, cache hit case}
\label{fig:stepuserhit}
\end{figure*}
\begin{figure*}[t]
\centering
\[
\begin{array}{lr}
\ottdruleCAddXXF{} & \ottdruleCPshXXF{} \\[4em]
\ottdruleCLdXXF{} & \ottdruleCStXXF{} \\[4em]
\ottdruleCJmpXXF{} & \ottdruleCBnzXXF{} \\[4em]
\ottdruleCCllXXF{} & \ottdruleCRetXXF{} \\[4em]
\multicolumn{2}{c}{\ottdruleCOutXXF{}}
\end{array}
\]
\caption{Concrete step relation: user mode, cache miss case}
\label{fig:stepusermiss}
\end{figure*}
\fi

In the ``cache miss'' case \iffull (\autoref{fig:stepusermiss})\else{}\fi,
the relevant parts of the current state (opcode, PC tag, argument tags) are stored
into the input part of the single cache line and the machine simulates
a $ \text{\sf Call} $ to the fault handler.

To see how this works in more detail, consider the two
user-mode stepping rules for the $ \text{\sf Add} $ instruction.
\[\ottdruleCAdd{}\quad\ottdruleCAddXXF{}\]
In the first rule (cache hit), the side condition demands that the
input part of the current cache contents have the form $ \begin{array}{|@{\;}l@{\;}|@{\;}l@{\;}|@{\;}l@{\;}|@{\;}l@{\;}|@{\;}l@{\;}|}
                       \hline
                           \ottkw{add}  &  \mathtt{T}_{pc}  &  \mathtt{T}_{{\mathrm{1}}}  &  \mathtt{T}_{{\mathrm{2}}}  &  \mathtt{T}_\mathtt{D}  \\
                       \hline
                       \end{array} $, where $\mathtt{T}_{pc}$ is the tag on the current PC,
$\mathtt{T}_{{\mathrm{1}}}$ and $\mathtt{T}_{{\mathrm{2}}}$ are the tags on the top two atoms on the stack,
and the fourth element is the default tag.  In this case, the output
part of the rule, $ \begin{array}{|@{\;}l@{\;}|@{\;}l@{\;}|}
                    \hline
                           \mathtt{T}_{rpc}  &  \mathtt{T}_{r}  \\
                    \hline
                       \end{array} $, determines the tag $\mathtt{T}_{rpc}$ on
the PC and the tag $\mathtt{T}_{r}$ on the new atom pushed onto the stack in
the next machine state.

In the second rule (cache miss), the notation $ [  {\kappa_i}  ,  \kappa_o  ] $ means
``let ${\kappa_i}$ be the input part of the current rule cache and
$\kappa_o$ be the output part.'' The side condition says that the
current input part ${\kappa_i}$ does {\em} not have the desired form $ \begin{array}{|@{\;}l@{\;}|@{\;}l@{\;}|@{\;}l@{\;}|@{\;}l@{\;}|@{\;}l@{\;}|}
                       \hline
                           \ottkw{add}  &  \mathtt{T}_{pc}  &  \mathtt{T}_{{\mathrm{1}}}  &  \mathtt{T}_{{\mathrm{2}}}  &  \mathtt{T}_\mathtt{D}  \\
                       \hline
                       \end{array} $, so the machine needs to enter the fault
handler.  The next machine state is formed as follows:
\begin{inparaenum}[(i)]
\item the input part of the cache
is set to the desired form ${\kappa_j}$ and the output part is set to
$\kappa_\mathtt{D} \triangleq  \begin{array}{|@{\;}l@{\;}|@{\;}l@{\;}|}
                   \hline
                           \mathtt{T}_\mathtt{D}  &  \mathtt{T}_\mathtt{D}  \\
                   \hline
                    \end{array} $;
\item
a new return frame is pushed on top of the stack to remember the current PC and
privilege bit ($ \text{\sf u} $);
\item
the privilege bit is set to $ \text{\sf k} $ (which will cause
the next instruction to be read from the kernel instruction memory);
and \item
the PC is set to $0$, the location in the kernel
instruction memory where the fault handler routine begins.
\end{inparaenum}

What happens next is up to the fault handler code. Its job is to
examine the contents of the first five kernel memory locations and
either \begin{inparaenum}[(i)] \item write appropriate tags for the
result and new PC into the sixth and seventh kernel memory locations
and then perform a $ \text{\sf Ret} $ to go back to user mode and restart the
faulting instruction, or
\item stop the machine by jumping to an invalid PC (-1) to signal that
  the attempted combination of opcode and argument tags is illegal.%
\iffull
  \footnote{As explained in \autoref{sec:fullSAFE},
    in this work we assume for simplicity that policy violations are
    fatal. Recent work~\cite{Exceptional} has shown that it is possible to
    recover from IFC violations while preserving noninterference.}
\fi
\end{inparaenum}{}
This mechanism is general and can be used to implement many different
high-level policies (IFC and others).

In kernel mode\iffull{} (\autoref{fig:steppriv})\fi{}, the treatment
of tags is almost completely degenerate: to avoid infinite regress,
the concrete machine does not consult the rule cache while in kernel
mode.  For most instructions, tags read from the current machine state
are ignored (indicated by~$ \dummytag $) and tags written to the new state
are set to $ \mathtt{T}_\mathtt{D} $. This can be seen for instance in the
kernel-mode step rule for addition
{\renewcommand{\backslash}{}\ottusedrule{\ottdruleCAddXXP{}}}
The only significant exception\iffull s\fi{} to this pattern\iffull{} are $ \text{\sf Load} $ and $ \text{\sf Store} $, which preserve
the tag of the datum being read from or written to memory, and \else{} is \fi $ \text{\sf Ret} $, which takes both
the privilege bit and the new PC (including its tag!) from the return frame
at the top of the stack.  This is critical, since a $ \text{\sf Ret} $ instruction is
used to return from kernel to user mode when the fault handler has finished
executing.
\ottusedrule{\ottdruleCRetXXP{}}

\iffull
\begin{figure}[tbp!]
\centering
\[
\begin{array}{lr}
\ottdruleCAddXXP{} & \ottdruleCPshXXP{} \\[4em]
\ottdruleCLdXXP{} & \ottdruleCStXXP{} \\[4em]
\ottdruleCJmpXXP{} & \ottdruleCBnzXXP{} \\[4em]
\ottdruleCCllXXP{} & \ottdruleCRetXXP{}
\end{array}
\]
\caption{Concrete step relation (kernel mode)
\bcp{When the Coq has completely stabilized, we need to check once
  more, very carefully, that these rules exactly match the Coq (and
  the rules in the other step functions).  This is someplace that a
  bug could very easily creep in.} }
\label{fig:steppriv}
\end{figure}
\fi

A final point is that $ \text{\sf Output} $ is not permitted in kernel mode, which
guarantees that kernel actions are always the silent action $\tau$.
\ifmuchlater\apt{Worth saying here rather than introducing (parenthetically) in the proof
of Lemma 9.6 ?}\fi

\bcp{I think this section is the right place to discuss the fact that the
  Output instruction sends out ``bare pairs of tags'' (which may even
  include a full memory, in the extension), but that these should not be
  regarded as going straight to the outside world.  But since the Output
  instruction isn't really discussed in the first place, I'm not sure where
  this discussion actually belongs.  Although it is mentioned that the
  output traces from this machine are concrete atoms, the consequences of
  this fact are not explored.  Maybe we should take that sentence from the
  second paragraph and turn it into its own paragraph?}

\iffull As an illustration of how all this works, suppose again that $ \iota = [...,
 \text{\sf Add} , ...]$, and that the concrete integer tag $0$ represents the
abstract label $ \bot $, $1$ represents $ \top $, and -1 is
$ \mathtt{T}_\mathtt{D} $.  Then, in a cache-hit configuration, we have (omitting the
silent $\tau$ label on transitions):
\[
\begin{array}{l@{\;\;}l@{\;\;}l@{\;\;}l@{\;\;}l@{\;\;}l}
 \text{\sf u}  &  \begin{array}{|@{\;}l@{\;}|@{\;}l@{\;}|@{\;}l@{\;}|@{\;}l@{\;}|@{\;}l@{\;}||@{\;}l@{\;}|@{\;}l@{\;}|}
                       \hline
                           \ottkw{add}  &  \ottsym{0}  &  \ottsym{0}  &  \ottsym{1}  &   \text{-1}   &  \ottsym{0}  &  \ottsym{1}  \\
                       \hline
                       \end{array} 
          & \mu & [ \ottsym{7}  \mathord{\scriptstyle @}  \ottsym{0}  \mathord{,}\,  \ottsym{5}  \mathord{\scriptstyle @}  \ottsym{1} ] & \ottnt{n}  \mathord{\scriptstyle @}  \ottsym{0}
   & \longrightarrow \\[.8ex]
 \text{\sf u}  &  \begin{array}{|@{\;}l@{\;}|@{\;}l@{\;}|@{\;}l@{\;}|@{\;}l@{\;}|@{\;}l@{\;}||@{\;}l@{\;}|@{\;}l@{\;}|}
                       \hline
                           \ottkw{add}  &  \ottsym{0}  &  \ottsym{0}  &  \ottsym{1}  &   \text{-1}   &  \ottsym{0}  &  \ottsym{1}  \\
                       \hline
                       \end{array} 
          & \mu & [ \ottsym{12}  \mathord{\scriptstyle @}  \ottsym{1} ] & \ottsym{(}  \ottnt{n}  \mathord{+}  \ottsym{1}  \ottsym{)}  \mathord{\scriptstyle @}  \ottsym{0}
\end{array}
\]
On the other hand, if the tags on both operands are $1$ (i.e., $ \top $),
then the first step will miss in the cache and reduction will proceed as follows:
\[
\begin{array}{@{}l@{\;\;}l@{\;\;}l@{\;\;}l@{\;\;}l@{\;\;}ll}
 \text{\sf u}  &  \begin{array}{|@{\;}l@{\;}|@{\;}l@{\;}|@{\;}l@{\;}|@{\;}l@{\;}|@{\;}l@{\;}||@{\;}l@{\;}|@{\;}l@{\;}|}
                       \hline
                           \ottkw{add}  &  \ottsym{0}  &  \ottsym{0}  &  \ottsym{1}  &   \text{-1}   &  \ottsym{0}  &  \ottsym{1}  \\
                       \hline
                       \end{array} 
          & \mu & [ \ottsym{7}  \mathord{\scriptstyle @}  \ottsym{1}  \mathord{,}\,  \ottsym{5}  \mathord{\scriptstyle @}  \ottsym{1} ] & \ottnt{n}  \mathord{\scriptstyle @}  \ottsym{0}
   & \longrightarrow & \mbox{\em (cache miss)} \\[.8ex]
 \text{\sf k}  &  \begin{array}{|@{\;}l@{\;}|@{\;}l@{\;}|@{\;}l@{\;}|@{\;}l@{\;}|@{\;}l@{\;}||@{\;}l@{\;}|@{\;}l@{\;}|}
                       \hline
                           \ottkw{add}  &  \ottsym{0}  &  \ottsym{1}  &  \ottsym{1}  &   \text{-1}   &   \text{-1}   &   \text{-1}   \\
                       \hline
                       \end{array} 
          & \mu & [ \ottsym{(}  \ottnt{n}  \mathord{\scriptstyle @}  \ottsym{0}  \mathord{,}\,  \text{\sf u}  \ottsym{)}  \mathord{;}\,  \ottsym{7}  \mathord{\scriptstyle @}  \ottsym{1}  \mathord{,}\,  \ottsym{5}  \mathord{\scriptstyle @}  \ottsym{1} ] & \ottsym{0}  \mathord{\scriptstyle @}   \text{-1} 
   & \longrightarrow & \mbox{\em (call fault handler, kernel mode)} \\[.8ex]
\multicolumn{4}{l}{\mbox{\em \quad ... fault handler runs ...}} \\[.8ex]
 \text{\sf k}  &  \begin{array}{|@{\;}l@{\;}|@{\;}l@{\;}|@{\;}l@{\;}|@{\;}l@{\;}|@{\;}l@{\;}||@{\;}l@{\;}|@{\;}l@{\;}|}
                       \hline
                           \ottkw{add}  &  \ottsym{0}  &  \ottsym{1}  &  \ottsym{1}  &   \text{-1}   &  \ottsym{0}  &  \ottsym{1}  \\
                       \hline
                       \end{array} 
          & \mu & [ \ottsym{(}  \ottnt{n}  \mathord{\scriptstyle @}  \ottsym{0}  \mathord{,}\,  \text{\sf u}  \ottsym{)}  \mathord{;}\,  \ottsym{7}  \mathord{\scriptstyle @}  \ottsym{1}  \mathord{,}\,  \ottsym{5}  \mathord{\scriptstyle @}  \ottsym{1} ] & \ottnt{k}  \mathord{\scriptstyle @}   \text{-1} 
   & \longrightarrow & \mbox{\em (fault handler returns to user mode)} \\[.8ex]
 \text{\sf u}  &  \begin{array}{|@{\;}l@{\;}|@{\;}l@{\;}|@{\;}l@{\;}|@{\;}l@{\;}|@{\;}l@{\;}||@{\;}l@{\;}|@{\;}l@{\;}|}
                       \hline
                           \ottkw{add}  &  \ottsym{0}  &  \ottsym{1}  &  \ottsym{1}  &   \text{-1}   &  \ottsym{0}  &  \ottsym{1}  \\
                       \hline
                       \end{array} 
          & \mu & [ \ottsym{7}  \mathord{\scriptstyle @}  \ottsym{1}  \mathord{,}\,  \ottsym{5}  \mathord{\scriptstyle @}  \ottsym{1} ] & \ottnt{n}  \mathord{\scriptstyle @}  \ottsym{0}
   & \longrightarrow & \mbox{\em (restarts instruction, cache now hits)} \\[.8ex]
 \text{\sf u}  &  \begin{array}{|@{\;}l@{\;}|@{\;}l@{\;}|@{\;}l@{\;}|@{\;}l@{\;}|@{\;}l@{\;}||@{\;}l@{\;}|@{\;}l@{\;}|}
                       \hline
                           \ottkw{add}  &  \ottsym{0}  &  \ottsym{0}  &  \ottsym{1}  &   \text{-1}   &  \ottsym{0}  &  \ottsym{1}  \\
                       \hline
                       \end{array} 
          & \mu & [ \ottsym{12}  \mathord{\scriptstyle @}  \ottsym{1} ] & \ottsym{(}  \ottnt{n}  \mathord{+}  \ottsym{1}  \ottsym{)}  \mathord{\scriptstyle @}  \ottsym{0}\!\!\!\!\!\!
\end{array}
\]
\fi

\section{Fault Handler for IFC}
\label{sec:fault-handler}

Now we assemble the pieces.
A concrete IFC machine implementing the symbolic rule machine defined
in
\autoref{sec:quasi} can be obtained by installing appropriate fault
handler code in the kernel instruction memory of the concrete
machine presented in \autoref{sec:conc}.  In essence, this handler must
emulate how the symbolic rule machine looks up
and evaluates the DSL expressions in a given IFC rule table.
We choose to generate the handler code by compiling the lookup and DSL
evaluation relations directly into machine code.  (An alternative would be
to represent the rule table as abstract syntax in the kernel memory and write an
interpreter in machine code for the DSL, but the compilation
approach seems to lead to simpler code and proofs.)

The handler compilation scheme is given\iffull\else{} (in part)\fi{}
in~\autoref{fig:rule_comp}.  Each $\ottkw{gen*}$ function generates a
list of concrete machine instructions; the sequence generated by the
top-level $\ottkw{genFaultHandler}$ is intended to be installed
starting at location 0 in the concrete machine's kernel instruction
memory.  The implicit $\ottkw{addr*}$ parameters are symbolic names
for the locations of the opcode and various tags in the concrete
machine's rule cache, as described in \autoref{sec:conc}.  The entire
generator is parameterized by an arbitrary rule table $\mathcal{R}$.  We
make heavy use of the (obvious) encoding of booleans where false is
represented by 0 and true by any non-zero value.  \iffull\else We omit
the straightforward definitions of some of the leaf
generators.\ifmuchlater\bcp{Double-checking: they are not omitted in
  the full version, right?}\fi\fi

The top-level handler works in three phases.  The first phase,
$\ottkw{genComputeResults}$, does most of the work: it consists of a
large nested if-then-else chain, built using
$\ottkw{genIndexedCases}$, that compares the opcode of the faulting
instruction against each possible opcode and, on a match, executes the
code generated for the corresponding symbolic IFC rule.  The code
generated for each symbolic IFC rule (by $\ottkw{genApplyRule}$)
pushes its results onto the stack: a flag indicating whether the
instruction is allowed and, if so, the result-PC and result-value
tags.  This first phase never writes to memory or transfers control
outside the handler; this makes it fairly easy to prove correct.

The second phase\iffull{} of the top-level handler\fi,
$\ottkw{genStoreResults}$, reads the computed results off the stack
and updates the rule cache appropriately.  If the result indicates
that the instruction is allowed, the result PC and value tags are
written to the cache, and true is pushed on the stack; otherwise,
nothing is written to the cache, and false is pushed on the stack.

The third and final phase of the top-level handler tests the boolean just
pushed onto the stack and either returns to user code (instruction is
allowed) or jumps to address -1 (disallowed).

The code for symbolic rule compilation is built by straightforward
recursive traversal of the rule DSL syntax for label-valued
expressions ($\ottkw{genELab}$) and boolean-valued expressions
($\ottkw{genBool}$).  These functions are (implicitly) parameterized
by \iffull the definitions of \fi lattice-specific generators
$\ottkw{genBot}$, $\ottkw{genJoin}$, and $\ottkw{genFlows}$.  To
implement these generators for a particular lattice, we first need to
choose how to represent abstract labels as integer tags, and then
determine a sequence of instructions that encodes each operation.  We
call such an encoding scheme a \emph{concrete lattice}.  For example,
the abstract labels in the two-point lattice can be encoded like
booleans, representing $\public$ by 0, $\secret$ by non-0, and
instantiating $\ottkw{genBot}$, $\ottkw{genJoin}$, and
$\ottkw{genFlows}$ with code for computing false, disjunction, and
implication, respectively.  A simple concrete lattice like this can be
formalized as a tuple $\ottnt{CL} = \left(  \mathsf{Tag} ,
   \mathsf{Lab} ,\ottkw{genBot}, \ottkw{genJoin}, \ottkw{genFlows}
\right)$,
where the encoding and decoding functions $ \mathsf{Lab} $ and
$ \mathsf{Tag} $ satisfy $ \mathsf{Lab} \circ  \mathsf{Tag}  = id$; to
streamline the exposition, we assume this form of concrete lattice for
most of the paper.  The more realistic encoding in
\autoref{sec:extensions} will require a more complex treatment.

To raise the level of abstraction of the handler code, we make heavy
use of structured code generators; this makes it easier both to
understand the code and to prove it correct using a custom Hoare logic
that follows the structure of the generators (see
\autoref{para:fault-handler-correct}).  For example, the
$\ottkw{genIf}$ function takes two code sequences, representing the
``then'' and ``else'' branches of a conditional, and generates code to
test the top of the stack and dispatch control appropriately.  The
higher-order generator $\ottkw{genIndexedCases}$ takes a list of
integer indices (e.g., opcodes) and functions for generating guards
and branch bodies from an index, and generates code that will run the
guards in order until one of them computes true, at which point the
corresponding branch body is run.

\begin{figure}[tbp]%
\centering
\hspace*{-6pt}\begin{minipage}{0.6\linewidth}  
\setlength{\parindent}{0pt}

{\small
$\begin{array}[t]{l@{\ }l}
\ottkw{genFaultHandler}\ \mathcal{R}  =
  & \ottkw{genComputeResults}\ \mathcal{R} \  {\scriptstyle \mathord{+\!+} } \ \\
  & \ottkw{genStoreResults}\  {\scriptstyle \mathord{+\!+} } \\
  & \ottkw{genIf}\ [  \text{\sf Ret}  ]\ [  \text{\sf Push} \ (\mbox{-1});\  \text{\sf Jump}  ]
\end{array}$

\vspace{0.1in}

$\begin{array}{l@{\ }l}
\ottkw{genComputeResults}\ \mathcal{R} = \\
  \hphantom{x}\ottkw{genIndexedCases}\ []\ \ottkw{genMatchOp}\ (\ottkw{genApplyRule} \circ \mbox{\it Rule}_\mathcal{R})\ \ottkw{opcodes}\\[0.1in]
\end{array}$

$\begin{array}{l@{\ }l@{\ }ll}
\ottkw{genMatchOp}\ op \  = & \\
 \hphantom{x} [  \text{\sf Push} \ op ]\  {\scriptstyle \mathord{+\!+} } \ \ottkw{genLoadFrom}\ \ottkw{addrOpLabel}\  {\scriptstyle \mathord{+\!+} } \ \ottkw{genEqual}\\
\iffull\ottkw{genEqual} = [ \ottkw{Sub} ]\  {\scriptstyle \mathord{+\!+} } \ \ottkw{genNot}\\\fi
\end{array}
$

\vspace{0.1in}

$\begin{array}{l@{\ }l@{\ }l@{\ }l}
\ottkw{genApplyRule} &  \langle  \ottnt{allow} ,  e_{rpc} ,  e_{r}  \rangle  & = & \ottkw{genBool}\ \ottnt{allow}\  {\scriptstyle \mathord{+\!+} } \\
   \multicolumn{4}{l}{\hphantom{x}
     \ottkw{genIf}\ (\ottkw{genSome}\ (\ottkw{genELab}\ e_{rpc}\  {\scriptstyle \mathord{+\!+} } \ \ottkw{genELab}\ e_{r}))\ \ottkw{genNone}} \\
\end{array}
$

\vspace{0.1in}

$\begin{array}{ll@{\ \ }l@{\ \ }l}
\ottkw{genELab} &  \mathtt{BOT}  & = & \ottkw{genBot} \\
              & \mathtt{LAB}_i  & = & \ottkw{genLoadFrom}\ \ottkw{addrTag}_i \\
              &  \ottnt{LE_{{\mathrm{1}}}} ~ \sqcup ~ \ottnt{LE_{{\mathrm{2}}}}  & = & \ottkw{genELab}\  \ottnt{LE_{{\mathrm{2}}}}\  {\scriptstyle \mathord{+\!+} } \ \ottkw{genELab}\  \ottnt{LE_{{\mathrm{1}}}}\  {\scriptstyle \mathord{+\!+} } \ \ottkw{genJoin}\\
\end{array}$

\vspace{0.1in}

$\begin{array}{ll@{\ \ }l@{\ \ }l}
\ottkw{genBool} &  \mathtt{TRUE}  & = & \ottkw{genTrue} \\
              &  \ottnt{LE_{{\mathrm{1}}}} \sqsubseteq \ottnt{LE_{{\mathrm{2}}}}  & = & \ottkw{genELab}\  \ottnt{LE_{{\mathrm{2}}}}\  {\scriptstyle \mathord{+\!+} } \ \ottkw{genELab}\  \ottnt{LE_{{\mathrm{1}}}}\  {\scriptstyle \mathord{+\!+} } \ \ottkw{genFlows}  \\
\end{array}$
\vspace{0.1in}

$\begin{array}{l@{\ }l}
\ottkw{genStoreResults} \ = \\
  \hphantom{x}\ottkw{genIf}\begin{array}[t]{@{\ }l}(\ottkw{genStoreAt}\ \ottkw{addrTag}_r\  {\scriptstyle \mathord{+\!+} } \ \ottkw{genStoreAt}\ \ottkw{addrTag}_{rpc}\  {\scriptstyle \mathord{+\!+} } \ \ottkw{genTrue}) \\
                                               \ottkw{genFalse}\end{array}
\end{array}$

\vspace{0.1in}

\iffull
$\begin{array}{lll}
\ottkw{genFalse} & = & [ \text{\sf Push} \, \ottsym{0} ]\\
\ottkw{genTrue} & = & [ \text{\sf Push} \, \ottsym{1} ]\\
\ottkw{genAnd} & = & \ottkw{genIf}\ [\,]\ (\ottkw{genPop}\  {\scriptstyle \mathord{+\!+} } \ \ottkw{genFalse}) \\
\ottkw{genOr} & = & \ottkw{genIf}\ (\ottkw{genPop}\  {\scriptstyle \mathord{+\!+} } \ \ottkw{genTrue})\ [\,]\\
\ottkw{genNot} & = & \ottkw{genIf}\ {\ottkw{genFalse}}\ {\ottkw{genTrue}} \\
\ottkw{genImpl} & = & \ottkw{genNot}\  {\scriptstyle \mathord{+\!+} } \ \ottkw{genOr} \\
\ottkw{genSome}\ c & = & c\  {\scriptstyle \mathord{+\!+} } \ \ottkw{genTrue}\\
\ottkw{genNone} & = & \ottkw{genFalse}\\
\end{array}
$

\vspace{0.1in}
\fi

$
\begin{array}{ll@{\ \ }l@{\ \ }l}
\multicolumn{4}{l}{\ottkw{genIndexedCases}\ \ottnt{genDefault}\ \ottnt{genGuard}\ \ottnt{genBody} =  g} \\
\ \ \ \ \ \mbox{\rm where} & g\ [] & = & \ottnt{genDefault}\\
        & g\ (n::ns) & = & \ottnt{genGuard}\ n\  {\scriptstyle \mathord{+\!+} } \ \ottkw{genIf}\ (\ottnt{genBody}\ n)\ (g\ ns)\\
\end{array}
$

\vspace{0.1in}

$
\begin{array}{lll}
\ottkw{genIf}\ t\ f & = & \ottkw{genSkipIf}\ (\ottkw{length}\ f')\  {\scriptstyle \mathord{+\!+} } \ f'\  {\scriptstyle \mathord{+\!+} } \ t \\
& & \ \ \mbox{\rm where}\ f' = f\  {\scriptstyle \mathord{+\!+} } \ \ottkw{genSkip}(\ottkw{length}\ t)\\
\ottkw{genSkip}\ n & = & \ottkw{genTrue}\  {\scriptstyle \mathord{+\!+} } \ \ottkw{genSkipIf}\ n\\
\ottkw{genSkipIf}\ n & = & [ \text{\sf Bnz} \, \ottsym{(}  \ottnt{n}  \mathord{+}  \ottsym{1}  \ottsym{)} ]\\
\end{array}
$

\iffull
\vspace{0.1in}
$
\begin{array}{lll}
\ottkw{genStoreAt}\ p & = & [ \text{\sf Push} \, \ottnt{p}; \text{\sf Store}  ]\\
\ottkw{genLoadFrom}\ p & = & [ \text{\sf Push} \, \ottnt{p}; \text{\sf Load}  ]\\
\ottkw{genPop} & = & [ \text{\sf Bnz} \, \ottsym{1} ]\\
\end{array}$
\fi

\vspace{0.1in}

$\begin{array}{l@{\ \ }l@{\ \ }l}
\ottkw{opcodes} & = &  
 [ \ottkw{add} ; \ottkw{output} ; \ldots ; \ottkw{ret} ]
\end{array}$
}
\end{minipage}
\caption{Generation of fault handler from IFC rule table.}
\label{fig:rule_comp}
\end{figure}

\section{Correctness of the Fault Handler Generator}
\label{sec:fault-handler-correct}

We now turn our attention to verification, beginning with the fault
handler.  We must show that the generated fault handler emulates the
IFC enforcement judgment $ \vdash_{ \mathcal{R} } \ruleeval{ \ottsym{(}  L_{pc}  \mathord{,}\,  \ell_{{\mathrm{1}}}  \mathord{,}\,  \ell_{{\mathrm{2}}}  \mathord{,}\,  \ell_{{\mathrm{3}}}  \ottsym{)} }{ \ottnt{opcode} }{ L_{rpc} }{ L_r } $ of the symbolic rule machine.
The statement and proof of correctness are parametric over the
symbolic IFC rule table $\mathcal{R}$ and concrete lattice, and hence
over correctness lemmas for the lattice operations.

\paragraph*{Correctness statement}
\label{para:fault-handler-correct}

Let $\mathcal{R}$ be an arbitrary rule table and $\phi_{\mathcal{R}}
\triangleq \ottkw{genFaultHandler}~\mathcal{R}$
be the corresponding generated fault handler.  We specify how
$\phi_{\mathcal{R}}$ behaves {as a whole}---as a
relation between initial state on entry and final state on completion---%
using the relation $\phi \vdash cs_1 \privsteps
cs_2$,
  defined as the reflexive transitive closure of the
  concrete step relation, with the constraints that the
  fault handler code is $\phi$ and all intermediate states
  (\IE strictly preceding $cs_2$) have privilege bit $ \text{\sf k} $.

The correctness statement is captured by the following two lemmas.
Intuitively, if the symbolic IFC enforcement judgment allows some given user
instruction, then executing $\phi_{\mathcal{R}}$ (stored at kernel mode
location 0) updates the cache to contain the tag encoding of the appropriate
result labels and returns to user-mode; otherwise, $\phi_{\mathcal{R}}$ halts
the machine ($\ottnt{pc}$ = -1).

\begin{lemma}[Fault handler correctness, allowed case]~\\
Suppose that $ \vdash_{ \mathcal{R} } \ruleeval{ \ottsym{(}  L_{pc}  \mathord{,}\,  \ell_{{\mathrm{1}}}  \mathord{,}\,  \ell_{{\mathrm{2}}}  \mathord{,}\,  \ell_{{\mathrm{3}}}  \ottsym{)} }{ \ottnt{opcode} }{ L_{rpc} }{ L_r } $ and
$${\kappa_i} =  \begin{array}{|@{\;}l@{\;}|@{\;}l@{\;}|@{\;}l@{\;}|@{\;}l@{\;}|@{\;}l@{\;}|}
                       \hline
                           \ottnt{opcode}  &   \mathsf{Tag} ( L_{pc} )   &   \mathsf{Tag} ( \ell_{{\mathrm{1}}} )   &   \mathsf{Tag} ( \ell_{{\mathrm{2}}} )   &   \mathsf{Tag} ( \ell_{{\mathrm{3}}} )   \\
                       \hline
                       \end{array} \;.$$
Then $$\phi_{\mathcal{R}}\vdash \runsToEscape{\langle  \text{\sf k}  }{  [  {\kappa_i}  ,  \kappa_o  ]  }{
\concretesymbol{\mu} }{ [ \ottsym{(}  \concretefont{pc}  \mathord{,}\,  \text{\sf u}  \ottsym{)}  \mathord{;}\,  \concretesymbol{\sigma} ] }{ \ottsym{0}  \mathord{\scriptstyle @}  \mathtt{T}_\mathtt{D}\rangle }{\langle  \text{\sf u}  }{
 [  {\kappa_i}  ,  \kappa_o'  ]  }{ \concretesymbol{\mu} }{ [ \concretesymbol{\sigma} ] }{ \concretefont{pc} \rangle}{\ottnt{NEWLINE}}$$
with output cache $\kappa_o' = \ottsym{(}  \mathsf{Tag} \, \ottsym{(}  L_{rpc}  \ottsym{)}  \mathord{,}\,  \mathsf{Tag} \, \ottsym{(}  L_r  \ottsym{)}  \ottsym{)}$\,.
\label{lem:fhdl-succ}
\end{lemma}

\begin{lemma}[Fault handler correctness, disallowed case]
Suppose that $ \vdash_\mathcal{R} ( L_{pc} , \ell_{{\mathrm{1}}} , \ell_{{\mathrm{2}}} ,
\ell_{{\mathrm{3}}}) \not\leadsto_{opcode}$, and
$${\kappa_i} =  \begin{array}{|@{\;}l@{\;}|@{\;}l@{\;}|@{\;}l@{\;}|@{\;}l@{\;}|@{\;}l@{\;}|}
                       \hline
                           \ottnt{opcode}  &   \mathsf{Tag} ( L_{pc} )   &   \mathsf{Tag} ( \ell_{{\mathrm{1}}} )   &   \mathsf{Tag} ( \ell_{{\mathrm{2}}} )   &   \mathsf{Tag} ( \ell_{{\mathrm{3}}} )   \\
                       \hline
                       \end{array} \;.$$
Then, for some final stack $\concretesymbol{\sigma}'$,
$$\phi_{\mathcal{R}}\vdash \runsToEscape{ \langle  \text{\sf k}  }{  [  {\kappa_i}  ,  \kappa_o  ]  }{ \concretesymbol{\mu} }
{ [ \ottsym{(}  \concretefont{pc}  \mathord{,}\,  \text{\sf u}  \ottsym{)}  \mathord{;}\,  \concretesymbol{\sigma} ] }{ \ottsym{0}  \mathord{\scriptstyle @}  \mathtt{T}_\mathtt{D}\rangle }{ \langle  \text{\sf k}  }
{  [  {\kappa_i}  ,  \kappa_o  ]  }{ \concretesymbol{\mu} }{ [ \concretesymbol{\sigma}' ] }{ \mbox{-}\ottsym{1}  \mathord{\scriptstyle @}  \mathtt{T}_\mathtt{D}\rangle .}{NEWLINE}$$
\label{lem:fhdl-fail}
\end{lemma}

\paragraph*{Proof methodology}
\label{para:hoare-logic}
The fault handler is compiled by composing generators
(\autoref{fig:rule_comp}); accordingly, the proofs of these two lemmas
reduce to correctness proofs for the generators.
We employ a
custom Hoare logic for specifying the generators themselves, which
makes the code generation proof simple, reusable, and scalable.
This is where defining a DSL for IFC rules and a \emph{structured}
compiler proves to be very useful approach, e.g., compared to
symbolic interpretation of hand-written code.

Our logic comprises two kinds of Hoare triples.
The generated code mostly consists of self-contained instruction
sequences that terminate by ``falling off the end''---\IE that never
return or jump outside themselves, although they may contain internal
jumps (e.g., to implement conditionals). The only exception is the
final step of the handler (third line of $\ottkw{genFaultHandler}$
in \autoref{fig:rule_comp}).
We therefore define a standard Hoare triple $ \{  \ottnt{P} \} \  \ottnt{c}  \ \{ \ottnt{Q} \} $, suitable
for reasoning about self-contained code, and use it for the bulk of
the proof.
To specify the final handler step, we define a non-standard triple
$ \{  \ottnt{P}  \} \  \ottnt{c}   \ \{  \ottnt{Q}  \}^{ O }_{ \concretefont{pc} } $
for reasoning about escaping code.

\paragraph{Self-contained-code Hoare triples}
The triple $ \{  \ottnt{P} \} \  \ottnt{c}  \ \{ \ottnt{Q} \} $, where $\ottnt{P}$ and $\ottnt{Q}$ are predicates on $\kappa
\times \concretesymbol{\sigma}$, says that, if the kernel instruction memory $\phi$ contains
the code sequence $\ottnt{c}$ starting at the current PC,
and if the current memory and stack satisfy $\ottnt{P}$, then the machine will
run (in kernel mode) until the PC points to the instruction immediately
following the sequence $\ottnt{c}$, with a resulting memory and stack satisfying
$\ottnt{Q}$.
\iffull
In symbols:
\[
\begin{array}{ll}
 \{  \ottnt{P} \} \  \ottnt{c}  \ \{ \ottnt{Q} \}  \triangleq &
        \ottnt{c} = \phi(n),\ldots,\phi(n'-1) \land \ottnt{P}  \ottsym{(}  \kappa  \mathord{,}\,  \concretesymbol{\sigma}  \ottsym{)}    \Longrightarrow  \\
  & \quad \exists~\kappa'~\concretesymbol{\sigma}'.\;
     \ottnt{Q}  \ottsym{(}  \kappa'  \mathord{,}\,  \concretesymbol{\sigma}'  \ottsym{)}
     \land\;   \phi \vdash \runsToEnd
      { \langle  \text{\sf k}  }{ \kappa  }{ \concretesymbol{\mu} }{ [ \concretesymbol{\sigma} ]  }{ \ottnt{n}  \mathord{\scriptstyle @}  \mathtt{T}_\mathtt{D} \rangle }
      {\langle  \text{\sf k}  }{ \kappa' }{ \concretesymbol{\mu} }{ [ \concretesymbol{\sigma}' ] }{ \ottnt{n'}  \mathord{\scriptstyle @}  \mathtt{T}_\mathtt{D} \rangle }
      {}
\end{array}
\]
\fi%
Note that the instruction memory $\phi$ is unconstrained outside of
$c$, so if $c$ is not self-contained, no triple about it will be
provable; thus, these triples obey the usual composition laws (\EG the
rule of consequence).
\iffull
\[ \ottdruleUnit \ottinterrule
   \qquad\qquad \ottdruleWeaken \ottinterrule
   \qquad\qquad \ottdruleCompose \ottinterrule\]
\fi
Also, because the concrete machine is
deterministic, these triples express total, rather than partial,
correctness, which is essential for proving termination in
\iffull\autoref{lem:fhdl-succ} and \autoref{lem:fhdl-fail}\else \ref{lem:fhdl-succ} and \ref{lem:fhdl-fail}\fi.  To aid automation of
proofs about code sequences, we give triples in weakest-precondition
style.

We build proofs by composing atomic specifications of individual
instructions, such as \ottusedrule{\ottdruleAdd{}~,} with
specifications for structured code generators, such as
\ottusedrule{\ottdruleIf{}~.}  (We emphasize that all such
specifications are \emph{verified}, not \emph{axiomatized} as the
inference rule notation might suggest.)
\iffull
We also prove a specification for the specialized case statement
$\ottkw{genIndexedCases}$. Although this specification is quite
complex when written in full detail (and thus omitted here), it
is intuitively simple: given a list of indices and functions for generating guards
and branches from the indices, $\ottkw{genIndexedCases}$ will run the
guards in order until one of them computes $\mathit{true}$ (more
precisely, its integer encoding $1$), at which point the corresponding
branch is run.
\fi

The concrete implementations of the lattice operations are also
specified using triples in this style.
\label{def:wf-clatt}
\ottusedrule{\ottdruleBot{}}
\ottusedrule{\ottdruleJoin{}}
\ottusedrule{\ottdruleFlows{}}
For the two-point lattice, it is easy to prove that the implemented
operators satisfy these specifications; \autoref{sec:extensions}
describes an analogous result for a lattice of sets of principals.

\iffull Going a bit further towards bridging the gap between the
symbolic rule and concrete machines, we prove specifications for the
generation of label expressions
\iffull
\ottusedrule{\ottdruleELabJCS{}}
\else
\ottusedrule{\ottdruleELab{}}
\fi
and for the code generated to implement the application of a symbolic IFC
symbolic rule. For instance, the case where the instruction is
allowed is described by the \iffull following \fi specification \iffull (the integer $1$ pushed on the output stack encodes the fact that the rule is allowed)\fi:
\iffull
\ottusedrule{\ottdruleAppRXXOKJCS{}}
\else
\ottusedrule{\ottdruleAppRXXOK{}}
\fi
\fi

\paragraph{Escaping-code Hoare triples}
To be able to specify the entire code of the generated fault
handler, we also define a second form of triple, $ \{  \ottnt{P}  \} \  \ottnt{c}   \ \{  \ottnt{Q}  \}^{ O }_{ \concretefont{pc} } $, which specifies mostly self-contained, total code $\ottnt{c}$ that
either makes \emph{exactly one} jump outside of $\ottnt{c}$ or returns out of
kernel mode.
\iffull This non-locality is needed because the fault handler checks
whether an information-flow violation is about to occur, and returns
to the user-mode caller if not, or jumps to an invalid address
otherwise.  \fi
More precisely, if $\ottnt{P}$ and $\ottnt{Q}$ are predicates on
$\kappa \times \concretesymbol{\sigma}$ and $O$ is a function from $\kappa \times
\concretesymbol{\sigma}$ to outcomes (the constants $ \mathtt{Success} $ and $ \mathtt{Failure} $),
then $ \{  \ottnt{P}  \} \  \ottnt{c}   \ \{  \ottnt{Q}  \}^{ O }_{ \concretefont{pc} } $ holds if, whenever the kernel instruction
memory $\phi$ contains the sequence $\ottnt{c}$ starting at the current
PC, the current cache and stack satisfy $\ottnt{P}$, and
\begin{itemize}
\item if $O$ computes $ \mathtt{Success} $ then the machine runs (in
  kernel mode) until it returns to user code at $\concretefont{pc}$,
  and $\ottnt{Q}$ is satisfied.
\item if $O$ computes $ \mathtt{Failure} $ then the machine runs (in
  kernel mode) until it halts ($pc = -1$ in kernel mode), and
  $\ottnt{Q}$ is satisfied.
\end{itemize}
\iffull
Or, in symbols,
  \[
   \{  \ottnt{P}  \} \  \ottnt{c}   \ \{  \ottnt{Q}  \}^{ O }_{ \concretefont{pc} }  \triangleq
  \begin{array}[t]{l}
 \ottnt{c} = \phi(n),\ldots,\phi(n+|c|-1) \land \ottnt{P}  \ottsym{(}  \kappa  \mathord{,}\,  \concretesymbol{\sigma}  \ottsym{)}  \Longrightarrow  \\
 \exists~\kappa'~\concretesymbol{\sigma}'.\;
    \begin{array}[t]{l}
    \ottnt{Q}  \ottsym{(}  \kappa'  \mathord{,}\,  \concretesymbol{\sigma}'  \ottsym{)} \\
    \land  (
    O  \ottsym{(}  \kappa  \mathord{,}\,  \concretesymbol{\sigma}  \ottsym{)}  \ottsym{=}   \mathtt{Success}   \Longrightarrow 
    \phi\vdash \runsToEscape
        { \langle  \text{\sf k}  }{ \kappa  }{ \concretesymbol{\mu} }{ [ \concretesymbol{\sigma} ]  }{ \ottnt{n}  \mathord{\scriptstyle @}  \mathtt{T}_\mathtt{D}\rangle }
        { \langle  \text{\sf u}  }{ \kappa' }{ \concretesymbol{\mu} }{ [ \concretesymbol{\sigma}' ] }{ \concretefont{pc} \rangle}
        {}
        )
        \\
   \land  (
    O  \ottsym{(}  \kappa  \mathord{,}\,  \concretesymbol{\sigma}  \ottsym{)}  \ottsym{=}   \mathtt{Failure}   \Longrightarrow 
    \phi\vdash  \runsToEscape
        { \langle  \text{\sf k}  }{ \kappa  }{ \concretesymbol{\mu} }{ [ \concretesymbol{\sigma} ]  }{ \ottnt{n}  \mathord{\scriptstyle @}  \mathtt{T}_\mathtt{D} \rangle}
        { \langle  \text{\sf k}  }{ \kappa' }{ \concretesymbol{\mu} }{ [ \concretesymbol{\sigma}' ] }{ \ottsym{-}  \ottsym{1}  \mathord{\scriptstyle @}  \mathtt{T}_\mathtt{D} \rangle}
        {}
        )
  \end{array}
  \end{array}
  \]

\fi
To compose self-contained code with escaping code,
we prove two composition laws for these triples, one for
pre-composing with specified self-contained code and another
for post-composing with arbitrary (unreachable) code:
\[ \ottdruleOXXCompose \ottinterrule
   \quad \ottdruleOXXAppend \ottinterrule \] We use these new triples
to specify the $ \text{\sf Ret} $ and $ \text{\sf Jump} $ instructions, which could not
be given useful specifications using the self-contained-code
triples:
\[\ottdruleOXXGenRet{}\quad\ottdruleOXXGenJump{}\]

Everything comes together in verifying the fault handler. We use
contained-code triples to specify everything except for $\ottsym{[}  \text{\sf Ret}  \ottsym{]}$,
$\ottsym{[}  \text{\sf Jump}  \ottsym{]}$, and the final $\ottkw{genIf}$, and then use the
escaping-code triple composition laws to connect the non-returning part of the
fault handler to the final $\ottkw{genIf}$.

\section{Refinement}
\label{sec:refinement}

We have two remaining verification goals.  First, we want to show that
the concrete machine of \autoref{sec:conc} (running the fault handler
of \autoref{sec:fault-handler} compiled from $\rulesabsifc$) enjoys
TINI.  Proving this directly for the concrete machine would be
dauntingly complex, so instead we show that the concrete machine is an
implementation of the abstract machine, for which noninterference will
be much easier to prove (\autoref{sec:ni}). Second, since a trivial
always-diverging machine also has TINI, we want to show that the
concrete machine is a \emph{faithful} implementation of the abstract
machine that emulates all its behaviors.

We phrase these two results using the notion of \emph{machine
  refinement}, which we develop in this section, and which we prove
in \autoref{sec:ni} to be TINI preserving.  In \autoref{sec:qa-c-ref}, we
prove a two-way refinement (one direction for each goal), between
the abstract and concrete machines, via the symbolic rule machine in both
directions.

From here on we sometimes mention different machines (abstract, symbolic
rule, or concrete) in the same statement (\EG when discussing refinement),
and sometimes talk about machines generically (\EG when defining TINI for
all our machines); for these purposes, it is useful to define a generic notion
of machine.

\newcommand{\steprwithcdots}{\cdot\!\!\stepr\!\!\cdot}

\begin{defn}
\label{def:gen-mach}
A \emph{generic machine} (or just \emph{machine}) is a 5-tuple
$M=(S,E,I,\steprwithcdots,\initf)$, where $S$ is a set of \emph{states}
(ranged over by $s$), $E$ is a set of \emph{events} (ranged over by $e$),
$\steprwithcdots\subseteq S \times (E+\{\tau\}) \times S$ is a step
relation, and $I$ is a set of \emph{input data} (ranged over by $i$) that
can be used to build \emph{initial states} of the machine with the function
$\initf\in I \to S$. We call \(E+\{\tau\}\) the set of \emph{actions} of
\(M\) (ranged over by $\action$).
\end{defn}
Conceptually, a machine's program is included in its input data and gets
``loaded'' by the function $\initf$, which also initializes the machine
memory, stack, and PC.  The notion of generic machine abstracts all
these details, allowing uniform definitions of refinement and TINI that
apply to all three of our IFC machines.  To avoid stating
it several times below, we stipulate that when we
instantiate~\autoref{def:gen-mach} to any of our IFC machines, $\initf$ must
produce an initial stack with no return frames.

A generic step \mbox{$\step{s_1}{e}{s_2}$}
or \mbox{$\step{s_1}{\tau}{s_2}$} produces event $e$ or is silent.
The reflexive-transitive closure of such steps, omitting silent steps (written
\mbox{$\multistep{s_1}{t}{s_2}$}) produces \emph{traces}---i.e., lists,
$t$, of events.
\iffull It is defined inductively by
\begin{equation}\label{eqn:traces}
\xrule{}{
\multistep{s}{\tnil}{s}}
\quad\;
\xrule{
\step{s_1}{e}{s_2} \quad
\multistep{s_2}{t}{s_3}}{
\multistep{s_1}{\tcons{e}{t}}{s_3}}
\quad\;
\xrule{
\step{s_1}{\tau}{s_2} \quad
\multistep{s_2}{t}{s_3}}{
\multistep{s_1}{t}{s_3}}
\end{equation}
where we write $\tnil$ for the empty trace and $\tcons{e}{t}$ for
consing $e$ to $t$.
\fi
When the
end state of a step starting in state $s$ is not relevant we write
$\step{s}{e}{}$, and similarly $\ftrace{s}{t}$ for traces.

When relating executions of two different machines through a
refinement, we establish a correspondence between their traces. This
relation is usually derived from an elementary relation on events,
\(\match \subseteq E_1 \times E_2\), which is lifted to actions and
traces: \iffull
\begin{defn}[Matching]\label{def:match-traces}
Given a relation \(\match \subseteq E_1 \times E_2\) between two sets of
events, its lifts to actions and traces are defined:
\fi
\[
\begin{array}{ccl}
\alpha_1\lift{\match}\alpha_2 &\triangleq &(\alpha_1=\tau=\alpha_2\;\vee\;
                     \alpha_1=e_1\match e_2=\alpha_2) \\
\vec{x}\lift{\match}\vec{y} &\triangleq & \length{\vec{x}} = \length{\vec{y}} \wedge
\forall\, i < \length{\vec{x}}.\:x_i\match y_i.
\end{array}
\]
\iffull
\end{defn}
\fi

\iffull We are now ready to define refinement. \fi

\newcommand{\picturerowheight}{\iffull 2em\else .7em\fi}

\begin{defn}[Refinement]
  Let \(M_1 = (S_1,E_1,I_1,\cdot \stepr_1 \cdot,\initf_1)\) and \(M_2 =
  (S_2,E_2,I_2,\cdot \stepr_2 \cdot,\initf_2)\) be two machines. A
  \emph{refinement} of \(M_1\) into \(M_2\) is a pair of relations
  \((\match_i,\match_e)\), where \(\match_i \subseteq I_1 \times I_2\)
  and \(\match_e \subseteq E_1 \times E_2\), such that whenever \(i_1
  \match_i i_2\) and \(\ftrace{\initf_2(i_2)}{t_2}\), there exists a
  trace \(t_1\) such that \(\ftrace{\initf_1(i_1)}{t_1}\) and
  \(t_1 \Matchu{e} t_2\).
  We also say that \(M_2\) \emph{refines} \(M_1\).
Graphically:
\vspace{-\smallskipamount}
 \begin{center}
    \begin{tikzpicture}
      \matrix (m) [matrix of math nodes, row sep=\iffull 2em\else .8em\fi,column sep=2em,text
      height=1.5ex, text depth=0.25ex]
      { i_1 & \initf_1(i_1) & {} \\
        i_2 & \initf_2(i_2) & {} \\};
      \path (m-1-1) edge[-] node[left] {$\match_i$} (m-2-1)
            (m-1-2) edge[dashed, -to*] node[above] (t1) {$t_1$} (m-1-3)
            (m-2-2) edge[-to*]         node[below] (t2) {$t_2$} (m-2-3) ;
      \draw[rounded corners=1pt,dashed]
      (t1)
      -- ($(t1)+(2cm,0)$)
      -- node[right] {$\Matchu{e}$}
      ($(t2)+(2cm,0)$)
      -- (t2) ;
    \end{tikzpicture}
  \end{center}
\vspace{-\smallskipamount}
(Plain lines denote premises, dashed ones conclusions.)
\end{defn}

In order to prove refinement, we need a variant
that considers executions starting at arbitrary related
states.
\begin{defn}[Refinement via states]\label{dfn:refStates}
  Let \(M_1\), \(M_2\) be as above. A \emph{state refinement}
  of \(M_1\) into \(M_2\) is a pair of relations \((\match_s,
  \match_e)\), where \(\match_s \subseteq S_1 \times S_2\) and
  \(\match_e \subseteq E_1 \times E_2\), such that, whenever \(s_1
  \match_s s_2\) and \(\ftrace{s_2}{t_2}\), there exists
  \(t_1\) such that \(\ftrace{s_1}{t_1}\) and \(t_1
  \Matchu{e} t_2\).
\iffull
  \begin{center}
    \begin{tikzpicture}
      \matrix (m) [matrix of math nodes, row sep=\iffull 2em\else .8em\fi,column sep=2em,text
      height=1.5ex, text depth=0.25ex]
      { s_1 & {} \\
        s_2 & {} \\};

      \path (m-1-1) edge[-] node[left] {$\match_s$} (m-2-1)

            (m-1-1) edge[dashed, -to*] node[above] (t1) {$t_1$} (m-1-2)
            (m-2-1) edge[-to*]         node[below] (t2) {$t_2$} (m-2-2) ;

      \draw[rounded corners=1pt,dashed]
      (t1)
      -- ($(t1)+(2cm,0)$)
      -- node[right] {$\Matchu{e}$}
      ($(t2)+(2cm,0)$)
      -- (t2) ;
    \end{tikzpicture}
  \end{center}
\fi
\end{defn}

If the relation on inputs is compatible with the one on states,
we can use state refinement to prove refinement.

\begin{lemma}\label{lem:refStates}
  Suppose \(i_1 \match_i i_2 \Rightarrow \initf_1(i_1) \match_s
  \initf_2(i_2)\), for all \(i_1\) and \(i_2\). If
  \((\match_s,\match_e)\) is a state refinement then
  \((\match_i,\match_e)\) is a refinement.
\end{lemma}

\iffull
Our plan to derive a refinement between the abstract and concrete
machines via the symbolic rule machine requires composition of
refinements.

\begin{lemma}[Refinement Composition]\label{lem:refCompose}
Let $(\match_i^{12},\match_e^{12})$ be a refinement  between \(M_1\)
and \(M_2\), and $(\match_i^{23},\match_e^{23})$ a refinement
between \(M_2\) and \(M_3\). The pair
$(\match_i^{23} \circ \match_i^{12}, \match_e^{23} \circ \match_e^{12})$
that composes the matching relations for initial data and events on
each layer is a refinement between \(M_1\) and \(M_3\).
\fi
\iffull
This can be summarized in the following diagram:
  \begin{center}
    \begin{tikzpicture}[descr/.style={fill=white,inner sep=2.5pt}]
      \matrix (m) [matrix of math nodes, row sep=\iffull 2em\else .8em\fi,column sep=2em,text
      height=1.5ex, text depth=0.25ex]
      { i_1 & s_1 \\
        i_2 & s_2 \\
        i_3 & s_3 \\};
      \path[dashed,to*,] (m-1-1) edge node[above] (t1) {$t_1$} (m-1-2);
      \path[dashed,to*,] (m-2-1) edge node[above] (t2) {$t_2$} (m-2-2);
      \path[to*,] (m-3-1) edge node[below] (t3) {$t_3$} (m-3-2);
      \path[-,font=\small] (m-1-1) edge node[left] {$\match_i^{12}$} (m-2-1);
      \path[-,font=\small] (m-2-1) edge node[left] {$\match_i^{23}$} (m-3-1);
      \draw[rounded corners=1pt,dashed]
      (t1)
      -- ($(t1)+(2cm,0)$)
      -- node[right] {$\lift{\match_e^{23} \circ \match_e^{12}}$}
      ($(t3)+(2cm,0)$)
      -- (t3) ;

      \draw[rounded corners=1pt]
      (m-1-1)
      -- ($(m-1-1)-(1cm,0)$)
      -- node[left] {$\match_i^{23} \circ \match_i^{12}$}
      ($(m-3-1)-(1cm,0)$)
      -- (m-3-1) ;

    \end{tikzpicture}
  \end{center}
\label{lem:ref-composition}
\end{lemma}
\fi

\section{Refinements Between Concrete and Abstract}
\label{sec:qa-c-ref}

In this section, we show that (1) the concrete machine refines the symbolic
rule machine, and (2) vice versa.  Using (1) we will be able to show in
\autoref{sec:ni} that the concrete machine is noninterfering.  From (2) we
know that the concrete machine faithfully implements the abstract one,
exactly reflecting its execution traces.

\iffull
\subsection{Abstract and symbolic rule machines}
\else
\paragraph*{Abstract and symbolic rule machines}
\fi

The symbolic rule machine (with the rule table $\rulesabsifc$) is a
simple reformulation of the abstract machine.  Their step relations
are (extensionally) equal, and started from the same input data they
emit the same traces.

\begin{defn}[Abstract and symbolic rule machines as generic machines]\label{def:a-qam}
For both abstract and symbolic rule machines, input
data is a 4-tuple $(p,\mathit{args},n,\ottnt{L})$ where $p$ is a program,
$\mathit{args}$ is a list of atoms (the initial stack), and $n$ is the
size of the memory, initialized with $n$ copies of $\ottsym{0}  \mathord{\scriptstyle @}  \ottnt{L}$. The initial
PC is  $\ottsym{0}  \mathord{\scriptstyle @}  \ottnt{L}$.
\end{defn}

\begin{lemma}\label{lem:refAbsQuasi}
  The symbolic rule machine instantiated with the rule table
  $\rulesabsifc$ refines the abstract machine
  through \((=,=)\).
\end{lemma}

\iffull
\subsection{Concrete machine refines symbolic rule machine}
\else
\paragraph*{Concrete machine refines symbolic rule machine}
\fi
We prove this refinement using a fixed but arbitrary rule table,
$\mathcal{R}$, an abstract lattice of labels, and a concrete lattice of
tags.  The proof uses the correctness of the fault handler
(\autoref{sec:fault-handler-correct}), so we assume that the fault
handler of the concrete machine corresponds to the rule table of the
symbolic rule machine ($\phi = \phi_{\mathcal{R}}$) and that the
encoding of abstract labels as integer tags is correct.

\begin{defn}[Concrete machine as generic machine]
The input data of the concrete machine is a 4-tuple
$(p,\mathit{args},n,\mathtt{T})$ where $p$ is a program, $\mathit{args}$ is a
list of concrete atoms (the initial stack), and the initial memory is $n$
copies of $\ottsym{0}  \mathord{\scriptstyle @}  \mathtt{T}$. The initial PC is $\ottsym{0}  \mathord{\scriptstyle @}  \mathtt{T}$. The machine starts in user
mode, the cache is initialized with an illegal opcode so that the first
instruction always faults\iffull{} (giving the fault handler a chance to run and
install a correct rule without requiring the initialization process to
invent one)\fi, and the fault handler code parameterizing the
machine is installed in the initial privileged instruction memory $\phi$.
\end{defn}

The input data and events of the symbolic rule and concrete machines
are of different kinds; they are matched using relations ($\match_i^{c}$ and
$\match_e^{c}$ respectively) stipulating that
payload values should be equal and that labels should correspond to tags
modulo the function $ \mathsf{Tag} $ of the concrete lattice.
\begin{gather*}
\xrule{
 \mathit{args'}=\mathsf{map~(\lambda(\ottnt{n}  \mathord{\scriptstyle @}  \ottnt{L}).\:\ottnt{n}  \mathord{\scriptstyle @}   \mathsf{Tag} ( \ottnt{L} ) )}~\mathit{args}
 } { (p,\mathit{args},n,\ottnt{L}) \match_i^{c}
 (p,\mathit{args'},n, \mathsf{Tag} ( \ottnt{L} ) )
 }\quad
\xrule{\strut}{\ottnt{n}  \mathord{\scriptstyle @}  \ottnt{L} \match_e^{c} \ottnt{n}  \mathord{\scriptstyle @}   \mathsf{Tag} ( \ottnt{L} ) }
\end{gather*}

\begin{thm}\label{thm:refConcQuasi}
The concrete IFC machine refines the symbolic rule machine,
 through \((\match_i^{c},\match_e^{c})\).
\label{thm:qa-c-ref}
\end{thm}

We prove this theorem by a refinement via states
(\autoref{lem:ref-state-qa-c}); this, in turn, relies on two technical
lemmas (\ref{lemma:ref-non-fault} and~\ref{lemma:ref-fault}).

\iffull
We begin by defining a matching relation $\match_s^c$ between the states of the concrete
and symbolic rule machines such that
\begin{enumerate}
\item \(i_q\match_i^c i_c \Rightarrow \initf_q(i_q) \match_s^c \initf_c(i_c)\),
\item $(\match_s^c, \match_e^c)$ is a state refinement of the symbolic rule machine into the concrete machine.
\end{enumerate}
We define $\match_s^c$ as
\else
The matching relation $\match_s^c$ between the states of the concrete
and symbolic rule machines is defined as\fi
\iffull\begin{equation}\else\begin{equation*}\fi\label{def:concStateMatch}
\xrule{
        \begin{array}{c}
        \mathcal{R} \vdash \kappa \quad
        \sigma_q \match_\sigma \concretesymbol{\sigma}_c \quad \mu_q \match_m \concretesymbol{\mu}_c
        \end{array}
      }
      { \mu_q,[ \sigma_q ],\ottnt{n}  \mathord{\scriptstyle @}  \ottnt{L} \; \match_s^c \;
         \text{\sf u} ,\kappa,\concretesymbol{\mu}_c,[ \concretesymbol{\sigma}_c ],\ottnt{n}  \mathord{\scriptstyle @}   \mathsf{Tag} ( \ottnt{L} )  }
\iffull\end{equation}\else\end{equation*}\fi
where the new notations are defined as follows.  The relation
$\match_m$ demands that the memories be equal up to the conversion of
labels to concrete tags. The relation $\match_\sigma$ on stacks is
similar, but additionally requires that return frames in the concrete
stack have their privilege bit set to $ \text{\sf u} $. The basic idea is
to match, in $\match_s^{c}$, only concrete states that are in user
mode.  We also need to track an extra invariant,
$\mathcal{R} \vdash \kappa$,
which means that the cache $\kappa$ is consistent with the
table $\mathcal{R}$---i.e., $\kappa$ never lies.  More
precisely, the output part of $\kappa$ represents the result of
applying the symbolic rule judgment of $\mathcal{R}$ to
the opcode and labels represented in the input part of $\kappa$.
\[
\begin{array}{rl}
\mathcal{R} \vdash  [  {\kappa_i}  ,  \kappa_o  ]   \triangleq \; & \forall \ottnt{opcode}\; \ottnt{L_{{\mathrm{1}}}}\; \ottnt{L_{{\mathrm{2}}}}\; \ottnt{L_{{\mathrm{3}}}}\; L_{pc},\\[.2em]
& \quad {\kappa_i} =  \begin{array}{|@{\;}l@{\;}|@{\;}l@{\;}|@{\;}l@{\;}|@{\;}l@{\;}|@{\;}l@{\;}|}
                       \hline
                           \ottnt{opcode}  &   \mathsf{Tag} ( L_{pc} )   &   \mathsf{Tag} ( \ottnt{L_{{\mathrm{1}}}} )   &   \mathsf{Tag} ( \ottnt{L_{{\mathrm{2}}}} )   &   \mathsf{Tag} ( \ottnt{L_{{\mathrm{3}}}} )   \\
                       \hline
                       \end{array}  \Rightarrow \\[.2em]
& \quad
  \exists L_{rpc}\; L_r,  \vdash_{ \mathcal{R} } \ruleeval{ \ottsym{(}  L_{pc}  \mathord{,}\,  \ottnt{L_{{\mathrm{1}}}}  \mathord{,}\,  \ottnt{L_{{\mathrm{2}}}}  \mathord{,}\,  \ottnt{L_{{\mathrm{3}}}}  \ottsym{)} }{ \ottnt{opcode} }{ L_{rpc} }{ L_r }  ~ \land ~
                             \kappa_o=\ottsym{(}  \mathsf{Tag} \, \ottsym{(}  L_{rpc}  \ottsym{)}  \mathord{,}\,  \mathsf{Tag} \, \ottsym{(}  L_r  \ottsym{)}  \ottsym{)}
\end{array}
\]

To prove refinement via states, we must account for two situations.
First, suppose the concrete machine can take a user step. In this case, we
match that step with a single symbolic rule machine step. We write
$\privstate{cs}{\pi}$ to denote a concrete state $cs$ whose privilege
bit is $\pi$.

\begin{lemma}[Refinement, non-faulting concrete step]
\label{lemma:ref-non-fault}
Let $\privstate{cs_1}{ \text{\sf u} }$ be a concrete state and suppose that
$\step{\privstate{cs_1}{ \text{\sf u} }}{\alpha_c}{\privstate{cs_2}{ \text{\sf u} }}$. Let
$qs_1$ be a symbolic rule machine state with
$qs_1 \match_s^{c} \privstate{cs_1}{ \text{\sf u} }$.  Then there exist
$qs_2$ and $\alpha_a$ such that
$\step{qs_1}{\alpha_a}{qs_2}$, with
$qs_2 \match_s^{c} \privstate{cs_2}{ \text{\sf u} }$, and
$\alpha_a \lift{\match_e^c} \alpha_c$.
\iffull Graphically:
\begin{center}
  \begin{tikzpicture}[descr/.style={fill=white,inner sep=2.5pt}]
    \matrix (m) [matrix of math nodes, row sep=\iffull 2em\else .8em\fi,column sep=2em,text
    height=1.5ex, text depth=0.25ex]
    { qs_1                        & & qs_2 \\
      \privstate{cs_1}{ \text{\sf u} } & & \privstate{cs_2}{ \text{\sf u} } \\ };
    \path[->,font=\small] (m-2-1) edge node[below] (ec) {\small  $\alpha_c$} (m-2-3);
    \path[-,font=\small] (m-2-1) edge node[left] {$\match_s^{c}$} (m-1-1);
    \path[->,densely dashed] (m-1-1) edge node [above] (ea) {\small  $\alpha_a$} (m-1-3);
    \path[densely dashed] (m-2-3) edge node [right] {$\match_s^{c}$}
    (m-1-3);
    \draw[rounded corners=1pt,dashed]
    (ec)
    -- ($(ec)+(0,-0.3cm)$)
    -- ($(ec)+(2cm,-0.3cm)$)
    -- node[right] {$\lift{\match_e^c}$}
    ($(ea)+(2cm,+0.3cm)$)
    -- ($(ea)+(0,+0.3cm)$)
    -- (ea) ;
  \end{tikzpicture}
\end{center}
\fi
\end{lemma}
\iffull \begin{proof}
We know that $qs_1 \match_s^{c} \privstate{cs_1}{ \text{\sf u} }$.  By
definition of $\match_s^{c}$ in~\eqnref{def:concStateMatch}, $qs_1$ and
$\privstate{cs_1}{ \text{\sf u} }$ are at the same opcode with the same
stack and memory (up to translation between labels and tags), and
$\mathcal{R} \vdash \kappa(\privstate{cs_1}{ \text{\sf u} })$.  Thus
$\kappa(\privstate{cs_1}{ \text{\sf u} })$ matches a line of the symbolic
IFC rule table, and since the concrete machine performs a user step
from $\privstate{cs_1}{ \text{\sf u} }$ to $\privstate{cs_2}{ \text{\sf u} }$,
it is a line that allows a step to be taken.  We conclude that the
symbolic rule machine is able to perform the step to $qs_2$ as
required.
\end{proof}
\else
Since the concrete machine is able to make a user step, the input part
of the cache must match the opcode and data of the current state. But
the invariant $\mathcal{R} \vdash \kappa$ says that the corresponding
symbolic rule judgment holds.  Hence the symbolic rule machine can
also make a step from $qs_2$, as required.
\fi

The second case is when the concrete machine faults into kernel mode
and returns to user mode after some number of steps.

\begin{lemma}[Refinement, faulting concrete step]
\label{lemma:ref-fault}
Let $\privstate{cs_0}{ \text{\sf u} }$ be a concrete state, and suppose
that the concrete machine does a faulting step to
$\privstate{cs_1}{ \text{\sf k} }$, stays in kernel mode until
$\privstate{cs_n}{ \text{\sf k} }$, and then exits kernel mode by stepping to
$\privstate{cs_{n+1}}{ \text{\sf u} }$.  Let $qs_0$ be a state of the
symbolic rule machine that matches $\privstate{cs_0}{ \text{\sf u} }$.
Then $qs_0 \match_s^{c} \privstate{cs_{n+1}}{ \text{\sf u} }$.
\iffull
Graphically:
\begin{center}
  \begin{tikzpicture}[descr/.style={fill=white,inner sep=2.3pt}]
      \matrix (n) [matrix of math nodes, row sep=3em, column sep=1.1em,text
      height=1ex, text depth=0.25ex]
      {  qs_0  \\
        \privstate{cs_0}{ \text{\sf u} } & \privstate{cs_1}{ \text{\sf k} } & \privstate{cs_n}{ \text{\sf k} } & \privstate{cs_{n+1}}{ \text{\sf u} }  \\};
      \path[->,font=\small] (n-2-1) edge node[below] {$\esilent$} (n-2-2);
      \path[tok*,font=\small] (n-2-2) edge node[below] {} (n-2-3);
      \path[->,font=\small] (n-2-3) edge node[below] {$\esilent$} (n-2-4);
      \path[-,font=\small] (n-2-1) edge node[left] {$\match_s^{c}$} (n-1-1);
      \path[densely dashed] (n-1-1) edge node [above] {$\match_s^{c}$} (n-2-4);
    \end{tikzpicture}
    \end{center}
\fi
\end{lemma}

\iffull{}\else
Lemmas \ref{lemma:ref-non-fault} and \ref{lemma:ref-fault} can be summarized graphically by:
\begin{center}
  \begin{tikzpicture}[descr/.style={fill=white,inner sep=2.3pt}]
    \matrix (m) [matrix of math nodes, row sep=\iffull 2em\else .8em\fi,column sep=2em,text
    height=1.5ex, text depth=0.25ex]
    { qs_1                        & & qs_2 \\
      \privstate{cs_1}{ \text{\sf u} } & & \privstate{cs_2}{ \text{\sf u} } \\ };
    \path[->,font=\small] (m-2-1) edge node[below] (ec) {\small  $\alpha_c$} (m-2-3);
    \path[-,font=\small] (m-2-1) edge node[left] {$\match_s^{c}$} (m-1-1);
    \path[->,densely dashed] (m-1-1) edge node [above] (ea) {\small  $\alpha_a$} (m-1-3);
    \path[densely dashed] (m-2-3) edge node [right] {$\match_s^{c}$}
    (m-1-3);
    \draw[rounded corners=1pt,dashed]
    (ec)
    -- ($(ec)+(0,-0.3cm)$)
    -- ($(ec)+(1.7cm,-0.3cm)$)
    -- node[right] {$\lift{\match_e^c}$}
    ($(ea)+(1.7cm,+0.3cm)$)
    -- ($(ea)+(0,+0.3cm)$)
    -- (ea) ;
    \end{tikzpicture}
    \hfill
  \begin{tikzpicture}[descr/.style={fill=white,inner sep=2.3pt}]
      \matrix (n) [matrix of math nodes, row sep=3em, column sep=1.1em,text
      height=1ex, text depth=0.25ex]
      {  qs_0  \\
        \privstate{cs_0}{ \text{\sf u} } & \privstate{cs_1}{ \text{\sf k} } & \privstate{cs_n}{ \text{\sf k} } & \privstate{cs_{n+1}}{ \text{\sf u} }  \\};
      \path[->,font=\small] (n-2-1) edge node[below] {$\esilent$} (n-2-2);
      \path[tok*,font=\small] (n-2-2) edge node[below] {} (n-2-3);
      \path[->,font=\small] (n-2-3) edge node[below] {$\esilent$} (n-2-4);
      \path[-,font=\small] (n-2-1) edge node[left] {$\match_s^{c}$} (n-1-1);
      \path[densely dashed] (n-1-1) edge node [above] {$\match_s^{c}$} (n-2-4);
    \end{tikzpicture}
    \end{center}
\fi
\iffull
\begin{proof}
Since the concrete machine performs a faulting step from
$\privstate{cs_0}{ \text{\sf u} }$ to $\privstate{cs_1}{ \text{\sf k} }$, we
know that the current cache input,
${\kappa_i}(\privstate{cs_1}{ \text{\sf k} })$, corresponds to the current
instruction and the tags it manipulates (they have been put there when
entering kernel mode).  Now, there are two cases.  If evaluating the
corresponding IFC rule at the symbolic rule level succeeds, then we
apply~\autoref{lem:fhdl-succ} to conclude directly.  Otherwise, we
apply~\autoref{lem:fhdl-fail} and derive that the fault handler ends
up in a failing state in kernel mode.  This contradicts our initial
hypothesis saying that the concrete machine performed a sequence of
steps returning to user-mode.
\end{proof}
\else
To prove this lemma, we must consider two cases.  If the corresponding
symbolic rule judgment holds, then we apply~\autoref{lem:fhdl-succ} to
conclude directly---\IE the machine exits kernel code into user mode.
Otherwise, we apply~\autoref{lem:fhdl-fail} and derive a contradiction
that the fault handler ends in a failing state in kernel mode.
\fi

Given two matching states of the concrete and symbolic rule machines,
and a concrete execution starting at that concrete state, these two
lemmas can be applied repeatedly to build a matching execution of the
symbolic rule machine. There is just one last case to consider, namely
when the execution ends with a fault into kernel mode and never
returns to user mode. However, no output is produced in this case,
guaranteeing that the full trace is matched. We thus derive the
following refinement via states, of which \autoref{thm:qa-c-ref} is a
corollary.

\begin{lemma}
\label{lem:ref-state-qa-c}
The pair $(\match_s^{c},\match_e^{c})$ defines a refinement via states
between the symbolic rule machine and the concrete machine.
\end{lemma}

\iffull
\subsection{Concrete machine refines abstract machine}
\else
\paragraph*{Concrete machine refines abstract machine}
\fi
By composing the refinement of \autoref{lem:refAbsQuasi} and the refinement
of \autoref{thm:qa-c-ref} instantiated to the concrete machine running
$\tiniFH$, we can conclude that the concrete machine refines the abstract one:

\iffull
\begin{thm}
  \label{thm:concRefinesAbs}
  The concrete IFC machine refines the abstract IFC machine via
  $(\match_s^{c},\match_e^{c})$.
\end{thm}
\fi

\iffull
\subsection{Abstract machine refines concrete machine}
\else
\paragraph*{Abstract machine refines concrete machine}
\fi
The previous refinement, $(\match_s^{c},\match_e^{c})$, would also
hold if the fault handler never returned when called. So, to ensure
the concrete machine reflects the behaviors of the abstract machine,
we next prove an inverse refinement:
\begin{thm}
  \label{cor:absRefinesConc}
  The abstract IFC machine refines the concrete IFC machine via
  \((\match_i^{-c},\match_e^{-c})\), where
  \(\match_i^{-c}\) and \(\match_e^{-c}\) are the relational inverses
  of \(\match_i^c\) and \(\match_e^{c}\).
\end{thm}
This guarantees that traces of the abstract machine are also emitted
by the concrete machine. As above we use the symbolic rule machine as
an intermediate step and show a state refinement of the concrete
 into the symbolic rule machine. We rely on the following
lemma.

\begin{lemma}[Forward refinement]
\label{lem:fwd-sim}
Let \(qs_0\) and \(cs_0\) be two states with
\(cs_0 \match_s^{-c} qs_0\). Suppose that the symbolic rule
machine takes a step \(\step{qs_0}{\alpha_a}{qs_1}\). Then there
exist concrete state \(cs_1\) and action $\alpha_c$ such that
$\multistep{cs_0}{\alpha_c}{cs_1}$, with \(cs_1 \match_s^{-c} qs_1 \) and
$\alpha_c \lift{\match_e^{-c}} \alpha_a $.
\iffull
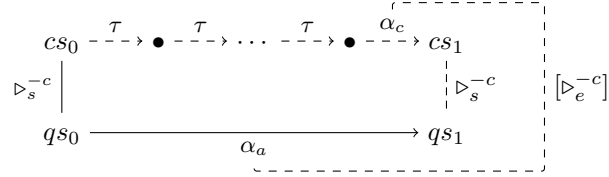
  \begin{center}
  \begin{tikzpicture}[descr/.style={fill=white,inner sep=2.5pt}]
   \matrix (m) [matrix of math nodes, row sep=\iffull 2em\else .8em\fi,column sep=2em,text height=1.5ex, text depth=0.25ex]
   {  cs_0 & \bullet & \cdots & \bullet & cs_1 \\
       qs_0 & & & & qs_1 \\ };
   \path[->,dashed,font=\small] (m-1-1) edge node[above] {$\esilent$} (m-1-2);
   \path[->,dashed,font=\small] (m-1-2) edge node[above] {$\esilent$} (m-1-3);
   \path[->,dashed,font=\small] (m-1-3) edge node[above] {$\esilent$} (m-1-4);
   \path[->,dashed,font=\small] (m-1-4) edge node[above] (ec) {$\alpha_c$} (m-1-5);
   \path[-,font=\small] (m-2-1) edge node[left] {$\match_s^{-c}$} (m-1-1);
   \path[densely dashed] (m-1-5) edge node [right] {$\match_s^{-c}$}
   (m-2-5);
   \path[->,font=\small] (m-2-1) edge node[below] (ea) {$\alpha_a$}
   (m-2-5);
   \draw[rounded corners=1pt,dashed]
   let \p1 = ($(ec)+(2cm,+0.3cm)$) in
   let \p2 = ($(ea)+(0,-0.3cm)$) in
   (ec)
   -- ($(ec)+(0,+0.3cm)$)
   -- ($(ec)+(2cm,+0.3cm)$)
   -- node[right] {$\lift{\match_e^{-c}}$}
   (\x1,\y2)
   -- ($(ea)+(0,-0.3cm)$)
   -- (ea) ;
 \end{tikzpicture}
 \end{center}
 where \(\match_s^{-c}\) and \(\match_e^{-c}\) denote the inverses of
 \(\match_s^c\) and \(\match_e^c\), respectively.
\fi
\end{lemma}

\begin{proof}
Because $cs_0\match_s^{-c} qs_0$, the cache is consistent with the
symbolic rule table $\mathcal{R}$.  If the cache input matches the opcode
and data of $cs_0$, then (because
\(\step{qs_0}{\alpha_a}{qs_1}\)) the cache output must allow a step
$\step{cs_0}{\alpha_c}{cs_1}$ as required.  On the other hand, if the
cache input does not match the opcode and data of $cs_0$, then a
cache fault occurs, loading the cache input and calling the fault
handler.  By \autoref{lem:fhdl-succ} and the fact
that \(\step{qs_0}{\alpha_a}{qs_1}\), the cache output is computed to
be consistent with $\mathcal{R}$, and this allows the concrete step as
claimed.
\end{proof}

\iffull
\subsection{Discussion}
\else
\paragraph*{Discussion}
\fi
The two top-level refinement properties (\autoref{thm:refConcQuasi} and
\autoref{cor:absRefinesConc}) share the same notion of matching relations
but they have been proved independently in our Coq development.  In
the context of compiler verification~\cite{Leroy09,SevcikVNJS11},
another proof methodology has been favored: a backward simulation
proof can be obtained from a proof of forward simulation under the
assumption that the lower-level machine is deterministic.
(CompCertTSO~\cite{SevcikVNJS11} also requires a \emph{receptiveness}
hypothesis that trivially holds in our context.)
Since our concrete machine is deterministic, we could apply a similar
technique.  However, unlike in compiler verification where it is
common to assume that the source program has a well-defined semantics
(i.e. it does not get stuck), we would have to consider the
possibility that the high-level semantics (the symbolic rule machine)
might block and prove that in this case either the IFC enforcement
judgment is stuck (and \autoref{lemma:ref-fault} applies) or the
current symbolic rule machine state and matching concrete state are
both ill-formed.

\section{Noninterference}
\label{sec:ni}

In this section we define TINI~\cite{askarov08:TINI_leaks_more_than_1_bit,
  HedinS11} for generic machines\iffull{} (recall \autoref{def:gen-mach}),
and present a set of {\em unwinding conditions}~\cite{GoguenM84} sufficient
to guarantee TINI for a generic machine (\autoref{thm:TINI}); we \else, \fi
show that the abstract machine of~\autoref{sec:abstract} \iffull satisfies
these unwinding conditions and thus \fi satisfies TINI
(\autoref{thm:abstractIFCtini}), that TINI is preserved by refinement
(\autoref{thm:tininPreserved}), and finally, using the fact that the concrete
IFC machine refines the abstract one (\autoref{thm:refConcQuasi}), that the
concrete machine satisfies TINI (\autoref{thm:concTINI}).

\paragraph*{Termination-insensitive noninterference (TINI)}
\label{para:tini-def}

To define noninterference, we need to talk about what can be observed
about the output trace produced by a run of a machine.
\begin{defn}[Observation]
A \emph{notion of observation} for a generic machine is a 3-tuple
$(\Omega,\filter{\cdot}{\cdot},\cdot\!\!\approx^i_{\cdot}\!\!\cdot)$.
$\Omega$ is a set of \emph{observers} (i.e., different degrees of
power to observe), ranged over by $o$.  For each $o\in\Omega$,
\mbox{$\filter{o}{\cdot}\subseteq E$} is a predicate of
\emph{observability of events for observer $o$}, and
\mbox{$\cdot \approx^i_o \cdot \subseteq I\times I$} is a relation of
\emph{indistinguishability of input data for observer $o$}.
\end{defn}
\iffull
The predicate $\filter{o}{e}$ is used to filter unobservable
events from traces (written $\filter{o}{t}$):\ch{one line}
\begin{align*}
\filter{o}{\tnil} &= \tnil\\
\filter{o}{\tcons{e}{t}} &= \left\{
  \begin{array}{ll}
  \tcons{e}{\filter{o}{t}}
    & \text{if $\filter{o}{e}$}
  \\[0.5em]
  \filter{o}{t}
    & \text{otherwise}
  \end{array}
\right.
\end{align*}
\else
We write $\filter{o}{t}$ for the trace in which all unobservable
events in $t$ are filtered out using $\filter{o}{\cdot}$.
\fi
\iffull
Also a notion of \emph{indistinguishability of traces}
(written $t_1 \approx^t t_2$) is defined inductively:
\begin{equation}\label{eqn:indist-traces}
\xrule{}{\tnil \approx^t t}
\qquad
\xrule{}{t \approx^t \tnil}
\qquad
\xrule{t_1 \approx^t t_2}{
       \tcons{e}{t_1} \approx^t \tcons{e}{t_2}}
\end{equation}
This definition
\else
We write $t_1 \approx^t t_2$ to say that traces $t_1$ and $t_2$ are
{\em indistinguishable}; this
\fi
truncates the longer trace to the same length as the
shorter and then demands that the remaining elements be pairwise
identical.

\begin{defn}[TINI]\label{defn:TINI}
A machine $(S,E,I,\cdot\! \stepr\! \cdot,\initf)$ with a notion of observation
$(\Omega,\filter{\cdot}{\cdot},\cdot\!\approx^i_{\cdot}\!\cdot)$ satisfies
TINI if, for any
observer $o\in\Omega$, pair of indistinguishable initial data
\mbox{$i_1 \approx^i_o i_2$}, and pair of executions
$\ftrace{\Initf{i_1}}{t_1}$ and $\ftrace{\Initf{i_2}}{t_2}$, we have
$\filter{o}{t_1} \approx^t \filter{o}{t_2}$.
\end{defn}
Notice that the input data for our machines includes the program to be
executed; hence, we can apply the definition above to the execution of
different programs.
The reason for calling this notion ``termination insensitive'' is
that, because of truncated traces in \eqnref{eqn:indist-traces}, we
only model the case where we distinguish two runs of the same program
by observing two distinguishable events that occur on the same
position. Hence, this definition does not attempt to protect against
attackers that try to learn a secret by seeing whether a program
terminates or not: our observers cannot distinguish between successful
termination, failure with an error, or entering an infinite loop with
no observable output.
This TINI property is standard for a machine with
output~\cite{askarov08:TINI_leaks_more_than_1_bit,HedinS11}.%
\footnote{It is called ``progress-insensitive noninterference'' in a recent
  survey~\cite{HedinS11}.
\iffull
  We have stated it for inductively defined
  executions and traces\iffull~\eqnref{eqn:traces}\fi, which is all we need
  in this paper, but it can easily be lifted to coinductive executions and traces:
  not only successfully terminating and finitely failing executions, but
  also infinite executions.  This holds because TINI is a 2-safety
  hyperproperty\iffull~\cite{ClarksonS10}\fi{}; a formal proof of this can be found in
  our Coq development.
\fi
}

\ifever
\ch{Try to relate our definition (at least informally) to ``standard''
  TINI definitions: \EG knowledge-based TINI definition of
  Askarov~\ETAL\cite{askarov08:TINI_leaks_more_than_1_bit};
  are there any other definitions for machines with output?}
\fi

\iffull
\paragraph*{Unwinding conditions}
\label{para:unwinding}

\ifmuchlater
\ch{It would help to (also?) draw the unwinding conditions as
  diagrams, as we did in the many talks we had on the topic.}
\bcp{Yes, this would be useful, but I worry that it's going to take a lot of
space that we'll be reluctant to give up, at least in the conference
version.  Let's wait till later and see how things look...}
\fi

Having defined TINI for generic notions of machine and observation, we
now explain a sufficient set of conditions for such a machine to have
the TINI property and sketch a proof of TINI from these conditions.
The proof technique is standard~\cite{GoguenM84}.

A silent action cannot be observed, so we extend the given predicate
$\filter{o}{e}$ to actions by stating that $\filter{o}{\esilent}$
never holds.  From this we inductively define a notion
of \emph{indistinguishability of actions to observer $o$} (written
$\action_1 \approx^a_o \action_2$):
\begin{equation}\label{eqn:indist-events}
\xrule{}{\action \approx^a_o \action}
\qquad
\xrule{\neg\filter{o}{\action_1}\quad \neg\filter{o}{\action_2}}{\action_1 \approx^a_o \action_2}.
\end{equation}
Two actions are indistinguishable to $o$ if either they are equal, or
if neither can be observed by $o$.

\begin{thm}
\label{thm:TINI}
A machine $(S,E,I,\cdot\! \stepr\! \cdot,\initf)$ with notion of
observation $(\Omega,\filter{\cdot}{\cdot},\cdot\!\approx^i_{\cdot}\!\cdot)$
satisfies TINI if, for each $o\in\Omega$, there exist two
relations, \emph{indistinguishability of states to observer $o$}
(written \mbox{$s_1 \approx^s_o s_2$}) and \emph{observability of states
  to observer $o$} (written \mbox{$\filter{o}{s}$}), satisfying four
\emph{sanity conditions}
\begin{gather}
\label{eqn:sanity1}
\Implies{i_1 \approx^i_o i_2}{\Initf{i_1} \approx^s_o \Initf{i_2}}\\
\label{eqn:sanity2}
\Implies{s_1 \approx^s_o s_2}{s_2 \approx^s_o s_1}\\
\label{eqn:sanity3}
\Implies{s_1 \approx^s_o s_2}
        {(\filter{o}{s_1} \Leftrightarrow \filter{o}{s_2})}\\
\label{eqn:sanity4}
\Implies{(\filter{o}{\alpha} \Andip \step{s}{\alpha}{})}{\filter{o}{s}}
\end{gather}
and three \emph{unwinding conditions},
assuming \mbox{$s_1 \approx^s_o s_2$} and \mbox{$\step{s_1}{\alpha_1}{s_1'}$}:
\begin{gather}
\label{eqn:unwind1}
\Implies{(\filter{o}{s_1} \Andip \step{s_2}{\action_2}{s_2'})}
        {(\alpha_1 \approx^a_o \alpha_2 \Andip s_1' \approx^s_o s_2')}\\
\label{eqn:unwind2}
\Implies{(\neg\filter{o}{s_1} \Andip \neg\filter{o}{s_1'})}
        {s_1' \approx^s_o s_2}\\
\label{eqn:unwind3}
\Implies{(\neg\filter{o}{s_1} \Andip \filter{o}{s_1'} \Andip \filter{o}{s_2'}
          \Andip \step{s_2}{\alpha_2}{s_2'})}
        {s_1' \approx^s_o s_2'}
\end{gather}
\label{thm:unwinding}
\end{thm}

We outline the proof, which motivates each of the sanity and unwinding
conditions.  To prove TINI we must consider pairs of traces of machine
evaluations starting from initial states $\Initf{i_1}$ and
$\Initf{i_2}$ and show that, after filtering for observability, these
pairs of traces are indistinguishable.  For the proof, we also
maintain the invariant that the pairs of states reached by the two
evaluations are indistinguishable.  We are given
that \mbox{$i_1\approx^i_o i_2$}, so by \eqnref{eqn:sanity1} the initial
states are indistinguishable, as are the traces emitted so far (namely
$\tnil$).

Now suppose the two evaluations have arrived at two indistinguishable
states, \mbox{$s_1\approx^s_o s_2$}, and that the filtered traces emitted so
far are indistinguishable.  If $s_1$ can take a step,
\mbox{$\step{s_1}{\action_1}{s_1'}$}, what is possible for steps from $s_2$?
(We may assume that $\step{s_2}{\alpha_2}{s_2'}$: if no step is
possible from $s_2$ then we are already done because
\eqnref{eqn:indist-traces}, used in the definition of TINI, truncates
the trace from $s_1$ at this point.)  Proceed by cases on
observability of $s_1$.

Condition \eqnref{eqn:unwind1} says that, if $\filter{o}{s_1}$, then
the new states, $s_1'$ and $s_2'$, and the emitted traces remain
indistinguishable.

On the other hand, suppose $\neg\filter{o}{s_1}$; proceed by cases on
observability of $s_1'$.  \eqnref{eqn:unwind2} says that, if
$\neg\filter{o}{s_1'}$, then \mbox{$s_1'\approx^s_o s_2$}; and
by \eqnref{eqn:sanity4}, since $s_1$ is unobservable, $\action_1$ must be
unobservable, so the filtered emitted traces remain indistinguishable.

Finally, the case where $\neg\filter{o}{s_1}$ and $\filter{o}{s_1'}$.
Then $\neg\filter{o}{s_2}$ (by \eqnref{eqn:sanity3}), and $\alpha_1$ and
$\alpha_2$ are both unobservable by \eqnref{eqn:sanity4}.  Consider cases
on observability of $s_2'$. The filtered traces emitted up to $s_1'$
and $s_2'$ are indistinguishable, and if $\filter{o}{s_2'}$ we are
done by \eqnref{eqn:unwind3}. If $\neg\filter{o}{s_2'}$, we are in a
case symmetric to the paragraph above; by \eqnref{eqn:sanity2}
and \eqnref{eqn:unwind2} we have \mbox{$s_1\approx^s_o s_2'$}, and again
the filtered traces emitted up to these points are
indistinguishable.  \qed

\fi 

\paragraph*{TINI for abstract \iffull IFC \fi machine}
\label{para:abstract-tini}

\iffull
We now instantiate~\autoref{thm:TINI} with the
abstract machine defined in \autoref{sec:abstract}, showing it
satisfies TINI for the following notion of observation:
\fi
\begin{defn}[Observation for abstract machine]\label{defn:obs-abstr}
Let $\mathcal{L}$ be a lattice, with partial order $ \le $.
\iffull
For the abstract machine, events $\ottnt{n}  \mathord{\scriptstyle @}  \ottnt{L}$ are atoms; we define
indistinguishability of atoms, \mbox{$a_1\approx_o^{aa} a_2$}, as
in \eqnref{eqn:indist-events} above.
\else
Define indistinguishability of atoms, \mbox{$a_1\approx_o^{aa} a_2$} by
\ifmuchlater\bcp{this could be typeset in-line if needed for space}\fi
\begin{equation}\label{eqn:indist-atoms}
\xrule{}{a \approx_o^{aa} a}
\qquad
\xrule{\neg\filter{o}{a_1}\quad \neg\filter{o}{a_2}}{a_1 \approx_o^{aa} a_2}.
\end{equation}
\fi
The notion of observation \iffull for the abstract machine \fi is
$(\mathcal{L},\filter{\cdot}{\cdot}^{a},\cdot\approx_\cdot^{ia}\cdot)$,
where
\[
\begin{array}{rcl}
\filter{o}{\ottnt{n}  \mathord{\scriptstyle @}  \ottnt{L}}^{a} & \triangleq & \ottnt{L}   \le  o
\\
(p,\mathit{args_1},n,\ottnt{L}) \approx_o^{ia}
(p,\mathit{args_2},n,\ottnt{L}) & \triangleq &
\mathit{args_1}\lift{\approx_o^{aa}}\mathit{args_2}.
\end{array}
\]
\label{def:a-observation}%
(On the right-hand side of the second equation, $\lift{\approx_o^{aa}}$ is
indistinguishability of atoms, lifted to lists\iffull{} as in
\autoref{def:match-traces}\fi.)
\end{defn}

\iffull

To instantiate~\autoref{thm:TINI} we must exhibit relations of
observability and indistinguishability on states.  We outline these
definitions and the proofs of the sanity and unwinding conditions
here.

A state $s = \langle \mu\;[ \sigma ]\;\ottnt{pc}\rangle$ of the abstract
machine is observable by observer $o\in\mathcal{L}$, written
$\filter{o}{s}$, whenever $pc = \ottnt{n}  \mathord{\scriptstyle @}  L_{pc}$ is itself observable,
\IE $L_{pc}  \le  o$.

Indistinguishability of states is defined by two clauses: the first
for observable states (left), and the other for non-observable ones
(right).
\[
\xrule{\begin{array}{c}
        \filter{o}{\ottnt{pc}} \\
        \sigma_{{\mathrm{1}}} \lift{\approx_o^{aa}} \sigma_{{\mathrm{2}}} \qquad \mu_{{\mathrm{1}}} \lift{\approx_o^{aa}} \mu_{{\mathrm{2}}}
        \end{array}
        }
        {\mu_{{\mathrm{1}}} \; [ \sigma_{{\mathrm{1}}} ] \; \ottnt{pc}
         \approx_o^{sa}
        \mu_{{\mathrm{2}}} \; [ \sigma_{{\mathrm{2}}} ] \; \ottnt{pc}}
\quad
\xrule{\begin{array}{c}
        \lnot \filter{o}{\ottnt{pc_{{\mathrm{1}}}}} \quad \lnot\filter{o}{\ottnt{pc_{{\mathrm{2}}}}} \\
        \sigma_{{\mathrm{1}}} \sim_o^{a} \sigma_{{\mathrm{2}}}
         \qquad \mu_{{\mathrm{1}}} \lift{\approx_o^{aa}} \mu_{{\mathrm{2}}}
         \end{array}
}
        {\mu_{{\mathrm{1}}} \; [ \sigma_{{\mathrm{1}}} ] \; \ottnt{pc_{{\mathrm{1}}}}
        \approx_o^{sa}
        \mu_{{\mathrm{2}}} \; [ \sigma_{{\mathrm{2}}} ] \; \ottnt{pc_{{\mathrm{2}}}}}
\]
Here we abuse the notation of lifting, $\lift{\approx_o^{aa}}$, using
it for memories and stacks (two stack elements are indistinguishable
if they are indistinguishable atoms, or are both return stack frames,
with indistinguishable return addresses).

Let's have a more detailed look at the definition of state
indistinguishability. For observable states, we simply require that
all the state components be indistinguishable. For non-observable
ones, however, we must make the relation more permissive. Indeed, the
abstract IFC machine steps from an observable state to a
non-observable state when, e.g., branching on the value of a secret.
When that happens, the tight correspondence on states no longer
holds. Depending on the value of a secret, the machine could, e.g.,
jump to different instruction addresses, put different numbers of
values on its stack, perform more or fewer function calls,
etc. Because of that, we must allow states with different $\ottnt{pc}$
values to be related, and adopt a weaker indistinguishability relation
on stacks. This new relation, noted $\sim_o^a$, only is used when
relating unobservable states, and intuitively says that the stacks of
such states only need to be related up to the most recent return frame
to an observable one. Formally, $\sigma_{{\mathrm{1}}} \sim_o^{a} \sigma_{{\mathrm{2}}}$ is defined
as $\filter{o}{\sigma_{{\mathrm{1}}}} \lift{\approx_o^{aa}} \filter{o}{\sigma_{{\mathrm{2}}}}$,
where:
\begin{align*}
  \filter{o}{[\,]} & \triangleq [\,] \\
  \filter{o}{\ottnt{n}  \mathord{\scriptstyle @}  \ottnt{L}  \mathord{,}\,  \sigma} & \triangleq \filter{o}{\sigma} \\
  \filter{o}{\ottnt{n}  \mathord{\scriptstyle @}  \ottnt{L}  \mathord{;}\,  \sigma}
  & \triangleq
  \begin{cases}
    \ottnt{n}  \mathord{\scriptstyle @}  \ottnt{L}  \mathord{;}\,  \sigma & \text{if $\ottnt{L}  \le  o$} \\
    \filter{o}{\sigma} & \text{otherwise}
  \end{cases}
\end{align*}
In this way we relax the correspondence between call stacks of two
machines, while at the same time keeping the invariant that holds on
the ``observable'' part of the stacks, which we will need when proving
\autoref{eqn:unwind3} for the abstract machine.

\begin{thm}\label{thm:abstractIFCtini}
The relations $\filter{\cdot}{\cdot}^{a}$ and
$\cdot\approx_\cdot^{sa}\cdot$ satisfy the sanity and unwinding
conditions of~\autoref{thm:TINI}; thus, the abstract IFC machine has
TINI.
\end{thm}

\begin{proof}
  Most sanity conditions are easy consequences of the definitions, and
  do not require detailed explanation. We give an overview of the most
  interesting aspects of the proof; a more detailed account can be
  found in the formal development.

  The $ \text{\sf Output} $ instruction plays an important role for condition
  \eqnref{eqn:sanity4} and for the first conclusion of
  \eqnref{eqn:unwind1}. Crucially, since that instruction joins the
  label of the current $\ottnt{pc}$ to the output atom, an unobservable
  state necessarily produces an unobservable action. Further, when two
  low states are indistinguishable and step (\IE when they satisfy the
  preconditions of \eqnref{eqn:unwind1}), the atoms on top of the
  stack must be indistinguishable, leading to indistinguishable output
  actions.

  As for the second conclusion of \eqnref{eqn:unwind1}, since
  indistinguishable low states have equal $\ottnt{pc}$ values, they
  execute the same instructions. Thus, showing that the states remain
  indistinguishable after stepping is just a matter of reasoning about
  the values that are used by each instruction on both states. These
  values must be indistinguishable, and it is easy to show that
  storing them at the same locations in indistinguishable stacks and
  memories leads to stacks and memories that are still
  indistinguishable.

  Most of the cases of condition \eqnref{eqn:unwind2}---stepping from
  an unobservable state to another unobservable state---are trivial,
  since they only manipulate values or unobservable return frames on
  top of the stack (which by construction are irrelevant when checking
  whether the corresponding stacks are indistinguishable). The only
  exception is the $ \text{\sf Store} $ instruction, which also modifies the
  memory. Since the label on the $\ottnt{pc}$ is assumed to be above the
  level of the observer, the side condition of that instruction
  ensures that the same holds of the memory position being
  updated. This ensures that both memories remain indistinguishable,
  since the other positions are not affected.

  Finally, the precondition of \eqnref{eqn:unwind3} (stepping from
  unobservable to observable states) only applies when both states
  execute matching $ \text{\sf Ret} $ instructions. Since we assume that the
  resulting states are both observable, we conclude that the top of
  the original stacks contained the same observable return frame. The
  definition of indistinguishability says that the portions of the
  stacks below that frame are indistinguishable. Since those are
  exactly the values of the new stacks, and the returning $\ottnt{pc}$ is
  the same on both states, we conclude that the resulting observable
  states are indistinguishable.
\end{proof}

\else

We prove TINI for the abstract machine using a set of standard
{\em unwinding conditions}~\cite{GoguenM84,TestingNI}.
For this we need to define indistinguishability on states, and thus also
indistinguishability of stacks; this is where we encounter one subtlety.
Indistinguishability of stacks is defined pointwise when the label of
the PC is observable ($L_{pc}  \le  o$).
When the PC label is not observable, however, we only require that the
stacks are pointwise related below the most recent $ \text{\sf Call} $
from an observable state.
This is necessary because the two machines run in lock step only when
their PC labels are observable; they can execute completely different
instructions otherwise.

\begin{thm}\label{thm:abstractIFCtini}
The abstract IFC machine enjoys TINI.
\end{thm}

\fi

\paragraph*{TINI preserved by refinement}
\label{para:tini-refine}

\newcommand{\RARR}{\iffull\;\fi\Rightarrow\iffull\;\fi}

\begin{thm}[TINI preservation]\label{thm:tininPreserved}
Suppose that the generic machine \(M_2\) refines \(M_1\) by
refinement \((\match_i, \match_e)\) and that each machine is equipped
with a notion of observation. Suppose that, for all observers \(o_2\)
of \(M_2\), there exists an observer \(o_1\) of \(M_1\) such that the
following compatibility conditions hold \iffull: \else for all $e_1,
e_1'\in E_1$, all $e_2, e_2'\in E_2$, and all $i_2, i_2'\in I_2$.\fi
\iffull \begin{enumerate} \else \begin{inparaenum}[(i)] \fi
\item \label{thm:tininPreserved1}
      \iffull for all $e_1\in E_1$ and $e_2 \in E_2$,
      \fi
        $e_1 \match_e e_2 \RARR
        {(\filter{o_1}{e_1} \Leftrightarrow \filter{o_2}{e_2})}$\iffull\else; \fi
\item \label{thm:tininPreserved2}
      \iffull for all $i_2, i_2'\in I_2,\;
i_2 \approx^i_{o_2} i_2' \RARR
\exists\,i_1, i_1'\in I_1.\;
               i_1 \mathord{\approx^i_{o_1}} i'_1 \;\wedge\;
               i_1 \match_i i_2\;\wedge\;i'_1 \match_i i'_2$
\else
 $i_2 \mathord{\approx^i_{o_2}} i'_2 \RARR
\exists\, i_1 \mathord{\approx^i_{o_1}} i'_1 .\;
   (i_1 \match_i i_2\;\wedge\;i'_1 \match_i i'_2)$\fi
\iffull\else;\fi
\item \label{thm:tininPreserved3}
      \iffull for all $e_1,e_1'\in E_1$, and all $e_2,e'_2\in E_2$, \fi
 $(e_1\approx^a_{o_1} e'_1\;\wedge\; e_1 \match_e e_2 \;
   \wedge\; e'_1 \match_e e'_2)
     \RARR e_2 \approx^a_{o_2} e'_2$\iffull\else.\fi
\iffull \end{enumerate} \else \end{inparaenum} \fi{}
  Then, if \(M_1\) has TINI, \(M_2\) also has TINI.
\label{thm:tini-ref}
\end{thm}

\begin{proof}
  We include a brief proof sketch to convey the meaning of the theorem
  and the role of the compatibility conditions; intuitively, they say
  that $o_1$ does not have more observation power than $o_2$. Suppose
  that $o_2$ observes two traces $t_2$ and $t_2'$, starting from
  initial states $i_2$ and $i_2'$. We want to show that both traces
  are indistinguishable whenever the initial states are. By
  condition~\ref{thm:tininPreserved2}, we can find related initial
  states $i_1$ and $i_1'$ of $M_1$ that are indistinguishable. Since
  $M_2$ refines $M_1$, we know that these initial states produce
  traces $t_1$ and $t_1'$ that match $t_2$ and $t_2'$; furthermore,
  since $M_1$ has TINI, $t_1$ and $t_1'$ are indistinguishable. By
  condition~\ref{thm:tininPreserved1}, filtering related traces
  results in related traces; that is, $\filter{o_1}{t_1}
  \lift{\match_e} \filter{o_2}{t_2}$, and similarly for the other two
  traces. This implies, thanks to condition~\ref{thm:tininPreserved3},
  that we can use the indistinguishability of $\filter{o_1}{t_1}$ and
  $\filter{o_1}{t_1'}$ to argue that $\filter{o_2}{t_2}$ and
  $\filter{o_2'}{t_2'}$ are also indistinguishable, by a simple
  induction on the traces.
\end{proof}

Some formulations of noninterference are subject to the
\emph{refinement paradox}~\cite{Jacob89}, in which refinements of a
noninterferent system may violate noninterference.
We avoid this issue by employing a strong notion of noninterference
that restricts the amount of non-determinism in the system and is
thus preserved by {\em any} refinement (\autoref{thm:tininPreserved}).%
\footnote{The recent noninterference proof for the seL4
  microkernel~\cite{seL4:Oakland2013, MurrayMBGK12} works similarly
  (see \autoref{sec:relwork}).}
Since our abstract machine is deterministic, it is easy to show this
strong notion of noninterference for it.
In \autoref{sec:concl} we discuss a possible technique for generalizing to the
concurrent setting while preserving a high degree of determinism.

\paragraph*{TINI for concrete machine with IFC fault handler}
\label{para:tini-concrete}

It remains to define a notion of observation on the concrete machine,
instantiating the definition of TINI for this machine.  This
definition refers to a concrete lattice $\ottnt{CL}$, which must be a
correct encoding of an abstract lattice $\mathcal{L}$: the lattice
operators $\ottkw{genBot}$, $\ottkw{genJoin}$, and $\ottkw{genFlows}$
must satisfy the specifications in
\autoref{def:wf-clatt}.

\begin{defn}[Observation for the concrete machine]
Let $\mathcal{L}$ be an abstract lattice, and $\ottnt{CL}$ be correct
with respect to $\mathcal{L}$.  The observation for the concrete
machine is
$(\mathcal{L},\filter{\cdot}{\cdot}^{c},\cdot\approx_\cdot^{ic}\cdot)$,
where
\begin{gather*}
\filter{o}{\ottnt{n}  \mathord{\scriptstyle @}  \mathtt{T}}^{c} \;\triangleq\;
      \mathsf{Lab} ( \mathtt{T} )   \le  o,
\\
(p,\mathit{args_1'},n,\mathtt{T})
        \approx_o^{ic}
(p,\mathit{args_2'},n,\mathtt{T}) \;\triangleq\;
   \mathit{args_1}\lift{\approx_o^{aa}}\mathit{args_2},
\end{gather*}
and
$\mathit{args_i'}=
\mathsf{map~(\lambda (\ottnt{n}  \mathord{\scriptstyle @}  \ottnt{L}).~\ottnt{n}  \mathord{\scriptstyle @}   \mathsf{Tag} ( \ottnt{L} ) )}~\mathit{args_i}$.
\label{def:c-observation}
\end{defn}

Finally, we prove that the backward refinement proved in
\autoref{sec:qa-c-ref} \iffull(\autoref{thm:concRefinesAbs}) \fi satisfies the compatibility constraints
of~\autoref{thm:tininPreserved}, so we derive\iffull{} the main
result\fi:
\begin{thm}\label{thm:concTINI}
The concrete IFC machine running the fault handler
$\tiniFH$ satisfies TINI.
\end{thm}

\section{An Extended System}
\label{sec:extensions}

\begin{figure}[tbp!]
\ottgrammartabular{
\ottnewinstr\ottinterrule
}
\vspace*{-3ex}
\caption{Additional instructions for extensions}
\label{fig:newinstrs}
\end{figure}

\begin{figure}
\ottusedrule{\ottdruleEStepAlloc{}}
\ottusedrule{\ottdruleEStepSizeOf{}}
\ottusedrule{\ottdruleEStepGetOff{}}
\ottusedrule{\ottdruleEStepEq{}}
\ottusedrule{\ottdruleEStepSysCall{}}
\caption{Semantics of selected new abstract machine instructions}
\label{fig:extAbstractSteps}
\end{figure}

\begin{figure}
\ottusedrule{\ottdruleCEAlloc{}}
\ottusedrule{\ottdruleCEAllocP{}}
\ottusedrule{\ottdruleCEPushCachePtr{}}
\ottusedrule{\ottdruleCEUnpack{}}
\ottusedrule{\ottdruleCEPack{}}
\ottusedrule{\ottdruleCESysCall{}}
\caption{Semantics of selected new concrete machine instructions}
\label{fig:extConcreteSteps}
\end{figure}

Thus far we have presented our methodology in the context of a simple
machine architecture and IFC discipline. We now show how it can be
scaled up to a significantly more sophisticated setting, where the
basic machine is extended with a \emph{frame-based} memory model
supporting \emph{dynamic allocation} and a \emph{system call}
mechanism for adding special-purpose primitives. Building on these
features, we define an abstract IFC machine that uses \emph{sets of
  principals} as its labels and a corresponding concrete machine
implementation where tags are pointers to dynamically allocated
representations of these sets.  While still much less complex than the
real SAFE system, this extended model serves as good evidence of the
robustness our approach, and how it might apply to more realistic
designs: The new features were added by incrementally adapting the Coq
formalization of the basic system, without requiring any major changes
to the initial proof architecture.

\autoref{fig:newinstrs} shows the new instructions supported by the
extended model. Instruction $\ottkw{PushCachePtr}$, $\ottkw{Unpack}$, and
$\ottkw{Pack}$ are used only by the concrete machine, for the compiled
fault handler (hence they only have a kernel-mode stepping rule; they
simply get stuck if executed outside kernel mode, or on an abstract
machine).  We also add two stack-manipulation instructions, $ \text{\sf Dup} $
and $ \text{\sf Swap} $, to make programming the kernel routines more
convenient. It remains true that any program for the abstract machine
makes sense to run on the abstract rule machine and the concrete
machine.  For brevity, we detail stepping rules only for the extended
abstract IFC machine (\autoref{fig:extAbstractSteps}) and concrete
machine (\autoref{fig:extConcreteSteps}); corresponding extensions to
the symbolic IFC rule machine are straightforward (we also omit rules
for $ \text{\sf Dup} $ and $ \text{\sf Swap} $). Individual rules are explained below.

\subsection{Dynamic memory allocation}
High-level programming languages usually assume a structured memory model,
in which independently allocated \emph{frames} are disjoint by
construction and programs cannot depend on the relative placement of
frames in memory.
The SAFE hardware enforces this abstraction by attaching explicit
runtime types to all values, distinguishing pointers from other
data. Only data marked as pointers can be used to access memory. To
obtain a pointer, one must either call the (privileged) memory manager
to allocate a fresh \emph{frame} or else offset an existing
pointer. In particular, it is not possible to ``forge'' a pointer from
an integer.  Each pointer also carries information about its base and
bounds, and the hardware prevents it from being used to access memory
outside of its frame.

\paragraph{Frame-based memory model} In our extended system,
we model the user-level view of SAFE's memory
system by adding a frame-structured memory%
\iffull{} (similar to~\cite{LeroyB08})\fi%
, distinguished pointers (so \emph{values}, the payload field of atoms
and the tag field of concrete atoms, can now either be an integer
$(\ottkw{Int} \, \ottnt{n})$ or a pointer $(\ottkw{Ptr} \, p)$), and an allocation
instruction to our basic machines.  We do this (nearly) uniformly at
all levels of abstraction.\footnote{It would be interesting to
describe an \emph{implementation} of the memory manager in a
still-lower-level concrete machine with no built-in $\ottkw{Alloc}$
instruction, but we leave this as future work.}  A \emph{pointer} is a
pair \(p=(id,o\)) of a frame identifier \(id\) and an offset \(o\)
into that frame.  In the machine state, the data memory $\mu$ is a
partial function from pointers to individual storage cells that is
undefined on out-of-frame pointers.  By abuse of notation, $\mu$ is
also a partial function from frame identifiers to frames, which are
just lists of atoms.

The most important new rule of the extended abstract machine is
$\ottkw{Alloc}$ (\autoref{fig:extAbstractSteps}).  In this machine there
is a separate memory region (assumed infinite) corresponding to each
label.  The auxiliary function $ \text{\sf alloc} $ in the rule for
$\ottkw{Alloc}$ takes a size \(k\), the label (region) at which to
allocate, and a default atom $a$; it extends $\mu$ with a fresh
frame of size $k$, initializing its contents to $a$. It returns
the id of the new frame and the extended memory $\mu'$.

\paragraph{IFC and memory allocation} We require that the frame
identifiers produced by allocation at one label not be affected by
allocations at other labels; \EG $ \text{\sf alloc} $ might allocate
sequentially in each region.  Thus, indistinguishability of low atoms
is just syntactic equality, preserving \autoref{defn:obs-abstr} from
the simple abstract machine, which is convenient for proving
noninterference, as we explain below.  We allow a program to observe
frame sizes using a new $\ottkw{SizeOf}$ instruction, which requires
tainting the result of $\ottkw{Alloc}$ with $\ottnt{L}$, the label of the size
argument.  There are also new instructions $\ottkw{Eq}$, for comparing two
values (including pointers) for equality, and $\ottkw{GetOff}$, for
extracting the offset field of a pointer into an integer.  However,
frame ids are intuitively \emph{abstract}: the concrete representation
of frame ids is not accessible, and pointers cannot be forged or
output.  The extended concrete machine stepping rules for these new
instructions are analogous to the abstract machine rules, with the
important exception of $\ottkw{Alloc}$, which is discussed below.

A few small modifications to existing instructions in the basic machine
(\autoref{fig:abstractSteps}) are needed to handle pointers properly.  In
particular:
\begin{inparaenum}[(i)]
\item $ \text{\sf Load} $ and $ \text{\sf Store} $ require pointer arguments
and get stuck if the pointer's offset is out of range for its frame.
\item $ \text{\sf Add} $ takes either two integers or an integer and a pointer,
  where $\ottkw{Int} \, \ottnt{n}+\ottkw{Int} \, \ottnt{m}=\ottkw{Int} \, \ottsym{(}  \ottnt{n}  \mathord{+}  \ottnt{m}  \ottsym{)}$ and \(\ottkw{Ptr} \, \ottsym{(}  \ottnt{id}  \mathord{,}\,  o_{{\mathrm{1}}}  \ottsym{)}+\ottkw{Int} \, o_{{\mathrm{2}}}=\ottkw{Ptr} \, \ottsym{(}  \ottnt{id}  \mathord{,}\,  o_{{\mathrm{1}}}  \mathord{+}  o_{{\mathrm{2}}}  \ottsym{)}\).
\item $ \text{\sf Output} $ works only on integers, not pointers.
\end{inparaenum}
Analogous modifications are needed in the concrete machine semantic
rules.

\paragraph{Concrete allocator} The extended concrete machine's
semantics for $\ottkw{Alloc}$ differ from those of the abstract machine in
one key respect.  Using one region per tag would not be a realistic
strategy for a concrete implementation; e.g., the number of different
tags might be extremely large.
Instead, we use a single region for all user-mode allocations at the
concrete level. We also collapse the separate user and kernel memories
from the basic concrete machine into a single memory. Since we still
want to be able to distinguish user and kernel frames, we mark each
frame with a privilege mode (i.e., we use two allocation regions).
\autoref{fig:extConcreteSteps} shows the corresponding concrete
stepping rule for $\ottkw{Alloc}$ for two cases: non-faulting user mode
and kernel mode.
The rule cache is now just a distinguished kernel frame $ cache $;
to access it, the fault handler uses the (privileged)
$\ottkw{PushCachePtr}$ instruction.
The concrete $ \text{\sf Load} $ and $ \text{\sf Store} $ rules are modified to prevent
dereferencing kernel pointers in user mode. These checks are only
needed if we want to allow user-level code to manipulate kernel
pointers directly while protecting the data structures they point
to. For instance, we could allow certain operations on pointers
representing labels, such as taking the join of two labels, while
preserving noninterference. If kernel pointers cannot be ``leaked''
into user data (as in~\autoref{para:sets-of-prins}), these checks can
be safely omitted, since user-level code won't be able to tamper with
kernel data.

\paragraph{Proof by refinement} As before, we prove noninterference
for the concrete machine by combining a proof of noninterference of
the abstract machine with a two-stage proof that the concrete machine
refines the abstract machine. By using this approach we avoid some
well-known difficulties in proving noninterference directly for the
concrete machine. In particular, when frames allocated in low and high
contexts share the same region, allocations in high contexts can cause
variations in the precise pointer values returned for allocations in
low contexts, and these variations must be taken into account when
defining the indistinguishability relation. For example, Banerjee and
Naumann~\cite{BanerjeeN05} prove noninterference by parameterizing
their indistinguishability relation with a partial bijection that
keeps track of indistinguishable memory addresses.  Our approach, by
contrast, defines pointer indistinguishability only at the abstract
level, where indistinguishable low pointers are identical.  This proof
strategy still requires relating memory addresses when showing
refinement, but this relation does not appear in the noninterference
proof at the abstract level. The refinement proof itself uses a
simplified form of \emph{memory injections}~\cite{LeroyB08}.  The
differences in the memory region structure of both machines are
significant, but invisible to programs, since no information about
frame ids is revealed to programs beyond what can be obtained by
comparing pointers for equality.  This restriction allows the
refinement proof to go through straightforwardly.

\subsection{System calls}
To support the implementation of policy-specific primitives on top of the
concrete machine, we provide a new \emph{system call} instruction.
The $\text{\sf SysCall} \, \ottnt{id}$ instruction is parameterized by a system call
identifier.  The step relation of each machine is now parameterized by
a table $T$ that maps system call identifiers to their
implementations.

In the abstract and symbolic rule machines, a system call
implementation is an arbitrary Coq function that removes a list of
atoms from the top of the stack and either puts a result on top of the
stack or fails, halting the machine. The system call implementation is
responsible for computing the label of the result and performing any
checks that are needed to ensure noninterference.

In the concrete machine, system calls are implemented by kernel
routines and the call table contains the entry points of these
routines in the kernel instruction memory.  Executing a system call
involves inserting the return address on the stack (underneath the
call arguments) and jumping to the corresponding entry point. The
kernel code terminates either by returning a result to the user
program or by halting the machine.

This feature has no major impact on the proofs of noninterference and
refinement. For noninterference, we must show that all the abstract system
calls preserve indistinguishability of abstract machine states; for
refinement, we show that each concrete system call correctly
implements the abstract one using the machinery
of~\autoref{sec:fault-handler-correct}.

\subsection{Labeling with sets of principals}
\label{para:sets-of-prins}
The full SAFE machine supports dynamic creation of security
principals.  In the extended model, we make a first step toward
dynamic principal creation by taking principals to be integers and
instantiating the (parametric) lattice of labels with the lattice of
finite sets of integers. This lattice is statically known, but models
dynamic creation by supporting unbounded labels and having no top
element. In this lattice, $ \bot $ is $\emptyset$, $ \mathord{\vee} $ is
$\cup$, and $ \le $ is $\subseteq$.  We enrich our IFC model by
adding a new \emph{classification} primitive $ \text{\sf joinP} $ that adds a
principal to an atom's label, encoded using the system call mechanism
described above.  The operation of $ \text{\sf joinP} $ is given by the
following derived rule, which is an instance of the $ \text{\sf SysCall} $ rule
from \autoref{fig:extAbstractSteps}.
\ottusedrule{\ottdruleEStepjoinP{}}

At the concrete level, a tag is now a pointer to an array of
principals (integers) stored in kernel memory.  To keep the fault
handler code simple, we do not maintain canonical representations of
sets: one set may be represented by different arrays, and a given
array may have duplicate elements. (As a consequence, the mapping from
abstract labels to tags is no longer a function; we return to this
point below.)  Since the fault handler generator in the basic system
is parametric in the underlying lattice, it doesn't
require any modification. All we must do is provide concrete
implementations for the appropriate lattice operations: $\ottkw{genJoin}$
just allocates a fresh array and concatenates both argument arrays
into it; $\ottkw{genFlows}$ checks for array inclusion by iterating
through one array and testing whether each element appears in the
other; and $\ottkw{genBot}$ allocates a new empty array.  Finally,
we provide kernel code to implement $ \text{\sf joinP} $, which
requires two new privileged instructions, $\ottkw{Pack}$ and $\ottkw{Unpack}$
(\autoref{fig:extConcreteSteps}), to manipulate the payload and tag
fields of atoms; otherwise, the implementation is similar to that of
$\ottkw{genJoin}$.

A more realistic system would keep canonical representations of sets
and avoid unnecessary allocation in order to improve its memory
footprint and tag cache usage.  But even with the present simplistic
approach, both the code for the lattice operations and their proofs of
correctness are significantly more elaborate than for the trivial
two-point lattice. In particular, we need an additional code generator
to build counted loops, e.g., for computing the join of two tags.
\[
\begin{array}{rcl}
\ottkw{genFor}\ c\ & = &
[ \ottkw{Dup} ]\  {\scriptstyle \mathord{+\!+} } \ \ottkw{genIf}\ (\ottkw{genLoop} (c\  {\scriptstyle \mathord{+\!+} } \ [ \text{\sf Push} \, \ottsym{(}  \ottsym{-}  \ottsym{1}  \ottsym{)},\ottkw{Add} ]))\ [] \\
& & \mbox{\rm where}\ \ottkw{genLoop}\ c\ =\ c\  {\scriptstyle \mathord{+\!+} } \
[ \ottkw{Dup},\ottkw{Bnz}\ (-(\ottkw{length}\ c + 1)) ]
\end{array}
\]
Here, $c$ is a code sequence representing the loop body, which is
expected to preserve an index value on top of the stack; the generator
builds code to execute that body repeatedly, decrementing the index
each time until it reaches 0.  The corresponding specification is
\ottusedrule{\ottdruleRepeat{}}

To avoid reasoning about memory updates as far as possible, we code in
a style where all local context is stored on the stack and manipulated
using $ \text{\sf Dup} $ and $ \text{\sf Swap} $. Although the resulting code is
lengthy, it is relatively easy to automate the corresponding proofs.

\paragraph{Stateful encoding of labels}
Changing the representation of tags from integers to pointers requires
modifying one small part of the basic system proof.  Recall that
in \autoref{sec:fault-handler} we described the encoding of labels
into tags as a \emph{pure} function $ \mathsf{Lab} $.  To deal with the
memory-dependent and non-canonical representation of sets described
above, the extended system instead uses a \emph{relation} between an
abstract label, a concrete tag that encodes it, and a memory in which
this tag should be interpreted.

If tags are pointers to data structures, it is crucial that these data
structures remain intact as long as the tags appear in the machine
state.  We guarantee this by maintaining the very strong invariant
that each execution of the fault handler only allocates new frames,
and never modifies the contents of existing ones, except for the
$ cache $ frame (which tags never point into).  A more realistic
implementation might use mutable kernel memory for other purposes and
garbage collect unused tags; this would require a more complicated
memory invariant.

The TINI formulation is similar in essence to the one in
\autoref{sec:ni}, but some subtleties arise for concrete output
events, since their tags cannot be interpreted on their own
anymore. We wish to
\begin{inparaenum}[(i)]
\item keep the semantics of the concrete machine independent of high-level policies such as IFC
  and
\item give a statement of noninterference that does not refer to
  pointers.
\end{inparaenum}
To achieve these seemingly contradictory aims, we model an event of
the concrete machine as a pair of a concrete atom plus the whole
state of the kernel memory.
This memory is not visible to observers in the formulation of TINI,
but instead determines which events' payloads they are able to
observe. This is done by extending our notion of observation with a
function that interprets every concrete event present in the output
trace in higher-level terms. This interpretation abstracts away from
low-level representation issues, such as the layout of data structures
in memory, and allows us to give a more natural definition of event
indistinguishability in the formulation of TINI. For instance, in the
extended system described above, the interpretation of a pointer tag
is the set of principals that that pointer represents in kernel
memory---that is, the contents of the array it points to. This allows
us to define the event indistinguishability relation by simple
equality.

Our model of observation in terms of an interpretation function is an
idealization of what happens in the real SAFE machine, where
communication of labeled data with the outside world involves
cryptography\ch{wishful thinking}.  Extending this model this is left
as future work.

\section{Related Work}
\label{sec:relwork}

The SAFE design spans a number of research areas, and a comprehensive overview
of related work would be huge.  We focus here on a small set of especially
relevant points of comparison.

\paragraph{Language-based IFC}

Static approaches to IFC have generally dominated language-based security
research~\citeEtcShort{
  sabelfeld03:lang_based_security%
\iffull,%
volpano96:IF,%
myers:decentr_label_model,%
pottier03:flowcaml%
\fi}; however,
statically enforcing IFC at the lowest level of a real system is
challenging.
Soundly analyzing native binaries with reasonable
precision is hard\iffull{} (static IFC for low-level code
  usually stops at the bytecode level~\citeEtcShort{BarthePR07%
  \iffull,Grabowski12,MannS12,HammerS09\fi})\fi,
even more so without the compiler's
cooperation (\EG for stripped or obfuscated binaries).
Proof-carrying
code~\citeEtcShort{BarthePR07\iffull,Grabowski12,BartheCGJP07\fi} and
typed assembly language~\citeEtcShort{MedelCB05\iffull,YuI06,Yu07\fi}
have been used for enforcing IFC on low-level code
without low-level analysis or adding the compiler to the TCB.
\ifmuchlater
\ch{Does \cite{BartheRN06} fit in the story above?
  Any other important references for PCC/TAL/bytecode?}
\fi
In SAFE~\cite{SAFEPLOS11, interlocks_ahns2012} we follow a
different approach, enforcing noninterference using purely dynamic checks,
for arbitrary binaries in a custom-designed instruction set.
The mechanisms we use for this are similar to those found in recent
work on purely dynamic IFC for high-level languages~\citeEtcShort{AustinF09,
  RussoS10, StefanRMM11, HedinS12, Exceptional%
\iffull,%
  Guernic07,%
  Pistoia:2007,%
  Shinnar:2009,%
  GuernicBJS06,%
  MooreC11,%
  AustinF:2010,%
  AustinF12,%
  AustinFA12,%
  SabelfeldR09,%
  AskarovS09b%
\fi}; however, as far as we
know, we are the first to push these ideas to the lowest level.
%

\paragraph{seL4}

Murray~\ETAL\cite{seL4:Oakland2013} recently demonstrated a
machine-checked noninterference proof for the implementation of the
seL4 microkernel.
This proof is carried out by refinement and reuses
the specification and most of the
existing functional correctness proof of
seL4~\cite{Klein09sel4:formal}.
Like the TINI property in this paper, the variant of intransitive
noninterference used by Murray~\ETAL is preserved by refinement
because it implies a high degree of determinism~\cite{MurrayMBGK12}.
This organization of their proof was responsible
for a significant saving in effort, even when factoring in the
additional work required to remove all observable non-determinism from
the seL4 specification.
Beyond these similarities, SAFE and seL4 rely on completely different
mechanisms to achieve different notions of noninterference (seL4
admits intransitive IFC policies, capturing the ``where'' dimension of
declassification~\cite{sabelfeld05:_dimensions_declass}, while we
consider transitive ones).
Whereas, in SAFE, each word of data has an IFC label and labels are
propagated on each instruction, the seL4 kernel maintains separation between several
large partitions (\EG one partition can run an unmodified version of
Linux) and ensures that information is conveyed between such
partitions only in accordance with a fixed access control policy.

\paragraph{PROSPER}
In parallel work,
Dam~\ETAL\citeEtcShort{DamGKNS13,KhakpourSD13\iffull,DamGN13\fi}
verified information flow security for a tiny proof-of-concept
separation kernel running on ARMv7 and using a Memory Management Unit
for physical protection of memory regions belonging to different
partitions.
The authors argue that noninterference is not well suited for systems
in which components are supposed to communicate with each other.
Instead, they use the bisimulation proof method to show trace
equivalence between the real system and an ideal top-level
specification that is secure by construction.
As in seL4~\cite{seL4:Oakland2013}, the proof methodology precludes an
abstract treatment of scheduling, but the authors contend this is
to be expected when information flow is to be taken into
account.
In more recent work, \citet{Balliu14} propose a symbolic
execution-based information flow analysis for machine code, and use
this technique to verify a separation kernel system call handler, a
UART device driver, and a crypto service modular exponentiation
routine.


\paragraph{TIARA and ARIES}

The SAFE architecture embodies a number of innovations from earlier
paper designs.  In particular, the TIARA
design~\cite{shrobe09:tiara_nicecap_report} first proposed the idea of
a zero-kernel operating system and sketched a concrete architecture,
while the ARIES project proposed using a hardware rule cache to speed
up information-flow tracking~\cite{brown01:aries_ifc}.
In TIARA and ARIES, tags had a fixed set of fields
and were of limited length, whereas,
in SAFE, tags are pointers to arbitrary data
structures, allowing them to represent complex IFC labels
encoding sophisticated security policies~\cite{GenLabels}\iffull, for instance
decentralized ones~\cite{myers:decentr_label_model\iffull,StefanRMM11dc\fi}\fi.
Moreover, unlike TIARA and ARIES, which made no formal soundness claims, SAFE
proposes a set of IFC rules aimed at
achieving noninterference; the proof we present in this paper, though for
a simplified model, provides evidence that this goal is feasible.

\paragraph{RIFLE and other binary-rewriting-based IFC systems}

RIFLE \cite{rifle_micro2004} enforces user-specified
information-flow policies for x86 binaries using binary rewriting,
static analysis, and augmented hardware.
Binary rewriting is used to make implicit flows explicit; it heavily
relies on static analysis for reconstructing the program's
control-flow graph and performing reaching-definitions and alias
analysis.
The augmented hardware architecture associates labels with registers
and memory and updates these labels on each instruction to track
explicit flows.
Additional security registers are
used by the binary translation mechanism to help track implicit flows.
Beringer~\cite{Beringer2012} recently proved (in Coq) that the main
ideas in RIFLE can be used to achieve noninterference for a simple
While language.
Unlike RIFLE, SAFE achieves noninterference purely dynamically and
does not rely on binary rewriting or heroic static analysis of
binaries.
Moreover, the SAFE hardware is generic, simply caching instances
of software-managed rules.

While many other information flow tracking systems based on binary
rewriting have been proposed, few are concerned with
soundly handling implicit flows~\cite{MasriPL04,ClauseLO07}, and even
these do so only to the extent they can statically analyze binaries.
Since, unlike RIFLE (and SAFE), these systems use unmodified hardware,
the overhead for tracking implicit flows can be large.
To reduce this overhead, recent systems track implicit flows
selectively~\cite{KangMPS11} or not at
all\iffull~\cite{JeePKGAK2012,lift_micro2006}\fi---arguably a reasonable tradeoff
in settings such as malware analysis or attack detection,
where speed and precision are more important than soundness.

\paragraph{Hardware taint tracking}

The last decade has seen significant progress in specialized
hardware for accelerating taint
tracking%
~\citeEtcShort{secure_flow_track_asplos2004, raksha_isca2007, flexitaint,
deng_dsn2012%
\iffull,%
  deng_micro2010,%
  CrandallC04,%
  pointer_taintedness_dsn2005%
\fi}.
Most commonly, a single tag bit is associated with each word to
specify if it is tainted or not.
Initially aimed at mitigating low-level memory corruption attacks by
preventing the use of tainted pointers and the execution of tainted
instructions~\citeEtcShort{secure_flow_track_asplos2004%
\iffull,CrandallC04,pointer_taintedness_dsn2005\fi},
hardware-based taint tracking has also
been used to prevent high-level attacks such as SQL injection and
cross-site scripting~\cite{raksha_isca2007}.
In contrast to SAFE, these systems prioritize efficiency and overall
helpfulness over the soundness of the analysis, striking a heuristic
balance between false positives and false negatives (missed attacks).
As a consequence, these systems ignore implicit flows and
often do not even track all explicit flows.
While early systems supported a single hard-coded taint propagation
policy, recent ones allow the policy to be defined in
software~\cite{raksha_isca2007,flexitaint, deng_dsn2012} and support
monitoring policies that go beyond taint
tracking~\citeEtcShort{deng_dsn2012,ChenKSFGMRRRV08\iffull,deng_micro2010,RuwaseGMRCKR08\fi}.
Harmoni~\cite{deng_dsn2012}, for example,
provides a pair of caches that are quite similar to the SAFE rule
cache.  Possibly these could even be adapted to enforcing
noninterference, in which case we expect the proof methodology
introduced here to apply.

\iffull
\paragraph{Timing and termination}
Our TINI property ignores both termination and timing: a program that
diverges, fails, or takes varying amounts of time to run based on a
sensitive input is considered secure.
The full SAFE design includes a clearance-based access-control
mechanism~\cite{StefanRMM11} for addressing termination and timing
covert channels (\IE high-bandwidth channels through which malicious
code can exfiltrate secrets it directly has access to).
Stefan~\ETAL\cite{StefanRBLMM12} have also shown that in a concurrent
setting such leaks can be prevented by an adapted IFC mechanism,
at the risk of spawning very large numbers of threads.
We believe that this IFC mechanism could also be enforced
using the hardware mechanisms we describe here.
A recently proposed technique for instruction-based
scheduling~\cite{StefanBYLTRM13\iffull,BuirasLSRM13\fi}
is aimed at preventing leaks via the
internal timing side-channel (\EG malicious code sharing the same
processor inferring secrets through timing variations arising from
cache misses).
This could probably be adapted to SAFE, and since the SAFE processor
is very simple the mitigation could work well~\cite{CockGMH14}.
Finally, several mechanisms have been proposed for mitigating
the external timing side-channel (\IE leakage of secrets to an
attacker making timing observations over the network) and thus
reducing the rate at which bits can be
leaked~\citeEtcShort{ZhangAM11\iffull,AskarovZM10\fi}.
We do not consider any of these attacks or mitigations in this work.

\fi

\paragraph{Verification of low-level code}

The distinctive challenge in verifying machine code is coping with
unstructured control flow.  Our approach using structured generators
to build the fault handler is similar to the mechanisms used in
Chlipala's Bedrock system~\cite{Chlipala11,Chlipala2013} and by
Jensen~\ETAL\cite{JensenBK13}, but there are several points of
difference.  These systems each build macros on top of a powerful
low-level program logic for machine code (Ni and Shao's
XCAP~\cite{NiS06}, in the case of Bedrock), whereas we take a simpler,
ad-hoc approach, building directly on our stack machine's relatively
high-level semantics.  Both these systems are based on separation
logic, which we can do without since (at least in the present
simplified model) we have very few memory operations to reason about.
We have instead focused on developing a simple Hoare logic
specifically suited to verifying structured runtime-system code; e.g.,
we omit support for arbitrary code pointers, but add support for
reasoning about termination.  We use total-correctness Hoare triples
(similar to Myreen and Gordon~\cite{MyreenG07}) and weakest
preconditions to guarantee progress, not just safety, for our handler
code.  Finally, our level of automation is much more modest than
Bedrock's, though still adequate to discharge most verification
conditions on straight-line stack manipulation code rapidly and often
automatically.

\iffull
\paragraph{Work on testing noninterference}
The abstract machine in \autoref{sec:abstract} was proposed by
Hri\c{t}cu~\ETAL\cite{TestingNI}, extended in this work with dynamic
allocation and data classification (\autoref{sec:extensions}), and
recently further extended by Hri\c{t}cu~\ETAL to a sophisticated
machine featuring a highly permissive flow-sensitive dynamic
enforcement mechanism, public labels, and
registers~\cite{jfpsubmission2014}. While the focus of that work is on
verifying noninterference by random testing, it also shows how to use
invariants discovered during testing to formalize proofs of
noninterference in Coq.

Although the abstract machine and IFC mechanism considered here are
simpler than the most complex ones of \citet{jfpsubmission2014}, our
main concerns are the {\em concrete} machine, the IFC fault handler,
and the key properties of this combination, all of which are novel. We
believe nevertheless that our methodology could be extended to that
setting as well, verifying an implementation of this extended IFC
machine by a lower-level one. Depending on the hardware capabilities
at the lower level, some of the features of the machine could have to
be implemented in software, requiring further proofs. For instance,
this extended IFC machine still relies on a protected stack for
soundly performing function calls and returns: on a call, the entire
register file is stored on this stack, so that it can be restored upon
a return, thereby preventing data leakage. At the lowest level, this
protected stack could be implemented with a regular stack living in
kernel space, managed through special system calls.

\paragraph{Tagging hardware beyond IFC}
Although the tagging mechanism we discuss arose in the context of the
SAFE system, and was primarily designed for information-flow control,
it is sufficiently generic to be implemented in other architectures
and to enforce more security policies.

In follow-on work, Dhawan~\ETAL\cite{pump_asplos2015} adapt the
tagging mechanism to a more conventional RISC processor, using it to
implement policies such as memory safety and control-flow
integrity. They evaluate the performance of the mechanism on benchmark
simulations, which indicate a modest impact on speed (typically under
10\%) and power ceiling (less that 10\%), even when enforcing multiple
policies simultaneously.

Azevedo de Amorim~\ETAL\cite{micropolicies2015} use Coq to formalize a
generic version of the symbolic machine of~\autoref{sec:quasi}; that
machine is different from the one discussed here in that it is based
on a more conventional processor design (\EG with registers instead
of a protected stack), and serves as a high-level substrate for
programming many different security policies,
including compartmentalization and memory
safety. Finally, they formulate the intended effect of each policy as
a security property, using formal proofs to show that each policy
enforces the corresponding property.

A recent project at Draper Labs~\cite{Dover16} is working to extend the
RISC-V processor with tag propagation hardware in the style of the SAFE
processor.  As of March 2016, a prototype able to boot Linux is running on
FPGA boards.

\section{Conclusions and Future Work}
\label{sec:concl}

We have presented a formal model of the key IFC mechanisms of the SAFE
system: propagating and checking tags to enforce security, using a
hardware cache for common-case efficiency and a software fault handler
for maximum flexibility. To formalize and prove properties at such a
low level (including features such as dynamic memory allocation and
labels represented by pointers to in-memory data structures), we first
construct a high-level abstract specification of the system, then
refine it in two steps into a realistic concrete machine.  A
bidirectional refinement methodology allows us to prove (i) that the
concrete machine, loaded with the right fault handler (i.e. correctly
implementing the IFC enforcement of the abstract specification)
satisfies a traditional notion of termination-insensitive
noninterference, and (ii) that the concrete machine reflects all the
behaviors of the abstract specification. Our formalization reflects
the programmability of the fault handling mechanism, in that the fault
handler code is compiled from a rule table written in a small DSL. We
set up a custom Hoare logic to specify and verify the corresponding
machine code, following the structure of a simple compiler for this
DSL.

The development in this paper concerns three {\em deterministic}
machines and simplifies away {concurrency}.  While the lack of
concurrency is a significant current limitation that we would like to
remove by moving to a multithreading single-core
model, we still want to maintain the abstraction layers of a
proof-by-refinement architecture.
This requires some care so as not to run afoul of the refinement
paradox~\cite{Jacob89} since some standard notions of noninterference (for
example possibilistic noninterference) are not preserved by refinement in
the presence of non-determinism.
One promising path toward this objective is inspired by the recent
noninterference proof for seL4~\cite{seL4:Oakland2013,MurrayMBGK12}.
If we manage to share a common thread scheduler between the abstract
and concrete machines, we could still prove a strong double refinement
property (concrete refines abstract and vice versa) and hence preserve
a strong notion of noninterference (such as the TINI notion from this
work) or a possibilistic variation.

\ifmuchlater
\ch{We should probably mention the size of the Coq development
  [somewhere] as a proxy for effort (what other proxies are there?
  can we estimate how many person-months we have?)}  \ch{Also discuss
  about the extensions, and how easy/hard they were}
\fi

Although this paper focuses on IFC and noninterference, the tagging
facilities of the concrete machine are completely generic and have
been used since to enforce completely different properties like memory
safety, compartment isolation, and control-flow
integrity~\cite{micropolicies2015}.
Moreover, although the rule cache / fault handler design arose in the
context of SAFE, it has since been adapted to a conventional RISC
processor~\cite{pump_asplos2015}.
%

\paragraph*{Acknowledgments}
We are grateful to
Maxime D{\'e}n{\`e}s,
Deepak Garg,
Greg Morrisett,
Toby Murray,
Jeremy Planul,
Alejandro Russo,
Howie Shrobe,
Jonathan M. Smith,
Deian Stefan, and
Greg Sullivan
for useful discussions and helpful feedback on early drafts.
We also thank the anonymous reviewers for their insightful comments.
%
%
This material is based upon work supported by the DARPA CRASH and SOUND
programs through the US Air Force Research Laboratory (AFRL) under Contracts
No. FA8650-10-C-7090 and FA8650-11-C-7189.
This work was also partially supported by NSF award 1513854
  {\em Micro-Policies: A Framework for Tag-Based Security Monitors}.
The views expressed are those of the authors and do
not reflect the official policy or position of the Department of Defense or
the U.S. Government.

\bibliographystyle{abbrvnaturl} 
\bibliography{safe,local}

\ifapp
\appendix

\fi

\ifscr
\appendix
\section{Scratchpad}

\subsection{Friendly readers}

{\bf Note: this list can't contain people we propose as external
  reviewers, so we need to be careful here}

\begin{itemize}
\item (SENT) Rest of SAFE team + Howie
\item (SENT) SafeQC folks:
      Dimitrios Vytiniotis,
      John Hughes
\item (SENT) Idaho:
      James Alves-Foss (jimaf@uidaho.edu),
      Xia Yang (xyang@uidaho.edu)
\item (SENT) Robert Jakob (jakobro@informatik.uni-freiburg.de),
\item (SENT) RIFLE: Lennart Beringer
\item (SENT) Stanford/LIO: Deian Stefan, Alejandro Russo, Jeremy Planul
\item (SENT) seL4: Toby Murray
\item (SENT) Harvard folks: Steve Chong, Aslan Askarov, Scott Moore
  \ch{conflicted with DP and RAP, so they can't be our external
    reviewers anyway}
\item (SENT) Steve Zdancewic
\item (maybe) Heiko Mantel
\item (maybe) Michael R. Clarkson
\item (maybe) Dave Sands - Delphine? \dd{If we send it to Dave, then Andrei might have a look at it before the review, if we propose him as an external reviewer. Anyway, this is probably too late for sending the draft.}
\item Martin Abadi (abadi@microsoft.com),
\end{itemize}

\subsection{Submission}

\subsubsection{PC conflicts}
Select the PC members who have conflicts of interest with this
paper. This includes past advisors and students, people with the same
affiliation, and any recent (~2 years) coauthors and collaborators.

\begin{itemize}
\item Aaron Turon
\item Aditya Nori
\item Ahmed Bouajjani
\item Alan Jeffrey
\item Andrew Appel -- \ch{Conflict with Andrew Tolmach (advisor)}
\item Andrey Rybalchenko
\item Atsushi Igarashi -- \ch{Conflict with BCP???}\bcp{yes}
\item Dino Distefano
\item Elena Zucca
\item Gilles Barthe -- \ch{Conflict with DD (recent coauthors)}
\item James Cheney -- \rap{Conflict with RAP (affiliation Edinburgh)}
\item Jeremy Siek
\item Lars Birkedal
\item Mads Dam
\item Nate Foster -- \ch{Conflict with BCP (student)}
\item Neelakantan Krishnaswami
\item Nick Benton
\item Nikhil Swamy
\item Noam Rinetzky
\item Peter Sewell
\item Ross Tate
\item Sophia Drossopoulou
\item Stephanie Weirich -- \ch{Conflict with all Penn folks (affiliation)}
\item Tayssir Touili
\item Thomas Dillig
\item Thomas Wies
\item Viktor Vafeiadis
\item Xavier Rival -- \ch{conflicted with DP (same institution)}
\end{itemize}

\subsubsection{Other conflicts}
List other people and institutions with which the authors have
conflicts of interest. This will help us avoid conflicts when
assigning external reviews. No need to list people at the authors' own
institutions. List one conflict per line. For example: ``Jelena
Markovic (EPFL)'' or, for a whole institution, ``EPFL''.

\paragraph*{Obvious (no need to list these it seems):}
\begin{itemize}
\item University of Pennsylvania
\item Portland State University
\item INRIA
\item Harvard University
\end{itemize}

\paragraph*{Individual conflicts:}
\begin{itemize}
\item  Arthur Azevedo de Amorim:
None
\item  Nathan Collins:
Aaron Stump (University of Iowa), 
Harley Eades (University of Iowa), 
Peng Fu (University of Iowa), 
Garrin Kimmell (University of Iowa), 
\item  Andr\'e DeHon:
Olin Shivers (Northeastern University), 
Satnam Singh (Google),
Thomas F. Knight, Jr. (MIT) 
Howard Shrobe (MIT), 
\item  Delphine Demange:
Jade Alglave (UCL),
Jan Vitek (Purdue University),
Lei Zhao (Purdue University),
Suresh Jagannathan (Purdue University),
Dave Sands (Chalmers),
\item  C\u{a}t\u{a}lin Hri\c{t}cu:
Michael Backes (Saarland University),
Matteo Maffei (Saarland University),
Dimitrios Vytiniotis (MSR Cambridge),
Andy Gordon (MSR Cambridge),
Gavin Bierman (MSR Cambridge),
John Hughes (Chalmers),
\item  Benjamin C. Pierce:
Martin Hofmann (Ludwig-Maximilians University Munchen, Germany), 
Greg Morrisett (Harvard University, USA), 
Eijiro Sumii (Tohoku University, Japan), 
\item  David Pichardie:
Florent Kirchner (CEA, France)
\item  Randy Pollack:
Edinburgh University, 
Masahiko Sato (Kyoto University),
Helmut Schwichtenberg (University of Munich),
\item  Andrew Tolmach:
Andrew McCreight (Mozilla),
Tim Chevalier (Mozilla),
Michael Adams (UIUC),
Eelco Visser (Delft)
\end{itemize}

\subsection{Nomenclature}

\begin{itemize}
\item use the bare word ``rule'' only when context makes the meaning
  completely clear

\item a rule of the small-step semantics is called a ``small-step rule''

\item The previously called ``quasi-abstract IFC machine'' is now called
``symbolic IFC rule machine'' (shortened to ``symbolic rule machine'').
This machine is parameterized over a ``symbolic IFC rule table''
(shortened to ``rule table''), containing ``symbolic IFC rules''.


\item a judgment of this form
  $ \vdash_{ \mathcal{R} } \ruleeval{ \ottsym{(}  L_{pc}  \mathord{,}\,  \ell_{{\mathrm{1}}}  \mathord{,}\,  \ell_{{\mathrm{2}}}  \mathord{,}\,  \ell_{{\mathrm{3}}}  \ottsym{)} }{ \ottnt{opcode} }{ L_{rpc} }{ L_r } $
  is now called ``IFC enforcement judgment''

\item the concrete machine has a ``rule cache'' containing
   ``rule cache entries''  {\em or} ``concrete rules'' that are
   instances of the ``symbolic IFC rules''. We don't use the word
   TMU in this paper!




\item For the Hoare logic in \autoref{sec:fault-handler-correct} we use:
  ``Hoare triple'' only for something that really is a triple
  (a Hoare rule is not a triple);
   ``laws'' for generic things like composition;
   and ``specifications'' for specific pieces of code
  (``laws'' + ``specifications'' seems to cover everything
   that would otherwise be called a Hoare rule)
\end{itemize}

\fi

\end{document}

%% file: temp/defns.tex
\newcommand{\ottdrule}[4][]{{\displaystyle\frac{\begin{array}{l}#2\end{array}}{#3}\quad\ottdrulename{#4}}}
\newcommand{\ottusedrule}[1]{\[#1\]}
\newcommand{\ottpremise}[1]{ #1 \\}
\newenvironment{ottdefnblock}[3][]{ \framebox{\mbox{#2}} \quad #3 \\[0pt]}{}
\newenvironment{ottfundefnblock}[3][]{ \framebox{\mbox{#2}} \quad #3 \\[0pt]\begin{displaymath}\begin{array}{l}}{\end{array}\end{displaymath}}

\newcommand{\ottnt}[1]{\mathit{#1}}

\newcommand{\ottkw}[1]{\mathbf{#1}}
\newcommand{\ottsym}[1]{#1}
\newcommand{\ottcom}[1]{\text{#1}}
\newcommand{\ottdrulename}[1]{\textsc{#1}}

\newcommand{\ottgrammartabular}[1]{\begin{supertabular}{llcllllll}#1\end{supertabular}}

\newcommand{\ottrulehead}[3]{$#1$ & & $#2$ & & & \multicolumn{2}{l}{#3}}
\newcommand{\ottprodline}[6]{& & $#1$ & $#2$ & $#3 #4$ & $#5$ & $#6$}
\newcommand{\ottfirstprodline}[6]{\ottprodline{#1}{#2}{#3}{#4}{#5}{#6}}

\newcommand{\ottprodnewline}{\\}
\newcommand{\ottinterrule}{\\[5.0mm]}

\newcommand{\ottinstr}{
\ottrulehead{\ottnt{instr}}{::=}{\ottcom{Basic instruction set}}\ottprodnewline
\ottfirstprodline{|}{\text{\sf Add}}{}{}{}{\ottcom{addition}}\ottprodnewline
\ottprodline{|}{\text{\sf Output}}{}{}{}{\ottcom{output top of stack}}\ottprodnewline
\ottprodline{|}{\text{\sf Push} \, \ottnt{n}}{}{}{}{\ottcom{push integer constant}}\ottprodnewline
\ottprodline{|}{\text{\sf Load}}{}{}{}{\ottcom{indirect load from data memory}}\ottprodnewline
\ottprodline{|}{\text{\sf Store}}{}{}{}{\ottcom{indirect store to data memory}}\ottprodnewline
\ottprodline{|}{\text{\sf Jump}}{}{}{}{\ottcom{unconditional indirect jump}}\ottprodnewline
\ottprodline{|}{\text{\sf Bnz} \, \ottnt{n}}{}{}{}{\ottcom{conditional relative jump}}\ottprodnewline
\ottprodline{|}{\text{\sf Call}}{}{}{}{\ottcom{indirect call}}\ottprodnewline
\ottprodline{|}{\text{\sf Ret}}{}{}{}{\ottcom{return}}}

\newcommand{\ottnewinstr}{
\ottrulehead{\mathit{instr}}{::=}{\ottcom{extensions to instruction set}}\ottprodnewline
\ottfirstprodline{|}{ \ldots }{}{}{}{}\ottprodnewline
\ottprodline{|}{\ottkw{Alloc}}{}{}{}{\ottcom{allocate a new frame}}\ottprodnewline
\ottprodline{|}{\ottkw{SizeOf}}{}{}{}{\ottcom{fetch frame size}}\ottprodnewline
\ottprodline{|}{\ottkw{Eq}}{}{}{}{\ottcom{value equality}}\ottprodnewline
\ottprodline{|}{\text{\sf SysCall} \, \ottnt{id}}{}{}{}{\ottcom{system call}}\ottprodnewline
\ottprodline{|}{\ottkw{GetOff}}{}{}{}{\ottcom{extract pointer offset}}\ottprodnewline
\ottprodline{|}{\ottkw{Pack}}{}{}{}{\ottcom{atom from payload and tag}}\ottprodnewline
\ottprodline{|}{\ottkw{Unpack}}{}{}{}{\ottcom{atom into payload and tag}}\ottprodnewline
\ottprodline{|}{\ottkw{PushCachePtr}}{}{}{}{\ottcom{push cache address on stack}}\ottprodnewline
\ottprodline{|}{\text{\sf Dup} \, \ottnt{n}}{}{}{}{\ottcom{duplicate atom on stack}}\ottprodnewline
\ottprodline{|}{\text{\sf Swap} \, \ottnt{n}}{}{}{}{\ottcom{swap two data atoms on stack}}}

\newcommand{\ottLE}{
\ottrulehead{\ottnt{LE}  ,\ e_{r}  ,\ e_{rpc}}{::=}{}\ottprodnewline
\ottfirstprodline{|}{ \mathtt{BOT} }{}{}{}{}\ottprodnewline
\ottprodline{|}{ \mathtt{LAB}_1 }{}{}{}{}\ottprodnewline
\ottprodline{|}{ \mathtt{LAB}_2 }{}{}{}{}\ottprodnewline
\ottprodline{|}{ \mathtt{LAB}_3 }{}{}{}{}\ottprodnewline
\ottprodline{|}{ \mathtt{LAB}_\mathit{pc} }{}{}{}{}\ottprodnewline
\ottprodline{|}{ \ottnt{LE_{{\mathrm{1}}}} ~ \sqcup ~ \ottnt{LE_{{\mathrm{2}}}} }{}{}{}{}\ottprodnewline
\ottprodline{|}{\ottsym{\_\_}}{}{}{}{}\ottprodnewline
\ottprodline{|}{\ottsym{(}  \ottnt{LE}  \ottsym{)}}{}{}{}{}}

\newcommand{\ottBE}{
\ottrulehead{\ottnt{BE}  ,\ \ottnt{allow}}{::=}{}\ottprodnewline
\ottfirstprodline{|}{\mathtt{TRUE}}{}{}{}{}\ottprodnewline
\ottprodline{|}{ \ottnt{LE_{{\mathrm{1}}}} \sqsubseteq \ottnt{LE_{{\mathrm{2}}}} }{}{}{}{}\ottprodnewline
\ottprodline{|}{ \mathtt{AND} ~ \ottnt{BE_{{\mathrm{1}}}} ~ \ottnt{BE_{{\mathrm{2}}}} }{}{}{}{}\ottprodnewline
\ottprodline{|}{ \mathtt{OR} ~ \ottnt{BE_{{\mathrm{1}}}} ~ \ottnt{BE_{{\mathrm{2}}}} }{}{}{}{}\ottprodnewline
\ottprodline{|}{\ottsym{(}  \ottnt{BE}  \ottsym{)}}{}{}{}{}}

\newcommand{\ottdruleStepXXOutput}[1]{\ottdrule[#1]{%
\ottpremise{\iota  \ottsym{(}  \ottnt{n}  \ottsym{)}  \ottsym{=}  \text{\sf Output}}%
}{
 \aStep{ \mu }{[  \ottnt{m}  \mathord{\scriptstyle @}  \ottnt{L_{{\mathrm{1}}}}  \mathord{,}\,  \sigma  ]}{ \ottnt{n}  \mathord{\scriptstyle @}  L_{pc} }{ \ottnt{m}  \mathord{\scriptstyle @}  \ottsym{(}  \ottnt{L_{{\mathrm{1}}}}  \mathord{\vee}  L_{pc}  \ottsym{)} }{ \mu }{[  \sigma  ]}{ \ottsym{(}  \ottnt{n}  \mathord{+}  \ottsym{1}  \ottsym{)}  \mathord{\scriptstyle @}  L_{pc} }{ \ottsym{\mbox{$\backslash{}$}\mbox{$\backslash{}$}} } }{%
{\ottdrulename{Step\_Output}}{}%
}}

\newcommand{\ottdruleStepXXAdd}[1]{\ottdrule[#1]{%
\ottpremise{\iota  \ottsym{(}  \ottnt{n}  \ottsym{)}  \ottsym{=}  \text{\sf Add}}%
}{
 \aStep{ \mu }{[  \ottnt{n_{{\mathrm{1}}}}  \mathord{\scriptstyle @}  \ottnt{L_{{\mathrm{1}}}}  \mathord{,}\,  \ottnt{n_{{\mathrm{2}}}}  \mathord{\scriptstyle @}  \ottnt{L_{{\mathrm{2}}}}  \mathord{,}\,  \sigma  ]}{ \ottnt{n}  \mathord{\scriptstyle @}  L_{pc} }{ \tau }{ \mu }{[  \ottsym{(}  \ottnt{n_{{\mathrm{1}}}}  \mathord{+}  \ottnt{n_{{\mathrm{2}}}}  \ottsym{)}  \mathord{\scriptstyle @}  \ottsym{(}  \ottnt{L_{{\mathrm{1}}}}  \mathord{\vee}  \ottnt{L_{{\mathrm{2}}}}  \ottsym{)}  \mathord{,}\,  \sigma  ]}{ \ottsym{(}  \ottnt{n}  \mathord{+}  \ottsym{1}  \ottsym{)}  \mathord{\scriptstyle @}  L_{pc} }{ \ottsym{\mbox{$\backslash{}$}\mbox{$\backslash{}$}} } }{%
{\ottdrulename{Step\_Add}}{}%
}}

\newcommand{\ottdruleStepXXPush}[1]{\ottdrule[#1]{%
\ottpremise{\iota  \ottsym{(}  \ottnt{n}  \ottsym{)}  \ottsym{=}  \text{\sf Push} \, \ottnt{m}}%
}{
 \aStep{ \mu }{[  \sigma  ]}{ \ottnt{n}  \mathord{\scriptstyle @}  L_{pc} }{ \tau }{ \mu }{[  \ottnt{m}  \mathord{\scriptstyle @}  \bot  \mathord{,}\,  \sigma  ]}{ \ottsym{(}  \ottnt{n}  \mathord{+}  \ottsym{1}  \ottsym{)}  \mathord{\scriptstyle @}  L_{pc} }{  } }{%
{\ottdrulename{Step\_Push}}{}%
}}

\newcommand{\ottdruleStepXXLoad}[1]{\ottdrule[#1]{%
\ottpremise{\iota  \ottsym{(}  \ottnt{n}  \ottsym{)}  \ottsym{=}  \text{\sf Load} \, \qquad \,  \mu ( \ottnt{p} )   \ottsym{=}  \ottnt{m}  \mathord{\scriptstyle @}  \ottnt{L_{{\mathrm{2}}}}}%
}{
 \aStep{ \mu }{[  \ottnt{p}  \mathord{\scriptstyle @}  \ottnt{L_{{\mathrm{1}}}}  \mathord{,}\,  \sigma  ]}{ \ottnt{n}  \mathord{\scriptstyle @}  L_{pc} }{ \tau }{ \mu }{[  \ottnt{m}  \mathord{\scriptstyle @}  \ottsym{(}  \ottnt{L_{{\mathrm{1}}}}  \mathord{\vee}  \ottnt{L_{{\mathrm{2}}}}  \ottsym{)}  \mathord{,}\,  \sigma  ]}{ \ottsym{(}  \ottnt{n}  \mathord{+}  \ottsym{1}  \ottsym{)}  \mathord{\scriptstyle @}  L_{pc} }{ \ottsym{\mbox{$\backslash{}$}\mbox{$\backslash{}$}} } }{%
{\ottdrulename{Step\_Load}}{}%
}}

\newcommand{\ottdruleStepXXStore}[1]{\ottdrule[#1]{%
\ottpremise{\iota  \ottsym{(}  \ottnt{n}  \ottsym{)}  \ottsym{=}  \text{\sf Store} \, \qquad \,  \mu ( \ottnt{p} )   \ottsym{=}  \ottnt{k}  \mathord{\scriptstyle @}  \ottnt{L_{{\mathrm{3}}}}}%
\ottpremise{\ottnt{L_{{\mathrm{1}}}}  \mathord{\vee}  L_{pc}  \le  \ottnt{L_{{\mathrm{3}}}} \, \qquad \,  \mu ( \ottnt{p} ) \leftarrow  \ottsym{(}  \ottnt{m}  \mathord{\scriptstyle @}  \ottnt{L_{{\mathrm{1}}}}  \mathord{\vee}  \ottnt{L_{{\mathrm{2}}}}  \mathord{\vee}  L_{pc}  \ottsym{)}   \ottsym{=}  \mu'}%
}{
 \aStep{ \mu }{[  \ottnt{p}  \mathord{\scriptstyle @}  \ottnt{L_{{\mathrm{1}}}}  \mathord{,}\,  \ottnt{m}  \mathord{\scriptstyle @}  \ottnt{L_{{\mathrm{2}}}}  \mathord{,}\,  \sigma  ]}{ \ottnt{n}  \mathord{\scriptstyle @}  L_{pc} }{ \tau }{ \mu' }{[  \sigma  ]}{ \ottsym{(}  \ottnt{n}  \mathord{+}  \ottsym{1}  \ottsym{)}  \mathord{\scriptstyle @}  L_{pc} }{ \ottsym{\mbox{$\backslash{}$}\mbox{$\backslash{}$}} } }{%
{\ottdrulename{Step\_Store}}{}%
}}

\newcommand{\ottdruleStepXXJump}[1]{\ottdrule[#1]{%
\ottpremise{\iota  \ottsym{(}  \ottnt{n}  \ottsym{)}  \ottsym{=}  \text{\sf Jump}}%
}{
 \aStep{ \mu }{[  \ottnt{n'}  \mathord{\scriptstyle @}  \ottnt{L_{{\mathrm{1}}}}  \mathord{,}\,  \sigma  ]}{ \ottnt{n}  \mathord{\scriptstyle @}  L_{pc} }{ \tau }{ \mu }{[  \sigma  ]}{ \ottnt{n'}  \mathord{\scriptstyle @}  \ottsym{(}  \ottnt{L_{{\mathrm{1}}}}  \mathord{\vee}  L_{pc}  \ottsym{)} }{  } }{%
{\ottdrulename{Step\_Jump}}{}%
}}

\newcommand{\ottdruleStepXXBnz}[1]{\ottdrule[#1]{%
\ottpremise{\iota  \ottsym{(}  \ottnt{n}  \ottsym{)}  \ottsym{=}  \text{\sf Bnz} \, \ottnt{k} \, \qquad \, \ottnt{n'}  \ottsym{=}  \ottnt{n}  \mathord{+}  \ottsym{(}  \ottsym{(}  \ottnt{m}  \ottsym{=}  \ottsym{0}  \ottsym{)}  \ottsym{\mbox{?}}  \ottsym{1}  \ottsym{:}  \ottnt{k}  \ottsym{)}}%
}{
 \aStep{ \mu }{[  \ottnt{m}  \mathord{\scriptstyle @}  \ottnt{L_{{\mathrm{1}}}}  \mathord{,}\,  \sigma  ]}{ \ottnt{n}  \mathord{\scriptstyle @}  L_{pc} }{ \tau }{ \mu }{[  \sigma  ]}{ \ottnt{n'}  \mathord{\scriptstyle @}  \ottsym{(}  \ottnt{L_{{\mathrm{1}}}}  \mathord{\vee}  L_{pc}  \ottsym{)} }{  } }{%
{\ottdrulename{Step\_Bnz}}{}%
}}

\newcommand{\ottdruleStepXXCall}[1]{\ottdrule[#1]{%
\ottpremise{\iota  \ottsym{(}  \ottnt{n}  \ottsym{)}  \ottsym{=}  \text{\sf Call}}%
}{
 \aStep{ \mu }{[  \ottnt{n'}  \mathord{\scriptstyle @}  \ottnt{L_{{\mathrm{1}}}}  \mathord{,}\,  \ottnt{a}  \mathord{,}\,  \sigma  ]}{ \ottnt{n}  \mathord{\scriptstyle @}  L_{pc} }{ \tau }{ \mu }{[  \ottnt{a}  \mathord{,}\,  \ottsym{(}  \ottnt{n}  \mathord{+}  \ottsym{1}  \ottsym{)}  \mathord{\scriptstyle @}  L_{pc}  \mathord{;}\,  \sigma  ]}{ \ottnt{n'}  \mathord{\scriptstyle @}  \ottsym{(}  \ottnt{L_{{\mathrm{1}}}}  \mathord{\vee}  L_{pc}  \ottsym{)} }{ \ottsym{\mbox{$\backslash{}$}\mbox{$\backslash{}$}} } }{%
{\ottdrulename{Step\_Call}}{}%
}}

\newcommand{\ottdruleStepXXRet}[1]{\ottdrule[#1]{%
\ottpremise{\iota  \ottsym{(}  \ottnt{n}  \ottsym{)}  \ottsym{=}  \text{\sf Ret}}%
}{
 \aStep{ \mu }{[  \ottnt{n'}  \mathord{\scriptstyle @}  \ottnt{L_{{\mathrm{1}}}}  \mathord{;}\,  \sigma  ]}{ \ottnt{n}  \mathord{\scriptstyle @}  L_{pc} }{ \tau }{ \mu }{[  \sigma  ]}{ \ottnt{n'}  \mathord{\scriptstyle @}  \ottnt{L_{{\mathrm{1}}}} }{  } }{%
{\ottdrulename{Step\_Ret}}{}%
}}

\newcommand{\ottdruleEStepAlloc}[1]{\ottdrule[#1]{%
\ottpremise{\iota  \ottsym{(}  \ottnt{n}  \ottsym{)}  \ottsym{=}  \ottkw{Alloc} \, \qquad \, \text{\sf alloc} \, \ottnt{k} \, \ottsym{(}  \ottnt{L}  \mathord{\vee}  L_{pc}  \ottsym{)} \, a \, \mu  \ottsym{=}  \ottsym{(}  \ottnt{id}  \mathord{,}\,  \mu'  \ottsym{)}}%
}{
 \aStep{ \mu }{[  \ottsym{(}  \ottkw{Int} \, \ottnt{k}  \ottsym{)}  \mathord{\scriptstyle @}  \ottnt{L}  \mathord{,}\,  a  \mathord{,}\,  \sigma  ]}{ \ottnt{n}  \mathord{\scriptstyle @}  L_{pc} }{ \tau }{ \mu' }{[  \ottsym{(}  \ottkw{Ptr} \, \ottsym{(}  \ottnt{id}  \mathord{,}\,  \ottsym{0}  \ottsym{)}  \ottsym{)}  \mathord{\scriptstyle @}  \ottnt{L}  \mathord{,}\,  \sigma  ]}{ \ottsym{(}  \ottnt{n}  \mathord{+}  \ottsym{1}  \ottsym{)}  \mathord{\scriptstyle @}  L_{pc} }{ \ottsym{\mbox{$\backslash{}$}\mbox{$\backslash{}$}} } }{%
{\ottdrulename{EStepAlloc}}{}%
}}

\newcommand{\ottdruleEStepSysCall}[1]{\ottdrule[#1]{%
\ottpremise{\iota  \ottsym{(}  \ottnt{n}  \ottsym{)}  \ottsym{=}  \text{\sf SysCall} \, \ottnt{id} \, \qquad \, T  \ottsym{(}  \ottnt{id}  \ottsym{)}  \ottsym{=}  \ottsym{(}  \ottnt{k}  \mathord{,}\,  f  \ottsym{)}}%
\ottpremise{f  \ottsym{(}  \sigma_{{\mathrm{1}}}  \ottsym{)}  \ottsym{=}  v  \mathord{\scriptstyle @}  \ottnt{L} \, \qquad \, \ottkw{length} \, \ottsym{(}  \sigma_{{\mathrm{1}}}  \ottsym{)}  \ottsym{=}  \ottnt{k}}%
}{
 \aStep{ \mu }{[  \sigma_{{\mathrm{1}}}  {\scriptstyle \mathord{+\!+} }  \sigma_{{\mathrm{2}}}  ]}{ \ottnt{n}  \mathord{\scriptstyle @}  L_{pc} }{ \tau }{ \mu }{[  v  \mathord{\scriptstyle @}  \ottnt{L}  \mathord{,}\,  \sigma_{{\mathrm{2}}}  ]}{ \ottsym{(}  \ottnt{n}  \mathord{+}  \ottsym{1}  \ottsym{)}  \mathord{\scriptstyle @}  L_{pc} }{  } }{%
{\ottdrulename{EStepSysCall}}{}%
}}

\newcommand{\ottdruleEStepjoinP}[1]{\ottdrule[#1]{%
\ottpremise{\iota  \ottsym{(}  \ottnt{n}  \ottsym{)}  \ottsym{=}  \text{\sf SysCall} \, \text{\sf joinP}}%
}{
 \aStep{ \mu }{[  v  \mathord{\scriptstyle @}  \ottnt{L_{{\mathrm{1}}}}  \mathord{,}\,  \ottsym{(}  \ottkw{Int} \, \ottnt{m}  \ottsym{)}  \mathord{\scriptstyle @}  \ottnt{L_{{\mathrm{2}}}}  \mathord{,}\,  \sigma  ]}{ \ottnt{n}  \mathord{\scriptstyle @}  L_{pc} }{ \tau }{ \mu }{[  v  \mathord{\scriptstyle @}  \ottsym{(}  \ottnt{L_{{\mathrm{1}}}}  \mathord{\vee}  \ottnt{L_{{\mathrm{2}}}}  \mathord{\vee}  \ottsym{\{}  \ottnt{m}  \ottsym{\}}  \ottsym{)}  \mathord{,}\,  \sigma  ]}{ \ottsym{(}  \ottnt{n}  \mathord{+}  \ottsym{1}  \ottsym{)}  \mathord{\scriptstyle @}  L_{pc} }{  } }{%
{\ottdrulename{EStepjoinP}}{}%
}}

\newcommand{\ottdruleEStepGetOff}[1]{\ottdrule[#1]{%
\ottpremise{\iota  \ottsym{(}  \ottnt{n}  \ottsym{)}  \ottsym{=}  \ottkw{GetOff}}%
}{
 \aStep{ \mu }{[  \ottsym{(}  \ottkw{Ptr} \, \ottsym{(}  \ottnt{id}  \mathord{,}\,  o  \ottsym{)}  \ottsym{)}  \mathord{\scriptstyle @}  \ottnt{L}  \mathord{,}\,  \sigma  ]}{ \ottnt{n}  \mathord{\scriptstyle @}  L_{pc} }{ \tau }{ \mu }{[  \ottsym{(}  \ottkw{Int} \, o  \ottsym{)}  \mathord{\scriptstyle @}  \ottnt{L}  \mathord{,}\,  \sigma  ]}{ \ottsym{(}  \ottnt{n}  \mathord{+}  \ottsym{1}  \ottsym{)}  \mathord{\scriptstyle @}  L_{pc} }{  } }{%
{\ottdrulename{EStepGetOff}}{}%
}}

\newcommand{\ottdruleEStepSizeOf}[1]{\ottdrule[#1]{%
\ottpremise{\iota  \ottsym{(}  \ottnt{n}  \ottsym{)}  \ottsym{=}  \ottkw{SizeOf} \, \qquad \, \ottkw{length} \, \ottsym{(}   \mu ( \ottnt{id} )   \ottsym{)}  \ottsym{=}  \ottnt{k}}%
}{
 \aStep{ \mu }{[  \ottsym{(}  \ottkw{Ptr} \, \ottsym{(}  \ottnt{id}  \mathord{,}\,  o  \ottsym{)}  \ottsym{)}  \mathord{\scriptstyle @}  \ottnt{L}  \mathord{,}\,  \sigma  ]}{ \ottnt{n}  \mathord{\scriptstyle @}  L_{pc} }{ \tau }{ \mu }{[  \ottsym{(}  \ottkw{Int} \, \ottnt{k}  \ottsym{)}  \mathord{\scriptstyle @}  \ottnt{L}  \mathord{,}\,  \sigma  ]}{ \ottsym{(}  \ottnt{n}  \mathord{+}  \ottsym{1}  \ottsym{)}  \mathord{\scriptstyle @}  L_{pc} }{  } }{%
{\ottdrulename{EStepSizeOf}}{}%
}}

\newcommand{\ottdruleEStepEq}[1]{\ottdrule[#1]{%
\ottpremise{\iota  \ottsym{(}  \ottnt{n}  \ottsym{)}  \ottsym{=}  \ottkw{Eq}}%
}{
 \aStep{ \mu }{[  v_{{\mathrm{1}}}  \mathord{\scriptstyle @}  \ottnt{L_{{\mathrm{1}}}}  \mathord{,}\,  v_{{\mathrm{2}}}  \mathord{\scriptstyle @}  \ottnt{L_{{\mathrm{2}}}}  \mathord{,}\,  \sigma  ]}{ \ottnt{n}  \mathord{\scriptstyle @}  L_{pc} }{ \tau }{ \mu }{[  \ottsym{(}  \ottkw{Int} \, \ottsym{(}  v_{{\mathrm{1}}}  \ottsym{==}  v_{{\mathrm{2}}}  \ottsym{)}  \ottsym{)}  \mathord{\scriptstyle @}  \ottsym{(}  \ottnt{L_{{\mathrm{1}}}}  \mathord{\vee}  \ottnt{L_{{\mathrm{2}}}}  \ottsym{)}  \mathord{,}\,  \sigma  ]}{ \ottsym{(}  \ottnt{n}  \mathord{+}  \ottsym{1}  \ottsym{)}  \mathord{\scriptstyle @}  L_{pc} }{ \ottsym{\mbox{$\backslash{}$}\mbox{$\backslash{}$}} } }{%
{\ottdrulename{EStepEq}}{}%
}}

\newcommand{\ottdrulelbot}[1]{\ottdrule[#1]{%
}{
 \rho  \vdash    \mathtt{BOT}   \downarrow  \bot }{%
{\ottdrulename{lbot}}{}%
}}

\newcommand{\ottdrulelpc}[1]{\ottdrule[#1]{%
}{
 \ottsym{(}  L_{pc}  \mathord{,}\,  \ell_{{\mathrm{1}}}  \mathord{,}\,  \ell_{{\mathrm{2}}}  \mathord{,}\,  \ell_{{\mathrm{3}}}  \ottsym{)}  \vdash    \mathtt{LAB}_\mathit{pc}   \downarrow  L_{pc} }{%
{\ottdrulename{lpc}}{}%
}}

\newcommand{\ottdrulevarOne}[1]{\ottdrule[#1]{%
}{
 \ottsym{(}  L_{pc}  \mathord{,}\,  \ottnt{L_{{\mathrm{1}}}}  \mathord{,}\,  \ell_{{\mathrm{2}}}  \mathord{,}\,  \ell_{{\mathrm{3}}}  \ottsym{)}  \vdash    \mathtt{LAB}_1   \downarrow  \ottnt{L_{{\mathrm{1}}}} }{%
{\ottdrulename{var1}}{}%
}}

\newcommand{\ottdrulevarTwo}[1]{\ottdrule[#1]{%
}{
 \ottsym{(}  L_{pc}  \mathord{,}\,  \ell_{{\mathrm{1}}}  \mathord{,}\,  \ottnt{L_{{\mathrm{2}}}}  \mathord{,}\,  \ell_{{\mathrm{3}}}  \ottsym{)}  \vdash    \mathtt{LAB}_2   \downarrow  \ottnt{L_{{\mathrm{2}}}} }{%
{\ottdrulename{var2}}{}%
}}

\newcommand{\ottdrulevarThree}[1]{\ottdrule[#1]{%
}{
 \ottsym{(}  L_{pc}  \mathord{,}\,  \ell_{{\mathrm{1}}}  \mathord{,}\,  \ell_{{\mathrm{2}}}  \mathord{,}\,  \ottnt{L_{{\mathrm{3}}}}  \ottsym{)}  \vdash    \mathtt{LAB}_3   \downarrow  \ottnt{L_{{\mathrm{3}}}} }{%
{\ottdrulename{var3}}{}%
}}

\newcommand{\ottdrulejoin}[1]{\ottdrule[#1]{%
\ottpremise{ \rho  \vdash   \ottnt{LE_{{\mathrm{1}}}}  \downarrow  \ottnt{L_{{\mathrm{1}}}}  \, \qquad \,  \rho  \vdash   \ottnt{LE_{{\mathrm{2}}}}  \downarrow  \ottnt{L_{{\mathrm{2}}}} }%
}{
 \rho  \vdash   \ottsym{(}   \ottnt{LE_{{\mathrm{1}}}} ~ \sqcup ~ \ottnt{LE_{{\mathrm{2}}}}   \ottsym{)}  \downarrow  \ottsym{(}  \ottnt{L_{{\mathrm{1}}}}  \mathord{\vee}  \ottnt{L_{{\mathrm{2}}}}  \ottsym{)} }{%
{\ottdrulename{join}}{}%
}}

\newcommand{\ottdruletrue}[1]{\ottdrule[#1]{%
}{
 \rho  \vdash   \mathtt{TRUE} }{%
{\ottdrulename{true}}{}%
}}

\newcommand{\ottdruleflows}[1]{\ottdrule[#1]{%
\ottpremise{ \rho  \vdash   \ottnt{LE_{{\mathrm{1}}}}  \downarrow  \ottnt{L_{{\mathrm{1}}}}  \, \qquad \,  \rho  \vdash   \ottnt{LE_{{\mathrm{2}}}}  \downarrow  \ottnt{L_{{\mathrm{2}}}}  \, \qquad \, \ottnt{L_{{\mathrm{1}}}}  \le  \ottnt{L_{{\mathrm{2}}}}}%
}{
 \rho  \vdash    \ottnt{LE_{{\mathrm{1}}}} \sqsubseteq \ottnt{LE_{{\mathrm{2}}}}  }{%
{\ottdrulename{flows}}{}%
}}

\newcommand{\ottdruleand}[1]{\ottdrule[#1]{%
\ottpremise{ \rho  \vdash   \ottnt{BE_{{\mathrm{1}}}}  \, \qquad \,  \rho  \vdash   \ottnt{BE_{{\mathrm{2}}}} }%
}{
 \rho  \vdash    \mathtt{AND} ~ \ottnt{BE_{{\mathrm{1}}}} ~ \ottnt{BE_{{\mathrm{2}}}}  }{%
{\ottdrulename{and}}{}%
}}

\newcommand{\ottdruleorOne}[1]{\ottdrule[#1]{%
\ottpremise{ \rho  \vdash   \ottnt{BE_{{\mathrm{1}}}} }%
}{
 \rho  \vdash    \mathtt{OR} ~ \ottnt{BE_{{\mathrm{1}}}} ~ \ottnt{BE_{{\mathrm{2}}}}  }{%
{\ottdrulename{or1}}{}%
}}

\newcommand{\ottdruleorTwo}[1]{\ottdrule[#1]{%
\ottpremise{ \rho  \vdash   \ottnt{BE_{{\mathrm{2}}}} }%
}{
 \rho  \vdash    \mathtt{OR} ~ \ottnt{BE_{{\mathrm{1}}}} ~ \ottnt{BE_{{\mathrm{2}}}}  }{%
{\ottdrulename{or2}}{}%
}}

\newcommand{\ottdruleevalXXrule}[1]{\ottdrule[#1]{%
\ottpremise{ ${\it Rule}$_ \mathcal{R} ( \ottnt{op} )   \ottsym{=}   \langle  \ottnt{allow} ,  e_{rpc} ,  e_{r}  \rangle }%
\ottpremise{ \rho  \vdash   \ottnt{allow}  \, \qquad \,  \rho  \vdash   e_{rpc}  \downarrow  L_{rpc}  \, \qquad \,  \rho  \vdash   e_{r}  \downarrow  L_r }%
}{
 \vdash_{ \mathcal{R} } \ruleeval{ \rho }{ \ottnt{op} }{ L_{rpc} }{ L_r } }{%
{\ottdrulename{eval\_rule}}{}%
}}

\newcommand{\ottdruleevalXXruleXXNoRes}[1]{\ottdrule[#1]{%
\ottpremise{ ${\it Rule}$_ \mathcal{R} ( \ottnt{op} )   \ottsym{=}   \langle  \ottnt{allow} ,  e_{rpc} ,  \ottsym{\_\_}  \rangle }%
\ottpremise{ \rho  \vdash   \ottnt{allow}  \, \qquad \,  \rho  \vdash   e_{rpc}  \downarrow  L_{rpc} }%
}{
 \vdash_{ \mathcal{R} } \ruleeval{ \rho }{ \ottnt{op} }{ L_{rpc} }{  \dummytag  } }{%
{\ottdrulename{eval\_rule\_NoRes}}{}%
}}

\newcommand{\ottdruleOutput}[1]{\ottdrule[#1]{%
\ottpremise{\iota  \ottsym{(}  \ottnt{n}  \ottsym{)}  \ottsym{=}  \text{\sf Output}}%
\ottpremise{ \vdash_{ \mathcal{R} } \ruleeval{ \ottsym{(}  L_{pc}  \mathord{,}\,  \ottnt{L_{{\mathrm{1}}}}  \mathord{,}\,   \dummytag   \mathord{,}\,   \dummytag   \ottsym{)} }{ \ottkw{output} }{ L_{rpc} }{ L_r } }%
}{
 { \aStep{ \mu }{[  \ottnt{m}  \mathord{\scriptstyle @}  \ottnt{L_{{\mathrm{1}}}}  \mathord{,}\,  \sigma  ]}{ \ottnt{n}  \mathord{\scriptstyle @}  L_{pc} }{ \ottnt{m}  \mathord{\scriptstyle @}  L_r }{ \mu }{[  \sigma  ]}{ \ottsym{(}  \ottnt{n}  \mathord{+}  \ottsym{1}  \ottsym{)}  \mathord{\scriptstyle @}  L_{rpc} }{ \ottsym{\mbox{$\backslash{}$}\mbox{$\backslash{}$}} } } }{%
{\ottdrulename{Output}}{}%
}}

\newcommand{\ottdruleSAdd}[1]{\ottdrule[#1]{%
\ottpremise{\iota  \ottsym{(}  \ottnt{n}  \ottsym{)}  \ottsym{=}  \text{\sf Add}}%
\ottpremise{ \vdash_{ \mathcal{R} } \ruleeval{ \ottsym{(}  L_{pc}  \mathord{,}\,  \ottnt{L_{{\mathrm{1}}}}  \mathord{,}\,  \ottnt{L_{{\mathrm{2}}}}  \mathord{,}\,   \dummytag   \ottsym{)} }{ \ottkw{add} }{ L_{rpc} }{ L_r } }%
}{
 { \aStep{ \mu }{[  \ottnt{n_{{\mathrm{1}}}}  \mathord{\scriptstyle @}  \ottnt{L_{{\mathrm{1}}}}  \mathord{,}\,  \ottnt{n_{{\mathrm{2}}}}  \mathord{\scriptstyle @}  \ottnt{L_{{\mathrm{2}}}}  \mathord{,}\,  \sigma  ]}{ \ottnt{n}  \mathord{\scriptstyle @}  L_{pc} }{ \tau }{ \mu }{[  \ottsym{(}  \ottnt{n_{{\mathrm{1}}}}  \mathord{+}  \ottnt{n_{{\mathrm{2}}}}  \ottsym{)}  \mathord{\scriptstyle @}  L_r  \mathord{,}\,  \sigma  ]}{ \ottsym{(}  \ottnt{n}  \mathord{+}  \ottsym{1}  \ottsym{)}  \mathord{\scriptstyle @}  L_{rpc} }{ \ottsym{\mbox{$\backslash{}$}\mbox{$\backslash{}$}} } } }{%
{\ottdrulename{SAdd}}{}%
}}

\newcommand{\ottdrulePush}[1]{\ottdrule[#1]{%
\ottpremise{\iota  \ottsym{(}  \ottnt{n}  \ottsym{)}  \ottsym{=}  \text{\sf Push} \, \ottnt{m}}%
\ottpremise{ \vdash_{ \mathcal{R} } \ruleeval{ \ottsym{(}  L_{pc}  \mathord{,}\,   \dummytag   \mathord{,}\,   \dummytag   \mathord{,}\,   \dummytag   \ottsym{)} }{ \ottkw{push} }{ L_{rpc} }{ L_r } }%
}{
 { \aStep{ \mu }{[  \sigma  ]}{ \ottnt{n}  \mathord{\scriptstyle @}  L_{pc} }{ \tau }{ \mu }{[  \ottnt{m}  \mathord{\scriptstyle @}  L_r  \mathord{,}\,  \sigma  ]}{ \ottsym{(}  \ottnt{n}  \mathord{+}  \ottsym{1}  \ottsym{)}  \mathord{\scriptstyle @}  L_{rpc} }{  } } }{%
{\ottdrulename{Push}}{}%
}}

\newcommand{\ottdruleLoad}[1]{\ottdrule[#1]{%
\ottpremise{\iota  \ottsym{(}  \ottnt{n}  \ottsym{)}  \ottsym{=}  \text{\sf Load} \, \qquad \,  \mu ( \ottnt{p} )   \ottsym{=}  \ottnt{m}  \mathord{\scriptstyle @}  \ottnt{L_{{\mathrm{2}}}}}%
\ottpremise{ \vdash_{ \mathcal{R} } \ruleeval{ \ottsym{(}  L_{pc}  \mathord{,}\,  \ottnt{L_{{\mathrm{1}}}}  \mathord{,}\,  \ottnt{L_{{\mathrm{2}}}}  \mathord{,}\,   \dummytag   \ottsym{)} }{ \ottkw{load} }{ L_{rpc} }{ L_r } }%
}{
 { \aStep{ \mu }{[  \ottnt{p}  \mathord{\scriptstyle @}  \ottnt{L_{{\mathrm{1}}}}  \mathord{,}\,  \sigma  ]}{ \ottnt{n}  \mathord{\scriptstyle @}  L_{pc} }{ \tau }{ \mu }{[  \ottnt{m}  \mathord{\scriptstyle @}  L_r  \mathord{,}\,  \sigma  ]}{ \ottsym{(}  \ottnt{n}  \mathord{+}  \ottsym{1}  \ottsym{)}  \mathord{\scriptstyle @}  L_{rpc} }{ \ottsym{\mbox{$\backslash{}$}\mbox{$\backslash{}$}} } } }{%
{\ottdrulename{Load}}{}%
}}

\newcommand{\ottdruleStore}[1]{\ottdrule[#1]{%
\ottpremise{\iota  \ottsym{(}  \ottnt{n}  \ottsym{)}  \ottsym{=}  \text{\sf Store} \, \qquad \,  \mu ( \ottnt{p} )   \ottsym{=}  \ottnt{k}  \mathord{\scriptstyle @}  \ottnt{L_{{\mathrm{3}}}}}%
\ottpremise{ \vdash_{ \mathcal{R} } \ruleeval{ \ottsym{(}  L_{pc}  \mathord{,}\,  \ottnt{L_{{\mathrm{1}}}}  \mathord{,}\,  \ottnt{L_{{\mathrm{2}}}}  \mathord{,}\,  \ottnt{L_{{\mathrm{3}}}}  \ottsym{)} }{ \ottkw{store} }{ L_{rpc} }{ L_r } }%
\ottpremise{ \mu ( \ottnt{p} ) \leftarrow  \ottnt{m}  \mathord{\scriptstyle @}  L_r   \ottsym{=}  \mu'}%
}{
 { \aStep{ \mu }{[  \ottnt{p}  \mathord{\scriptstyle @}  \ottnt{L_{{\mathrm{1}}}}  \mathord{,}\,  \ottnt{m}  \mathord{\scriptstyle @}  \ottnt{L_{{\mathrm{2}}}}  \mathord{,}\,  \sigma  ]}{ \ottnt{n}  \mathord{\scriptstyle @}  L_{pc} }{ \tau }{ \mu' }{[  \sigma  ]}{ \ottsym{(}  \ottnt{n}  \mathord{+}  \ottsym{1}  \ottsym{)}  \mathord{\scriptstyle @}  L_{rpc} }{  } } }{%
{\ottdrulename{Store}}{}%
}}

\newcommand{\ottdruleJump}[1]{\ottdrule[#1]{%
\ottpremise{\iota  \ottsym{(}  \ottnt{n}  \ottsym{)}  \ottsym{=}  \text{\sf Jump}}%
\ottpremise{ \vdash_{ \mathcal{R} } \ruleeval{ \ottsym{(}  L_{pc}  \mathord{,}\,  \ottnt{L_{{\mathrm{1}}}}  \mathord{,}\,   \dummytag   \mathord{,}\,   \dummytag   \ottsym{)} }{ \ottkw{jump} }{ L_{rpc} }{  \dummytag  } }%
}{
 { \aStep{ \mu }{[  \ottnt{n'}  \mathord{\scriptstyle @}  \ottnt{L_{{\mathrm{1}}}}  \mathord{,}\,  \sigma  ]}{ \ottnt{n}  \mathord{\scriptstyle @}  L_{pc} }{ \tau }{ \mu }{[  \sigma  ]}{ \ottnt{n'}  \mathord{\scriptstyle @}  L_{rpc} }{  } } }{%
{\ottdrulename{Jump}}{}%
}}

\newcommand{\ottdruleBnz}[1]{\ottdrule[#1]{%
\ottpremise{\iota  \ottsym{(}  \ottnt{n}  \ottsym{)}  \ottsym{=}  \text{\sf Bnz} \, \ottnt{k} \, \qquad \, \ottnt{n'}  \ottsym{=}  \ottnt{n}  \mathord{+}  \ottsym{(}  \ottsym{(}  \ottnt{m}  \ottsym{=}  \ottsym{0}  \ottsym{)}  \ottsym{\mbox{?}}  \ottsym{1}  \ottsym{:}  \ottnt{k}  \ottsym{)}}%
\ottpremise{ \vdash_{ \mathcal{R} } \ruleeval{ \ottsym{(}  L_{pc}  \mathord{,}\,  \ottnt{L_{{\mathrm{1}}}}  \mathord{,}\,   \dummytag   \mathord{,}\,   \dummytag   \ottsym{)} }{ \ottkw{bnz} }{ L_{rpc} }{  \dummytag  } }%
}{
 { \aStep{ \mu }{[  \ottnt{m}  \mathord{\scriptstyle @}  \ottnt{L_{{\mathrm{1}}}}  \mathord{,}\,  \sigma  ]}{ \ottnt{n}  \mathord{\scriptstyle @}  L_{pc} }{ \tau }{ \mu }{[  \sigma  ]}{ \ottnt{n'}  \mathord{\scriptstyle @}  L_{rpc} }{  } } }{%
{\ottdrulename{Bnz}}{}%
}}

\newcommand{\ottdruleCall}[1]{\ottdrule[#1]{%
\ottpremise{\iota  \ottsym{(}  \ottnt{n}  \ottsym{)}  \ottsym{=}  \text{\sf Call}}%
\ottpremise{ \vdash_{ \mathcal{R} } \ruleeval{ \ottsym{(}  L_{pc}  \mathord{,}\,  \ottnt{L_{{\mathrm{1}}}}  \mathord{,}\,   \dummytag   \mathord{,}\,   \dummytag   \ottsym{)} }{ \ottkw{call} }{ L_{rpc} }{ L_r } }%
}{
 { \aStep{ \mu }{[  \ottnt{n'}  \mathord{\scriptstyle @}  \ottnt{L_{{\mathrm{1}}}}  \mathord{,}\,  \ottnt{a}  \mathord{,}\,  \sigma  ]}{ \ottnt{n}  \mathord{\scriptstyle @}  L_{pc} }{ \tau }{ \mu }{[  \ottnt{a}  \mathord{,}\,  \ottsym{(}  \ottnt{n}  \mathord{+}  \ottsym{1}  \ottsym{)}  \mathord{\scriptstyle @}  L_r  \mathord{;}\,  \sigma  ]}{ \ottnt{n'}  \mathord{\scriptstyle @}  L_{rpc} }{ \ottsym{\mbox{$\backslash{}$}\mbox{$\backslash{}$}} } } }{%
{\ottdrulename{Call}}{}%
}}

\newcommand{\ottdruleRet}[1]{\ottdrule[#1]{%
\ottpremise{\iota  \ottsym{(}  \ottnt{n}  \ottsym{)}  \ottsym{=}  \text{\sf Ret} \, \qquad \,  \vdash_{ \mathcal{R} } \ruleeval{ \ottsym{(}  L_{pc}  \mathord{,}\,  \ottnt{L_{{\mathrm{1}}}}  \mathord{,}\,   \dummytag   \mathord{,}\,   \dummytag   \ottsym{)} }{ \ottkw{ret} }{ L_{rpc} }{  \dummytag  } }%
}{
 { \aStep{ \mu }{[  \ottnt{n'}  \mathord{\scriptstyle @}  \ottnt{L_{{\mathrm{1}}}}  \mathord{;}\,  \sigma  ]}{ \ottnt{n}  \mathord{\scriptstyle @}  L_{pc} }{ \tau }{ \mu }{[  \sigma  ]}{ \ottnt{n'}  \mathord{\scriptstyle @}  L_{rpc} }{  } } }{%
{\ottdrulename{Ret}}{}%
}}

\newcommand{\ottdruleCOut}[1]{\ottdrule[#1]{%
\ottpremise{\iota  \ottsym{(}  \ottnt{n}  \ottsym{)}  \ottsym{=}  \text{\sf Output}}%
\ottpremise{\kappa  =   \begin{array}{|@{\;}l@{\;}|@{\;}l@{\;}|@{\;}l@{\;}|@{\;}l@{\;}|@{\;}l@{\;}||@{\;}l@{\;}|@{\;}l@{\;}|}
                       \hline
                           \ottkw{output}  &  \mathtt{T}_{pc}  &  \mathtt{T}_{{\mathrm{1}}}  &  \mathtt{T}_\mathtt{D}  &  \mathtt{T}_\mathtt{D}  &  \mathtt{T}_{rpc}  &  \mathtt{T}_{r}  \\
                       \hline
                       \end{array} }%
}{
 \cStep{  \text{\sf u}  }{  \kappa  }{  \concretesymbol{\mu}  }{ [  \ottnt{m}  \mathord{\scriptstyle @}  \mathtt{T}_{{\mathrm{1}}}  \mathord{,}\,  \concretesymbol{\sigma}  ] }{  \ottnt{n}  \mathord{\scriptstyle @}  \mathtt{T}_{pc}  }{  \ottnt{m}  \mathord{\scriptstyle @}  \mathtt{T}_{r}  }{  \text{\sf u}  }{  \kappa  }{  \concretesymbol{\mu}  }{ [  \concretesymbol{\sigma}  ] }{  \ottnt{n}  \mathord{+}  \ottsym{1}  \mathord{\scriptstyle @}  \mathtt{T}_{rpc}  }{  \ottsym{\mbox{$\backslash{}$}\mbox{$\backslash{}$}}  } }{%
{\ottdrulename{COut}}{}%
}}

\newcommand{\ottdruleCOutXXF}[1]{\ottdrule[#1]{%
\ottpremise{\iota  \ottsym{(}  \ottnt{n}  \ottsym{)}  \ottsym{=}  \text{\sf Output}}%
\ottpremise{{\kappa_i}  \not=   \begin{array}{|@{\;}l@{\;}|@{\;}l@{\;}|@{\;}l@{\;}|@{\;}l@{\;}|@{\;}l@{\;}|}
                       \hline
                           \ottkw{output}  &  \mathtt{T}_{pc}  &  \mathtt{T}_{{\mathrm{1}}}  &  \mathtt{T}_\mathtt{D}  &  \mathtt{T}_\mathtt{D}  \\
                       \hline
                       \end{array}   =  {\kappa_j}}%
}{
 \cStep{  \text{\sf u}  }{   [  {\kappa_i}  ,  \kappa_o  ]   }{  \concretesymbol{\mu}  }{ [  \ottnt{m}  \mathord{\scriptstyle @}  \mathtt{T}_{{\mathrm{1}}}  \mathord{,}\,  \concretesymbol{\sigma}  ] }{  \ottnt{n}  \mathord{\scriptstyle @}  \mathtt{T}_{pc}  }{  \tau  }{  \text{\sf k}  }{   [  {\kappa_j}  ,  \kappa_\mathtt{D}  ]   }{  \concretesymbol{\mu}  }{ [  \ottsym{(}  \ottnt{n}  \mathord{\scriptstyle @}  \mathtt{T}_{pc}  \mathord{,}\,  \text{\sf u}  \ottsym{)}  \mathord{;}\,  \ottnt{m}  \mathord{\scriptstyle @}  \mathtt{T}_{{\mathrm{1}}}  \mathord{,}\,  \concretesymbol{\sigma}  ] }{  \ottsym{0}  \mathord{\scriptstyle @}  \mathtt{T}_\mathtt{D}  }{  \ottsym{\mbox{$\backslash{}$}\mbox{$\backslash{}$}}  } }{%
{\ottdrulename{COut\_F}}{}%
}}

\newcommand{\ottdruleCAdd}[1]{\ottdrule[#1]{%
\ottpremise{\iota  \ottsym{(}  \ottnt{n}  \ottsym{)}  \ottsym{=}  \text{\sf Add}}%
\ottpremise{\kappa  =   \begin{array}{|@{\;}l@{\;}|@{\;}l@{\;}|@{\;}l@{\;}|@{\;}l@{\;}|@{\;}l@{\;}||@{\;}l@{\;}|@{\;}l@{\;}|}
                       \hline
                           \ottkw{add}  &  \mathtt{T}_{pc}  &  \mathtt{T}_{{\mathrm{1}}}  &  \mathtt{T}_{{\mathrm{2}}}  &  \mathtt{T}_\mathtt{D}  &  \mathtt{T}_{rpc}  &  \mathtt{T}_{r}  \\
                       \hline
                       \end{array} }%
}{
 \cStep{  \text{\sf u}  }{  \kappa  }{  \concretesymbol{\mu}  }{ [  \ottnt{n_{{\mathrm{1}}}}  \mathord{\scriptstyle @}  \mathtt{T}_{{\mathrm{1}}}  \mathord{,}\,  \ottnt{n_{{\mathrm{2}}}}  \mathord{\scriptstyle @}  \mathtt{T}_{{\mathrm{2}}}  \mathord{,}\,  \concretesymbol{\sigma}  ] }{  \ottnt{n}  \mathord{\scriptstyle @}  \mathtt{T}_{pc}  }{  \tau  }{  \text{\sf u}  }{  \kappa  }{  \concretesymbol{\mu}  }{ [  \ottsym{(}  \ottnt{n_{{\mathrm{1}}}}  \mathord{+}  \ottnt{n_{{\mathrm{2}}}}  \ottsym{)}  \mathord{\scriptstyle @}  \mathtt{T}_{r}  \mathord{,}\,  \concretesymbol{\sigma}  ] }{  \ottnt{n}  \mathord{+}  \ottsym{1}  \mathord{\scriptstyle @}  \mathtt{T}_{rpc}  }{  \ottsym{\mbox{$\backslash{}$}\mbox{$\backslash{}$}}  } }{%
{\ottdrulename{CAdd}}{}%
}}

\newcommand{\ottdruleCAddXXF}[1]{\ottdrule[#1]{%
\ottpremise{\iota  \ottsym{(}  \ottnt{n}  \ottsym{)}  \ottsym{=}  \text{\sf Add}}%
\ottpremise{{\kappa_i}  \not=   \begin{array}{|@{\;}l@{\;}|@{\;}l@{\;}|@{\;}l@{\;}|@{\;}l@{\;}|@{\;}l@{\;}|}
                       \hline
                           \ottkw{add}  &  \mathtt{T}_{pc}  &  \mathtt{T}_{{\mathrm{1}}}  &  \mathtt{T}_{{\mathrm{2}}}  &  \mathtt{T}_\mathtt{D}  \\
                       \hline
                       \end{array}   =  {\kappa_j}}%
}{
 \cStep{  \text{\sf u}  }{   [  {\kappa_i}  ,  \kappa_o  ]   }{  \concretesymbol{\mu}  }{ [  \ottnt{n_{{\mathrm{1}}}}  \mathord{\scriptstyle @}  \mathtt{T}_{{\mathrm{1}}}  \mathord{,}\,  \ottnt{n_{{\mathrm{2}}}}  \mathord{\scriptstyle @}  \mathtt{T}_{{\mathrm{2}}}  \mathord{,}\,  \concretesymbol{\sigma}  ] }{  \ottnt{n}  \mathord{\scriptstyle @}  \mathtt{T}_{pc}  }{  \tau  }{  \text{\sf k}  }{   [  {\kappa_j}  ,  \kappa_\mathtt{D}  ]   }{  \concretesymbol{\mu}  }{ [  \ottsym{(}  \ottnt{n}  \mathord{\scriptstyle @}  \mathtt{T}_{pc}  \mathord{,}\,  \text{\sf u}  \ottsym{)}  \mathord{;}\,  \ottnt{n_{{\mathrm{1}}}}  \mathord{\scriptstyle @}  \mathtt{T}_{{\mathrm{1}}}  \mathord{,}\,  \ottnt{n_{{\mathrm{2}}}}  \mathord{\scriptstyle @}  \mathtt{T}_{{\mathrm{2}}}  \mathord{,}\,  \concretesymbol{\sigma}  ] }{  \ottsym{0}  \mathord{\scriptstyle @}  \mathtt{T}_\mathtt{D}  }{  \ottsym{\mbox{$\backslash{}$}\mbox{$\backslash{}$}}  } }{%
{\ottdrulename{CAdd\_F}}{}%
}}

\newcommand{\ottdruleCAddXXP}[1]{\ottdrule[#1]{%
\ottpremise{\phi  \ottsym{(}  \ottnt{n}  \ottsym{)}  \ottsym{=}  \text{\sf Add}}%
}{
 \cStep{  \text{\sf k}  }{  \kappa  }{  \concretesymbol{\mu}  }{ [  \ottnt{n_{{\mathrm{1}}}}  \mathord{\scriptstyle @}   \dummytag   \mathord{,}\,  \ottnt{n_{{\mathrm{2}}}}  \mathord{\scriptstyle @}   \dummytag   \mathord{,}\,  \concretesymbol{\sigma}  ] }{  \ottnt{n}  \mathord{\scriptstyle @}   \dummytag   }{  \tau  }{  \text{\sf k}  }{  \kappa  }{  \concretesymbol{\mu}  }{ [  \ottsym{(}  \ottnt{n_{{\mathrm{1}}}}  \mathord{+}  \ottnt{n_{{\mathrm{2}}}}  \ottsym{)}  \mathord{\scriptstyle @}  \mathtt{T}_\mathtt{D}  \mathord{,}\,  \concretesymbol{\sigma}  ] }{  \ottnt{n}  \mathord{+}  \ottsym{1}  \mathord{\scriptstyle @}  \mathtt{T}_\mathtt{D}  }{  \ottsym{\mbox{$\backslash{}$}\mbox{$\backslash{}$}}  } }{%
{\ottdrulename{CAdd\_P}}{}%
}}

\newcommand{\ottdruleCPsh}[1]{\ottdrule[#1]{%
\ottpremise{\iota  \ottsym{(}  \ottnt{n}  \ottsym{)}  \ottsym{=}  \text{\sf Push} \, \ottnt{m}}%
\ottpremise{\kappa  =   \begin{array}{|@{\;}l@{\;}|@{\;}l@{\;}|@{\;}l@{\;}|@{\;}l@{\;}|@{\;}l@{\;}||@{\;}l@{\;}|@{\;}l@{\;}|}
                       \hline
                           \ottkw{push}  &  \mathtt{T}_{pc}  &  \mathtt{T}_\mathtt{D}  &  \mathtt{T}_\mathtt{D}  &  \mathtt{T}_\mathtt{D}  &  \mathtt{T}_{rpc}  &  \mathtt{T}_{r}  \\
                       \hline
                       \end{array} }%
}{
 \cStep{  \text{\sf u}  }{  \kappa  }{  \concretesymbol{\mu}  }{ [  \concretesymbol{\sigma}  ] }{  \ottnt{n}  \mathord{\scriptstyle @}  \mathtt{T}_{pc}  }{  \tau  }{  \text{\sf u}  }{  \kappa  }{  \concretesymbol{\mu}  }{ [  \ottnt{m}  \mathord{\scriptstyle @}  \mathtt{T}_{r}  \mathord{,}\,  \concretesymbol{\sigma}  ] }{  \ottnt{n}  \mathord{+}  \ottsym{1}  \mathord{\scriptstyle @}  \mathtt{T}_{rpc}  }{  \ottsym{\mbox{$\backslash{}$}\mbox{$\backslash{}$}}  } }{%
{\ottdrulename{CPsh}}{}%
}}

\newcommand{\ottdruleCPshXXF}[1]{\ottdrule[#1]{%
\ottpremise{\iota  \ottsym{(}  \ottnt{n}  \ottsym{)}  \ottsym{=}  \text{\sf Push} \, \ottnt{m}}%
\ottpremise{{\kappa_i}  \not=   \begin{array}{|@{\;}l@{\;}|@{\;}l@{\;}|@{\;}l@{\;}|@{\;}l@{\;}|@{\;}l@{\;}|}
                       \hline
                           \ottkw{push}  &  \mathtt{T}_{pc}  &  \mathtt{T}_\mathtt{D}  &  \mathtt{T}_\mathtt{D}  &  \mathtt{T}_\mathtt{D}  \\
                       \hline
                       \end{array}   =  {\kappa_j}}%
}{
 \cStep{  \text{\sf u}  }{   [  {\kappa_i}  ,  \kappa_o  ]   }{  \concretesymbol{\mu}  }{ [  \concretesymbol{\sigma}  ] }{  \ottnt{n}  \mathord{\scriptstyle @}  \mathtt{T}_{pc}  }{  \tau  }{  \text{\sf k}  }{   [  {\kappa_j}  ,  \kappa_\mathtt{D}  ]   }{  \concretesymbol{\mu}  }{ [  \ottsym{(}  \ottnt{n}  \mathord{\scriptstyle @}  \mathtt{T}_{pc}  \mathord{,}\,  \text{\sf u}  \ottsym{)}  \mathord{;}\,  \concretesymbol{\sigma}  ] }{  \ottsym{0}  \mathord{\scriptstyle @}  \mathtt{T}_\mathtt{D}  }{  \ottsym{\mbox{$\backslash{}$}\mbox{$\backslash{}$}}  } }{%
{\ottdrulename{CPsh\_F}}{}%
}}

\newcommand{\ottdruleCPshXXP}[1]{\ottdrule[#1]{%
\ottpremise{\phi  \ottsym{(}  \ottnt{n}  \ottsym{)}  \ottsym{=}  \text{\sf Push} \, \ottnt{m}}%
}{
 \cStep{  \text{\sf k}  }{  \kappa  }{  \concretesymbol{\mu}  }{ [  \concretesymbol{\sigma}  ] }{  \ottnt{n}  \mathord{\scriptstyle @}   \dummytag   }{  \tau  }{  \text{\sf k}  }{  \kappa  }{  \concretesymbol{\mu}  }{ [  \ottnt{m}  \mathord{\scriptstyle @}  \mathtt{T}_\mathtt{D}  \mathord{,}\,  \concretesymbol{\sigma}  ] }{  \ottnt{n}  \mathord{+}  \ottsym{1}  \mathord{\scriptstyle @}  \mathtt{T}_\mathtt{D}  }{    } }{%
{\ottdrulename{CPsh\_P}}{}%
}}

\newcommand{\ottdruleCLd}[1]{\ottdrule[#1]{%
\ottpremise{\iota  \ottsym{(}  \ottnt{n}  \ottsym{)}  \ottsym{=}  \text{\sf Load} \, \qquad \,  \concretesymbol{\mu} ( \ottnt{p} )   \ottsym{=}  \ottnt{m}  \mathord{\scriptstyle @}  \mathtt{T}_{{\mathrm{2}}}}%
\ottpremise{\kappa  =   \begin{array}{|@{\;}l@{\;}|@{\;}l@{\;}|@{\;}l@{\;}|@{\;}l@{\;}|@{\;}l@{\;}||@{\;}l@{\;}|@{\;}l@{\;}|}
                       \hline
                           \ottkw{load}  &  \mathtt{T}_{pc}  &  \mathtt{T}_{{\mathrm{1}}}  &  \mathtt{T}_{{\mathrm{2}}}  &  \mathtt{T}_\mathtt{D}  &  \mathtt{T}_{rpc}  &  \mathtt{T}_{r}  \\
                       \hline
                       \end{array} }%
}{
 \cStep{  \text{\sf u}  }{  \kappa  }{  \concretesymbol{\mu}  }{ [  \ottnt{p}  \mathord{\scriptstyle @}  \mathtt{T}_{{\mathrm{1}}}  \mathord{,}\,  \concretesymbol{\sigma}  ] }{  \ottnt{n}  \mathord{\scriptstyle @}  \mathtt{T}_{pc}  }{  \tau  }{  \text{\sf u}  }{  \kappa  }{  \concretesymbol{\mu}  }{ [  \ottnt{m}  \mathord{\scriptstyle @}  \mathtt{T}_{r}  \mathord{,}\,  \concretesymbol{\sigma}  ] }{  \ottnt{n}  \mathord{+}  \ottsym{1}  \mathord{\scriptstyle @}  \mathtt{T}_{rpc}  }{  \ottsym{\mbox{$\backslash{}$}\mbox{$\backslash{}$}}  } }{%
{\ottdrulename{CLd}}{}%
}}

\newcommand{\ottdruleCLdXXF}[1]{\ottdrule[#1]{%
\ottpremise{\iota  \ottsym{(}  \ottnt{n}  \ottsym{)}  \ottsym{=}  \text{\sf Load} \, \qquad \,  \concretesymbol{\mu} ( \ottnt{p} )   \ottsym{=}  \ottnt{m}  \mathord{\scriptstyle @}  \mathtt{T}_{{\mathrm{2}}}}%
\ottpremise{{\kappa_i}  \not=   \begin{array}{|@{\;}l@{\;}|@{\;}l@{\;}|@{\;}l@{\;}|@{\;}l@{\;}|@{\;}l@{\;}|}
                       \hline
                           \ottkw{load}  &  \mathtt{T}_{pc}  &  \mathtt{T}_{{\mathrm{1}}}  &  \mathtt{T}_{{\mathrm{2}}}  &  \mathtt{T}_\mathtt{D}  \\
                       \hline
                       \end{array}   =  {\kappa_j}}%
}{
 \cStep{  \text{\sf u}  }{   [  {\kappa_i}  ,  \kappa_o  ]   }{  \concretesymbol{\mu}  }{ [  \ottnt{p}  \mathord{\scriptstyle @}  \mathtt{T}_{{\mathrm{1}}}  \mathord{,}\,  \concretesymbol{\sigma}  ] }{  \ottnt{n}  \mathord{\scriptstyle @}  \mathtt{T}_{pc}  }{  \tau  }{  \text{\sf k}  }{   [  {\kappa_j}  ,  \kappa_\mathtt{D}  ]   }{  \concretesymbol{\mu}  }{ [  \ottsym{(}  \ottnt{n}  \mathord{\scriptstyle @}  \mathtt{T}_{pc}  \mathord{,}\,  \text{\sf u}  \ottsym{)}  \mathord{;}\,  \ottnt{p}  \mathord{\scriptstyle @}  \mathtt{T}_{{\mathrm{1}}}  \mathord{,}\,  \concretesymbol{\sigma}  ] }{  \ottsym{0}  \mathord{\scriptstyle @}  \mathtt{T}_\mathtt{D}  }{  \ottsym{\mbox{$\backslash{}$}\mbox{$\backslash{}$}}  } }{%
{\ottdrulename{CLd\_F}}{}%
}}

\newcommand{\ottdruleCLdXXP}[1]{\ottdrule[#1]{%
\ottpremise{\phi  \ottsym{(}  \ottnt{n}  \ottsym{)}  \ottsym{=}  \text{\sf Load} \, \qquad \,  \kappa ( \ottnt{p} )   \ottsym{=}  \ottnt{m}  \mathord{\scriptstyle @}  \mathtt{T}_{{\mathrm{1}}}}%
}{
 \cStep{  \text{\sf k}  }{  \kappa  }{  \concretesymbol{\mu}  }{ [  \ottnt{p}  \mathord{\scriptstyle @}   \dummytag   \mathord{,}\,  \concretesymbol{\sigma}  ] }{  \ottnt{n}  \mathord{\scriptstyle @}   \dummytag   }{  \tau  }{  \text{\sf k}  }{  \kappa  }{  \concretesymbol{\mu}  }{ [  \ottnt{m}  \mathord{\scriptstyle @}  \mathtt{T}_{{\mathrm{1}}}  \mathord{,}\,  \concretesymbol{\sigma}  ] }{  \ottnt{n}  \mathord{+}  \ottsym{1}  \mathord{\scriptstyle @}  \mathtt{T}_\mathtt{D}  }{  \ottsym{\mbox{$\backslash{}$}\mbox{$\backslash{}$}}  } }{%
{\ottdrulename{CLd\_P}}{}%
}}

\newcommand{\ottdruleCSt}[1]{\ottdrule[#1]{%
\ottpremise{\iota  \ottsym{(}  \ottnt{n}  \ottsym{)}  \ottsym{=}  \text{\sf Store} \, \qquad \,  \concretesymbol{\mu} ( \ottnt{p} )   \ottsym{=}  \ottnt{k}  \mathord{\scriptstyle @}  \mathtt{T}_{{\mathrm{3}}}}%
\ottpremise{\kappa  =   \begin{array}{|@{\;}l@{\;}|@{\;}l@{\;}|@{\;}l@{\;}|@{\;}l@{\;}|@{\;}l@{\;}||@{\;}l@{\;}|@{\;}l@{\;}|}
                       \hline
                           \ottkw{store}  &  \mathtt{T}_{pc}  &  \mathtt{T}_{{\mathrm{1}}}  &  \mathtt{T}_{{\mathrm{2}}}  &  \mathtt{T}_{{\mathrm{3}}}  &  \mathtt{T}_{rpc}  &  \mathtt{T}_{r}  \\
                       \hline
                       \end{array} }%
\ottpremise{ \concretesymbol{\mu} ( \ottnt{p} ) \leftarrow  \ottsym{(}  \ottnt{m}  \mathord{\scriptstyle @}  \mathtt{T}_{r}  \ottsym{)}   \ottsym{=}  \concretesymbol{\mu}'}%
}{
 \cStep{  \text{\sf u}  }{  \kappa  }{  \concretesymbol{\mu}  }{ [  \ottnt{p}  \mathord{\scriptstyle @}  \mathtt{T}_{{\mathrm{1}}}  \mathord{,}\,  \ottnt{m}  \mathord{\scriptstyle @}  \mathtt{T}_{{\mathrm{2}}}  \mathord{,}\,  \concretesymbol{\sigma}  ] }{  \ottnt{n}  \mathord{\scriptstyle @}  \mathtt{T}_{pc}  }{  \tau  }{  \text{\sf u}  }{  \kappa  }{  \concretesymbol{\mu}'  }{ [  \concretesymbol{\sigma}  ] }{  \ottnt{n}  \mathord{+}  \ottsym{1}  \mathord{\scriptstyle @}  \mathtt{T}_{rpc}  }{  \ottsym{\mbox{$\backslash{}$}\mbox{$\backslash{}$}}  } }{%
{\ottdrulename{CSt}}{}%
}}

\newcommand{\ottdruleCStXXF}[1]{\ottdrule[#1]{%
\ottpremise{\iota  \ottsym{(}  \ottnt{n}  \ottsym{)}  \ottsym{=}  \text{\sf Store} \, \qquad \,  \concretesymbol{\mu} ( \ottnt{p} )   \ottsym{=}  \ottnt{k}  \mathord{\scriptstyle @}  \mathtt{T}_{{\mathrm{3}}}}%
\ottpremise{{\kappa_i}  \not=   \begin{array}{|@{\;}l@{\;}|@{\;}l@{\;}|@{\;}l@{\;}|@{\;}l@{\;}|@{\;}l@{\;}|}
                       \hline
                           \ottkw{store}  &  \mathtt{T}_{pc}  &  \mathtt{T}_{{\mathrm{1}}}  &  \mathtt{T}_{{\mathrm{2}}}  &  \mathtt{T}_{{\mathrm{3}}}  \\
                       \hline
                       \end{array}   =  {\kappa_j}}%
}{
 \cStep{  \text{\sf u}  }{   [  {\kappa_i}  ,  \kappa_o  ]   }{  \concretesymbol{\mu}  }{ [  \ottnt{p}  \mathord{\scriptstyle @}  \mathtt{T}_{{\mathrm{1}}}  \mathord{,}\,  \ottnt{m}  \mathord{\scriptstyle @}  \mathtt{T}_{{\mathrm{2}}}  \mathord{,}\,  \concretesymbol{\sigma}  ] }{  \ottnt{n}  \mathord{\scriptstyle @}  \mathtt{T}_{pc}  }{  \tau  }{  \text{\sf k}  }{   [  {\kappa_j}  ,  \kappa_\mathtt{D}  ]   }{  \concretesymbol{\mu}  }{ [  \ottsym{(}  \ottnt{n}  \mathord{\scriptstyle @}  \mathtt{T}_{pc}  \mathord{,}\,  \text{\sf u}  \ottsym{)}  \mathord{;}\,  \ottnt{p}  \mathord{\scriptstyle @}  \mathtt{T}_{{\mathrm{1}}}  \mathord{,}\,  \ottnt{m}  \mathord{\scriptstyle @}  \mathtt{T}_{{\mathrm{2}}}  \mathord{,}\,  \concretesymbol{\sigma}  ] }{  \ottsym{0}  \mathord{\scriptstyle @}  \mathtt{T}_\mathtt{D}  }{  \ottsym{\mbox{$\backslash{}$}\mbox{$\backslash{}$}}  } }{%
{\ottdrulename{CSt\_F}}{}%
}}

\newcommand{\ottdruleCStXXP}[1]{\ottdrule[#1]{%
\ottpremise{\phi  \ottsym{(}  \ottnt{n}  \ottsym{)}  \ottsym{=}  \text{\sf Store} \, \qquad \, \ottkw{store} \, \kappa \, \ottnt{p} \, \ottsym{(}  \ottnt{m}  \mathord{\scriptstyle @}  \mathtt{T}_{{\mathrm{1}}}  \ottsym{)}  \ottsym{=}  \kappa'}%
}{
 \cStep{  \text{\sf k}  }{  \kappa  }{  \concretesymbol{\mu}  }{ [  \ottnt{p}  \mathord{\scriptstyle @}   \dummytag   \mathord{,}\,  \ottnt{m}  \mathord{\scriptstyle @}  \mathtt{T}_{{\mathrm{1}}}  \mathord{,}\,  \concretesymbol{\sigma}  ] }{  \ottnt{n}  \mathord{\scriptstyle @}   \dummytag   }{  \tau  }{  \text{\sf k}  }{  \kappa'  }{  \concretesymbol{\mu}  }{ [  \concretesymbol{\sigma}  ] }{  \ottnt{n}  \mathord{+}  \ottsym{1}  \mathord{\scriptstyle @}  \mathtt{T}_\mathtt{D}  }{  \ottsym{\mbox{$\backslash{}$}\mbox{$\backslash{}$}}  } }{%
{\ottdrulename{CSt\_P}}{}%
}}

\newcommand{\ottdruleCJmp}[1]{\ottdrule[#1]{%
\ottpremise{\iota  \ottsym{(}  \ottnt{n}  \ottsym{)}  \ottsym{=}  \text{\sf Jump}}%
\ottpremise{\kappa  =   \begin{array}{|@{\;}l@{\;}|@{\;}l@{\;}|@{\;}l@{\;}|@{\;}l@{\;}|@{\;}l@{\;}||@{\;}l@{\;}|@{\;}l@{\;}|}
                       \hline
                           \ottkw{jump}  &  \mathtt{T}_{pc}  &  \mathtt{T}_{{\mathrm{1}}}  &  \mathtt{T}_\mathtt{D}  &  \mathtt{T}_\mathtt{D}  &  \mathtt{T}_{rpc}  &  \mathtt{T}_\mathtt{D}  \\
                       \hline
                       \end{array} }%
}{
 \cStep{  \text{\sf u}  }{  \kappa  }{  \concretesymbol{\mu}  }{ [  \ottnt{n'}  \mathord{\scriptstyle @}  \mathtt{T}_{{\mathrm{1}}}  \mathord{,}\,  \concretesymbol{\sigma}  ] }{  \ottnt{n}  \mathord{\scriptstyle @}  \mathtt{T}_{pc}  }{  \tau  }{  \text{\sf u}  }{  \kappa  }{  \concretesymbol{\mu}  }{ [  \concretesymbol{\sigma}  ] }{  \ottnt{n'}  \mathord{\scriptstyle @}  \mathtt{T}_{rpc}  }{  \ottsym{\mbox{$\backslash{}$}\mbox{$\backslash{}$}}  } }{%
{\ottdrulename{CJmp}}{}%
}}

\newcommand{\ottdruleCJmpXXF}[1]{\ottdrule[#1]{%
\ottpremise{\iota  \ottsym{(}  \ottnt{n}  \ottsym{)}  \ottsym{=}  \text{\sf Jump}}%
\ottpremise{{\kappa_i}  \not=   \begin{array}{|@{\;}l@{\;}|@{\;}l@{\;}|@{\;}l@{\;}|@{\;}l@{\;}|@{\;}l@{\;}|}
                       \hline
                           \ottkw{jump}  &  \mathtt{T}_{pc}  &  \mathtt{T}_{{\mathrm{1}}}  &  \mathtt{T}_\mathtt{D}  &  \mathtt{T}_\mathtt{D}  \\
                       \hline
                       \end{array}   =  {\kappa_j}}%
}{
 \cStep{  \text{\sf u}  }{   [  {\kappa_i}  ,  \kappa_o  ]   }{  \concretesymbol{\mu}  }{ [  \ottnt{n'}  \mathord{\scriptstyle @}  \mathtt{T}_{{\mathrm{1}}}  \mathord{,}\,  \concretesymbol{\sigma}  ] }{  \ottnt{n}  \mathord{\scriptstyle @}  \mathtt{T}_{pc}  }{  \tau  }{  \text{\sf k}  }{   [  {\kappa_j}  ,  \kappa_\mathtt{D}  ]   }{  \concretesymbol{\mu}  }{ [  \ottsym{(}  \ottnt{n}  \mathord{\scriptstyle @}  \mathtt{T}_{pc}  \mathord{,}\,  \text{\sf u}  \ottsym{)}  \mathord{;}\,  \ottnt{n'}  \mathord{\scriptstyle @}  \mathtt{T}_{{\mathrm{1}}}  \mathord{,}\,  \concretesymbol{\sigma}  ] }{  \ottsym{0}  \mathord{\scriptstyle @}  \mathtt{T}_\mathtt{D}  }{  \ottsym{\mbox{$\backslash{}$}\mbox{$\backslash{}$}}  } }{%
{\ottdrulename{CJmp\_F}}{}%
}}

\newcommand{\ottdruleCJmpXXP}[1]{\ottdrule[#1]{%
\ottpremise{\phi  \ottsym{(}  \ottnt{n}  \ottsym{)}  \ottsym{=}  \text{\sf Jump}}%
}{
 \cStep{  \text{\sf k}  }{  \kappa  }{  \concretesymbol{\mu}  }{ [  \ottnt{n'}  \mathord{\scriptstyle @}   \dummytag   \mathord{,}\,  \concretesymbol{\sigma}  ] }{  \ottnt{n}  \mathord{\scriptstyle @}   \dummytag   }{  \tau  }{  \text{\sf k}  }{  \kappa  }{  \concretesymbol{\mu}  }{ [  \concretesymbol{\sigma}  ] }{  \ottnt{n'}  \mathord{\scriptstyle @}  \mathtt{T}_\mathtt{D}  }{    } }{%
{\ottdrulename{CJmp\_P}}{}%
}}

\newcommand{\ottdruleCBnz}[1]{\ottdrule[#1]{%
\ottpremise{\iota  \ottsym{(}  \ottnt{n}  \ottsym{)}  \ottsym{=}  \text{\sf Bnz} \, \ottnt{k}}%
\ottpremise{\kappa  =   \begin{array}{|@{\;}l@{\;}|@{\;}l@{\;}|@{\;}l@{\;}|@{\;}l@{\;}|@{\;}l@{\;}||@{\;}l@{\;}|@{\;}l@{\;}|}
                       \hline
                           \ottkw{bnz}  &  \mathtt{T}_{pc}  &  \mathtt{T}_{{\mathrm{1}}}  &  \mathtt{T}_\mathtt{D}  &  \mathtt{T}_\mathtt{D}  &  \mathtt{T}_{rpc}  &  \mathtt{T}_\mathtt{D}  \\
                       \hline
                       \end{array} }%
\ottpremise{\ottnt{n'}  \ottsym{=}  \ottnt{n}  \mathord{+}  \ottsym{(}  \ottsym{(}  \ottnt{m}  \ottsym{=}  \ottsym{0}  \ottsym{)}  \ottsym{\mbox{?}}  \ottsym{1}  \ottsym{:}  \ottnt{k}  \ottsym{)}}%
}{
 \cStep{  \text{\sf u}  }{  \kappa  }{  \concretesymbol{\mu}  }{ [  \ottnt{m}  \mathord{\scriptstyle @}  \mathtt{T}_{{\mathrm{1}}}  \mathord{,}\,  \concretesymbol{\sigma}  ] }{  \ottnt{n}  \mathord{\scriptstyle @}  \mathtt{T}_{pc}  }{  \tau  }{  \text{\sf u}  }{  \kappa  }{  \concretesymbol{\mu}  }{ [  \concretesymbol{\sigma}  ] }{  \ottnt{n'}  \mathord{\scriptstyle @}  \mathtt{T}_{rpc}  }{  \ottsym{\mbox{$\backslash{}$}\mbox{$\backslash{}$}}  } }{%
{\ottdrulename{CBnz}}{}%
}}

\newcommand{\ottdruleCBnzXXF}[1]{\ottdrule[#1]{%
\ottpremise{\iota  \ottsym{(}  \ottnt{n}  \ottsym{)}  \ottsym{=}  \text{\sf Bnz} \, \ottnt{k}}%
\ottpremise{{\kappa_i}  \not=   \begin{array}{|@{\;}l@{\;}|@{\;}l@{\;}|@{\;}l@{\;}|@{\;}l@{\;}|@{\;}l@{\;}|}
                       \hline
                           \ottkw{bnz}  &  \mathtt{T}_{pc}  &  \mathtt{T}_{{\mathrm{1}}}  &  \mathtt{T}_\mathtt{D}  &  \mathtt{T}_\mathtt{D}  \\
                       \hline
                       \end{array}   =  {\kappa_j}}%
}{
 \cStep{  \text{\sf u}  }{   [  {\kappa_i}  ,  \kappa_o  ]   }{  \concretesymbol{\mu}  }{ [  \ottnt{m}  \mathord{\scriptstyle @}  \mathtt{T}_{{\mathrm{1}}}  \mathord{,}\,  \concretesymbol{\sigma}  ] }{  \ottnt{n}  \mathord{\scriptstyle @}  \mathtt{T}_{pc}  }{  \tau  }{  \text{\sf k}  }{   [  {\kappa_j}  ,  \kappa_\mathtt{D}  ]   }{  \concretesymbol{\mu}  }{ [  \ottsym{(}  \ottnt{n}  \mathord{\scriptstyle @}  \mathtt{T}_{pc}  \mathord{,}\,  \text{\sf u}  \ottsym{)}  \mathord{;}\,  \ottnt{m}  \mathord{\scriptstyle @}  \mathtt{T}_{{\mathrm{1}}}  \mathord{,}\,  \concretesymbol{\sigma}  ] }{  \ottsym{0}  \mathord{\scriptstyle @}  \mathtt{T}_\mathtt{D}  }{  \ottsym{\mbox{$\backslash{}$}\mbox{$\backslash{}$}}  } }{%
{\ottdrulename{CBnz\_F}}{}%
}}

\newcommand{\ottdruleCBnzXXP}[1]{\ottdrule[#1]{%
\ottpremise{\phi  \ottsym{(}  \ottnt{n}  \ottsym{)}  \ottsym{=}  \text{\sf Bnz} \, \ottnt{k} \, \qquad \, \ottnt{n'}  \ottsym{=}  \ottnt{n}  \mathord{+}  \ottsym{(}  \ottsym{(}  \ottnt{m}  \ottsym{=}  \ottsym{0}  \ottsym{)}  \ottsym{\mbox{?}}  \ottsym{1}  \ottsym{:}  \ottnt{k}  \ottsym{)}}%
}{
 \cStep{  \text{\sf k}  }{  \kappa  }{  \concretesymbol{\mu}  }{ [  \ottnt{m}  \mathord{\scriptstyle @}   \dummytag   \mathord{,}\,  \concretesymbol{\sigma}  ] }{  \ottnt{n}  \mathord{\scriptstyle @}   \dummytag   }{  \tau  }{  \text{\sf k}  }{  \kappa  }{  \concretesymbol{\mu}  }{ [  \concretesymbol{\sigma}  ] }{  \ottnt{n'}  \mathord{\scriptstyle @}  \mathtt{T}_\mathtt{D}  }{    } }{%
{\ottdrulename{CBnz\_P}}{}%
}}

\newcommand{\ottdruleCCll}[1]{\ottdrule[#1]{%
\ottpremise{\iota  \ottsym{(}  \ottnt{n}  \ottsym{)}  \ottsym{=}  \text{\sf Call}}%
\ottpremise{\kappa  =   \begin{array}{|@{\;}l@{\;}|@{\;}l@{\;}|@{\;}l@{\;}|@{\;}l@{\;}|@{\;}l@{\;}||@{\;}l@{\;}|@{\;}l@{\;}|}
                       \hline
                           \ottkw{call}  &  \mathtt{T}_{pc}  &  \mathtt{T}_{{\mathrm{1}}}  &  \mathtt{T}_\mathtt{D}  &  \mathtt{T}_\mathtt{D}  &  \mathtt{T}_{rpc}  &  \mathtt{T}_{r}  \\
                       \hline
                       \end{array} }%
}{
 \cStep{  \text{\sf u}  }{  \kappa  }{  \concretesymbol{\mu}  }{ [  \ottnt{n'}  \mathord{\scriptstyle @}  \mathtt{T}_{{\mathrm{1}}}  \mathord{,}\,  \concretefont{a}  \mathord{,}\,  \concretesymbol{\sigma}  ] }{  \ottnt{n}  \mathord{\scriptstyle @}  \mathtt{T}_{pc}  }{  \tau  }{  \text{\sf u}  }{  \kappa  }{  \concretesymbol{\mu}  }{ [  \concretefont{a}  \mathord{,}\,  \ottsym{(}  \ottnt{n}  \mathord{+}  \ottsym{1}  \mathord{\scriptstyle @}  \mathtt{T}_{r}  \mathord{,}\,  \text{\sf u}  \ottsym{)}  \mathord{;}\,  \concretesymbol{\sigma}  ] }{  \ottnt{n'}  \mathord{\scriptstyle @}  \mathtt{T}_{rpc}  }{  \ottsym{\mbox{$\backslash{}$}\mbox{$\backslash{}$}}  } }{%
{\ottdrulename{CCll}}{}%
}}

\newcommand{\ottdruleCCllXXF}[1]{\ottdrule[#1]{%
\ottpremise{\iota  \ottsym{(}  \ottnt{n}  \ottsym{)}  \ottsym{=}  \text{\sf Call}}%
\ottpremise{{\kappa_i}  \not=   \begin{array}{|@{\;}l@{\;}|@{\;}l@{\;}|@{\;}l@{\;}|@{\;}l@{\;}|@{\;}l@{\;}|}
                       \hline
                           \ottkw{call}  &  \mathtt{T}_{pc}  &  \mathtt{T}_{{\mathrm{1}}}  &  \mathtt{T}_\mathtt{D}  &  \mathtt{T}_\mathtt{D}  \\
                       \hline
                       \end{array}   =  {\kappa_j}}%
}{
 \cStep{  \text{\sf u}  }{   [  {\kappa_i}  ,  \kappa_o  ]   }{  \concretesymbol{\mu}  }{ [  \ottnt{n'}  \mathord{\scriptstyle @}  \mathtt{T}_{{\mathrm{1}}}  \mathord{,}\,  \concretefont{a}  \mathord{,}\,  \concretesymbol{\sigma}  ] }{  \ottnt{n}  \mathord{\scriptstyle @}  \mathtt{T}_{pc}  }{  \tau  }{  \text{\sf k}  }{   [  {\kappa_j}  ,  \kappa_\mathtt{D}  ]   }{  \concretesymbol{\mu}  }{ [  \ottsym{(}  \ottnt{n}  \mathord{\scriptstyle @}  \mathtt{T}_{pc}  \mathord{,}\,  \text{\sf u}  \ottsym{)}  \mathord{;}\,  \ottnt{n'}  \mathord{\scriptstyle @}  \mathtt{T}_{{\mathrm{1}}}  \mathord{,}\,  \concretefont{a}  \mathord{,}\,  \concretesymbol{\sigma}  ] }{  \ottsym{0}  \mathord{\scriptstyle @}  \mathtt{T}_\mathtt{D}  }{  \ottsym{\mbox{$\backslash{}$}\mbox{$\backslash{}$}}  } }{%
{\ottdrulename{CCll\_F}}{}%
}}

\newcommand{\ottdruleCCllXXP}[1]{\ottdrule[#1]{%
\ottpremise{\phi  \ottsym{(}  \ottnt{n}  \ottsym{)}  \ottsym{=}  \text{\sf Call}}%
}{
 \cStep{  \text{\sf k}  }{  \kappa  }{  \concretesymbol{\mu}  }{ [  \ottnt{n'}  \mathord{\scriptstyle @}   \dummytag   \mathord{,}\,  \concretefont{a}  \mathord{,}\,  \concretesymbol{\sigma}  ] }{  \ottnt{n}  \mathord{\scriptstyle @}   \dummytag   }{  \tau  }{  \text{\sf k}  }{  \kappa  }{  \concretesymbol{\mu}  }{ [  \concretefont{a}  \mathord{,}\,  \ottsym{(}  \ottnt{n}  \mathord{+}  \ottsym{1}  \mathord{\scriptstyle @}  \mathtt{T}_\mathtt{D}  \mathord{,}\,  \text{\sf k}  \ottsym{)}  \mathord{;}\,  \concretesymbol{\sigma}  ] }{  \ottnt{n'}  \mathord{\scriptstyle @}  \mathtt{T}_\mathtt{D}  }{  \ottsym{\mbox{$\backslash{}$}\mbox{$\backslash{}$}}  } }{%
{\ottdrulename{CCll\_P}}{}%
}}

\newcommand{\ottdruleCRet}[1]{\ottdrule[#1]{%
\ottpremise{\iota  \ottsym{(}  \ottnt{n}  \ottsym{)}  \ottsym{=}  \text{\sf Ret}}%
\ottpremise{\kappa  =   \begin{array}{|@{\;}l@{\;}|@{\;}l@{\;}|@{\;}l@{\;}|@{\;}l@{\;}|@{\;}l@{\;}||@{\;}l@{\;}|@{\;}l@{\;}|}
                       \hline
                           \ottkw{ret}  &  \mathtt{T}_{pc}  &  \mathtt{T}_{{\mathrm{1}}}  &  \mathtt{T}_\mathtt{D}  &  \mathtt{T}_\mathtt{D}  &  \mathtt{T}_{rpc}  &  \mathtt{T}_\mathtt{D}  \\
                       \hline
                       \end{array} }%
}{
 \cStep{  \text{\sf u}  }{  \kappa  }{  \concretesymbol{\mu}  }{ [  \ottsym{(}  \ottnt{n'}  \mathord{\scriptstyle @}  \mathtt{T}_{{\mathrm{1}}}  \mathord{,}\,  \text{\sf u}  \ottsym{)}  \mathord{;}\,  \concretesymbol{\sigma}  ] }{  \ottnt{n}  \mathord{\scriptstyle @}  \mathtt{T}_{pc}  }{  \tau  }{  \text{\sf u}  }{  \kappa  }{  \concretesymbol{\mu}  }{ [  \concretesymbol{\sigma}  ] }{  \ottnt{n'}  \mathord{\scriptstyle @}  \mathtt{T}_{rpc}  }{  \ottsym{\mbox{$\backslash{}$}\mbox{$\backslash{}$}}  } }{%
{\ottdrulename{CRet}}{}%
}}

\newcommand{\ottdruleCRetXXF}[1]{\ottdrule[#1]{%
\ottpremise{\iota  \ottsym{(}  \ottnt{n}  \ottsym{)}  \ottsym{=}  \text{\sf Ret}}%
\ottpremise{{\kappa_i}  \not=   \begin{array}{|@{\;}l@{\;}|@{\;}l@{\;}|@{\;}l@{\;}|@{\;}l@{\;}|@{\;}l@{\;}|}
                       \hline
                           \ottkw{ret}  &  \mathtt{T}_{pc}  &  \mathtt{T}_{{\mathrm{1}}}  &  \mathtt{T}_\mathtt{D}  &  \mathtt{T}_\mathtt{D}  \\
                       \hline
                       \end{array}   =  {\kappa_j}}%
}{
 \cStep{  \text{\sf u}  }{   [  {\kappa_i}  ,  \kappa_o  ]   }{  \concretesymbol{\mu}  }{ [  \ottsym{(}  \ottnt{n'}  \mathord{\scriptstyle @}  \mathtt{T}_{{\mathrm{1}}}  \mathord{,}\,  \pi  \ottsym{)}  \mathord{;}\,  \concretesymbol{\sigma}  ] }{  \ottnt{n}  \mathord{\scriptstyle @}  \mathtt{T}_{pc}  }{  \tau  }{  \text{\sf k}  }{   [  {\kappa_j}  ,  \kappa_\mathtt{D}  ]   }{  \concretesymbol{\mu}  }{ [  \ottsym{(}  \ottnt{n}  \mathord{\scriptstyle @}  \mathtt{T}_{pc}  \mathord{,}\,  \text{\sf u}  \ottsym{)}  \mathord{;}\,  \ottsym{(}  \ottnt{n'}  \mathord{\scriptstyle @}  \mathtt{T}_{{\mathrm{1}}}  \mathord{,}\,  \pi  \ottsym{)}  \mathord{;}\,  \concretesymbol{\sigma}  ] }{  \ottsym{0}  \mathord{\scriptstyle @}  \mathtt{T}_\mathtt{D}  }{  \ottsym{\mbox{$\backslash{}$}\mbox{$\backslash{}$}}  } }{%
{\ottdrulename{CRet\_F}}{}%
}}

\newcommand{\ottdruleCRetXXP}[1]{\ottdrule[#1]{%
\ottpremise{\phi  \ottsym{(}  \ottnt{n}  \ottsym{)}  \ottsym{=}  \text{\sf Ret}}%
}{
 \cStep{  \text{\sf k}  }{  \kappa  }{  \concretesymbol{\mu}  }{ [  \ottsym{(}  \ottnt{n'}  \mathord{\scriptstyle @}  \mathtt{T}_{{\mathrm{1}}}  \mathord{,}\,  \pi  \ottsym{)}  \mathord{;}\,  \concretesymbol{\sigma}  ] }{  \ottnt{n}  \mathord{\scriptstyle @}   \dummytag   }{  \tau  }{  \pi  }{  \kappa  }{  \concretesymbol{\mu}  }{ [  \concretesymbol{\sigma}  ] }{  \ottnt{n'}  \mathord{\scriptstyle @}  \mathtt{T}_{{\mathrm{1}}}  }{    } }{%
{\ottdrulename{CRet\_P}}{}%
}}

\newcommand{\ottdruleCEPushCachePtr}[1]{\ottdrule[#1]{%
\ottpremise{\phi  \ottsym{(}  \ottnt{n}  \ottsym{)}  \ottsym{=}  \ottkw{PushCachePtr}}%
}{
 \ceStep{  \text{\sf k}  }{  \concretesymbol{\mu}  }{ [  \concretesymbol{\sigma}  ] }{  \ottnt{n}  \mathord{\scriptstyle @}   \dummytag   }{  \tau  }{  \text{\sf k}  }{  \concretesymbol{\mu}  }{ [  \ottsym{(}  \ottkw{Ptr} \, \ottsym{(}   cache   \mathord{,}\,  \ottsym{0}  \ottsym{)}  \ottsym{)}  \mathord{\scriptstyle @}  \mathtt{T}_\mathtt{D}  \mathord{,}\,  \concretesymbol{\sigma}  ] }{  \ottsym{(}  \ottnt{n}  \mathord{+}  \ottsym{1}  \ottsym{)}  \mathord{\scriptstyle @}  \mathtt{T}_\mathtt{D}  }{    } }{%
{\ottdrulename{CEPushCachePtr}}{}%
}}

\newcommand{\ottdruleCEAlloc}[1]{\ottdrule[#1]{%
\ottpremise{\iota  \ottsym{(}  \ottnt{n}  \ottsym{)}  \ottsym{=}  \ottkw{Alloc} \, \qquad \, \text{\sf alloc} \, \ottnt{k} \, \text{\sf u} \, \concretefont{a} \, \concretesymbol{\mu}  \ottsym{=}  \ottsym{(}  \ottnt{id}  \mathord{,}\,  \concretesymbol{\mu}'  \ottsym{)}}%
\ottpremise{\concretesymbol{\mu}  \ottsym{(}   cache   \ottsym{)}  =   \begin{array}{|@{\;}l@{\;}|@{\;}l@{\;}|@{\;}l@{\;}|@{\;}l@{\;}|@{\;}l@{\;}||@{\;}l@{\;}|@{\;}l@{\;}|}
                       \hline
                           \text{\sf alloc}  &  \mathtt{T}_{pc}  &  \mathtt{T}_{{\mathrm{1}}}  &  \mathtt{T}_\mathtt{D}  &  \mathtt{T}_\mathtt{D}  &  \mathtt{T}_{rpc}  &  \mathtt{T}_{r}  \\
                       \hline
                       \end{array} }%
}{
 \ceStep{  \text{\sf u}  }{  \concretesymbol{\mu}  }{ [  \ottsym{(}  \ottkw{Int} \, \ottnt{k}  \ottsym{)}  \mathord{\scriptstyle @}  \mathtt{T}_{{\mathrm{1}}}  \mathord{,}\,  \concretefont{a}  \mathord{,}\,  \concretesymbol{\sigma}  ] }{  \ottnt{n}  \mathord{\scriptstyle @}  \mathtt{T}_{pc}  }{  \tau  }{  \text{\sf u}  }{  \concretesymbol{\mu}'  }{ [  \ottsym{(}  \ottkw{Ptr} \, \ottsym{(}  \ottnt{id}  \mathord{,}\,  \ottsym{0}  \ottsym{)}  \ottsym{)}  \mathord{\scriptstyle @}  \mathtt{T}_{r}  \mathord{,}\,  \concretesymbol{\sigma}  ] }{  \ottsym{(}  \ottnt{n}  \mathord{+}  \ottsym{1}  \ottsym{)}  \mathord{\scriptstyle @}  \mathtt{T}_{rpc}  }{  \ottsym{\mbox{$\backslash{}$}\mbox{$\backslash{}$}}  } }{%
{\ottdrulename{CEAlloc}}{}%
}}

\newcommand{\ottdruleCEAllocP}[1]{\ottdrule[#1]{%
\ottpremise{\phi  \ottsym{(}  \ottnt{n}  \ottsym{)}  \ottsym{=}  \ottkw{Alloc} \, \qquad \, \text{\sf alloc} \, \ottnt{k} \, \text{\sf k} \, \concretefont{a} \, \concretesymbol{\mu}  \ottsym{=}  \ottsym{(}  \ottnt{id}  \mathord{,}\,  \concretesymbol{\mu}'  \ottsym{)}}%
}{
 \ceStep{  \text{\sf k}  }{  \concretesymbol{\mu}  }{ [  \ottsym{(}  \ottkw{Int} \, \ottnt{k}  \ottsym{)}  \mathord{\scriptstyle @}   \dummytag   \mathord{,}\,  \concretefont{a}  \mathord{,}\,  \concretesymbol{\sigma}  ] }{  \ottnt{n}  \mathord{\scriptstyle @}   \dummytag   }{  \tau  }{  \text{\sf k}  }{  \concretesymbol{\mu}'  }{ [  \ottsym{(}  \ottkw{Ptr} \, \ottsym{(}  \ottnt{id}  \mathord{,}\,  \ottsym{0}  \ottsym{)}  \ottsym{)}  \mathord{\scriptstyle @}  \mathtt{T}_\mathtt{D}  \mathord{,}\,  \concretesymbol{\sigma}  ] }{  \ottsym{(}  \ottnt{n}  \mathord{+}  \ottsym{1}  \ottsym{)}  \mathord{\scriptstyle @}  \mathtt{T}_\mathtt{D}  }{  \ottsym{\mbox{$\backslash{}$}\mbox{$\backslash{}$}}  } }{%
{\ottdrulename{CEAllocP}}{}%
}}

\newcommand{\ottdruleCEPack}[1]{\ottdrule[#1]{%
\ottpremise{\phi  \ottsym{(}  \ottnt{n}  \ottsym{)}  \ottsym{=}  \ottkw{Pack}}%
}{
 \ceStep{  \text{\sf k}  }{  \concretesymbol{\mu}  }{ [  v_{{\mathrm{2}}}  \mathord{\scriptstyle @}   \dummytag   \mathord{,}\,  v_{{\mathrm{1}}}  \mathord{\scriptstyle @}   \dummytag   \mathord{,}\,  \concretesymbol{\sigma}  ] }{  \ottnt{n}  \mathord{\scriptstyle @}   \dummytag   }{  \tau  }{  \text{\sf k}  }{  \concretesymbol{\mu}  }{ [  v_{{\mathrm{1}}}  \mathord{\scriptstyle @}  v_{{\mathrm{2}}}  \mathord{,}\,  \concretesymbol{\sigma}  ] }{  \ottsym{(}  \ottnt{n}  \mathord{+}  \ottsym{1}  \ottsym{)}  \mathord{\scriptstyle @}  \mathtt{T}_\mathtt{D}  }{    } }{%
{\ottdrulename{CEPack}}{}%
}}

\newcommand{\ottdruleCEUnpack}[1]{\ottdrule[#1]{%
\ottpremise{\phi  \ottsym{(}  \ottnt{n}  \ottsym{)}  \ottsym{=}  \ottkw{Unpack}}%
}{
 \ceStep{  \text{\sf k}  }{  \concretesymbol{\mu}  }{ [  v_{{\mathrm{1}}}  \mathord{\scriptstyle @}  v_{{\mathrm{2}}}  \mathord{,}\,  \concretesymbol{\sigma}  ] }{  \ottnt{n}  \mathord{\scriptstyle @}   \dummytag   }{  \tau  }{  \text{\sf k}  }{  \concretesymbol{\mu}  }{ [  v_{{\mathrm{2}}}  \mathord{\scriptstyle @}  \mathtt{T}_\mathtt{D}  \mathord{,}\,  v_{{\mathrm{1}}}  \mathord{\scriptstyle @}  \mathtt{T}_\mathtt{D}  \mathord{,}\,  \concretesymbol{\sigma}  ] }{  \ottsym{(}  \ottnt{n}  \mathord{+}  \ottsym{1}  \ottsym{)}  \mathord{\scriptstyle @}  \mathtt{T}_\mathtt{D}  }{    } }{%
{\ottdrulename{CEUnpack}}{}%
}}

\newcommand{\ottdruleCESysCall}[1]{\ottdrule[#1]{%
\ottpremise{\iota  \ottsym{(}  \ottnt{n}  \ottsym{)}  \ottsym{=}  \text{\sf SysCall} \, \ottnt{id} \, \qquad \, T  \ottsym{(}  \ottnt{id}  \ottsym{)}  \ottsym{=}  \ottsym{(}  \ottnt{k}  \mathord{,}\,  \ottnt{n'}  \ottsym{)} \, \qquad \, \ottkw{length} \, \ottsym{(}  \concretesymbol{\sigma}_{{\mathrm{1}}}  \ottsym{)}  \ottsym{=}  \ottnt{k}}%
}{
 \ceStep{  \text{\sf u}  }{  \concretesymbol{\mu}  }{ [  \concretesymbol{\sigma}_{{\mathrm{1}}}  {\scriptstyle \mathord{+\!+} }  \concretesymbol{\sigma}_{{\mathrm{2}}}  ] }{  \ottnt{n}  \mathord{\scriptstyle @}  \mathtt{T}  }{  \tau  }{  \text{\sf k}  }{  \concretesymbol{\mu}  }{ [  \concretesymbol{\sigma}_{{\mathrm{1}}}  {\scriptstyle \mathord{+\!+} }  \ottsym{(}  \ottnt{n}  \mathord{+}  \ottsym{1}  \mathord{\scriptstyle @}  \mathtt{T}  \mathord{,}\,  \text{\sf u}  \ottsym{)}  \mathord{;}\,  \concretesymbol{\sigma}_{{\mathrm{2}}}  ] }{  \ottnt{n'}  \mathord{\scriptstyle @}  \mathtt{T}_\mathtt{D}  }{    } }{%
{\ottdrulename{CESysCall}}{}%
}}

\newcommand{\ottdruleUnit}[1]{\ottdrule[#1]{%
}{
 \{  \ottnt{P} \} \  []  \ \{ \ottnt{P} \} }{%
{\ottdrulename{Unit}}{}%
}}

\newcommand{\ottdruleCompose}[1]{\ottdrule[#1]{%
\ottpremise{ \{  \ottnt{P_{{\mathrm{1}}}} \} \  \ottnt{c_{{\mathrm{1}}}}  \ \{ \ottnt{P_{{\mathrm{2}}}} \}  \, \qquad \,  \{  \ottnt{P_{{\mathrm{2}}}} \} \  \ottnt{c_{{\mathrm{2}}}}  \ \{ \ottnt{P_{{\mathrm{3}}}} \} }%
}{
 \{  \ottnt{P_{{\mathrm{1}}}} \} \  \ottnt{c_{{\mathrm{1}}}}  {\scriptstyle \mathord{+\!+} }  \ottnt{c_{{\mathrm{2}}}}  \ \{ \ottnt{P_{{\mathrm{3}}}} \} }{%
{\ottdrulename{Compose}}{}%
}}

\newcommand{\ottdruleWeaken}[1]{\ottdrule[#1]{%
\ottpremise{\forall \, \kappa \, \concretesymbol{\sigma}  .~  \ottnt{P'}  \ottsym{(}  \kappa  \mathord{,}\,  \concretesymbol{\sigma}  \ottsym{)}  \Longrightarrow  \ottnt{P}  \ottsym{(}  \kappa  \mathord{,}\,  \concretesymbol{\sigma}  \ottsym{)}}%
\ottpremise{\forall \, \kappa \, \concretesymbol{\sigma}  .~  \ottnt{Q}  \ottsym{(}  \kappa  \mathord{,}\,  \concretesymbol{\sigma}  \ottsym{)}  \Longrightarrow  \ottnt{Q'}  \ottsym{(}  \kappa  \mathord{,}\,  \concretesymbol{\sigma}  \ottsym{)}}%
\ottpremise{ \{  \ottnt{P} \} \  \ottnt{c}  \ \{ \ottnt{Q} \} }%
}{
 \{  \ottnt{P'} \} \  \ottnt{c}  \ \{ \ottnt{Q'} \} }{%
{\ottdrulename{Weaken}}{}%
}}

\newcommand{\ottdruleAdd}[1]{\ottdrule[#1]{%
\ottpremise{\ottnt{P}  \ottsym{(}  \kappa  \mathord{,}\,  \concretesymbol{\sigma}  \ottsym{)}  :=  \exists \, \ottnt{n_{{\mathrm{1}}}} \, \mathtt{T}_{{\mathrm{1}}} \, \ottnt{n_{{\mathrm{2}}}} \, \mathtt{T}_{{\mathrm{2}}} \, \concretesymbol{\sigma}'  .~  \concretesymbol{\sigma}  \ottsym{=}  \ottnt{n_{{\mathrm{1}}}}  \mathord{\scriptstyle @}  \mathtt{T}_{{\mathrm{1}}}  \mathord{,}\,  \ottnt{n_{{\mathrm{2}}}}  \mathord{\scriptstyle @}  \mathtt{T}_{{\mathrm{2}}}  \mathord{,}\,  \concretesymbol{\sigma}' \, \wedge \, \ottnt{Q}  \ottsym{(}  \kappa  \mathord{,}\,  \ottsym{(}  \ottsym{(}  \ottnt{n_{{\mathrm{1}}}}  \mathord{+}  \ottnt{n_{{\mathrm{2}}}}  \ottsym{)}  \mathord{\scriptstyle @}  \mathtt{T}_\mathtt{D}  \mathord{,}\,  \concretesymbol{\sigma}'  \ottsym{)}  \ottsym{)}}%
}{
 \{  \ottnt{P} \} \  \ottsym{[}  \text{\sf Add}  \ottsym{]}  \ \{ \ottnt{Q} \} }{%
{\ottdrulename{Add}}{}%
}}

\newcommand{\ottdruleIf}[1]{\ottdrule[#1]{%
\ottpremise{\ottnt{P}  \ottsym{(}  \kappa  \mathord{,}\,  \concretesymbol{\sigma}  \ottsym{)}  :=  \exists \, \ottnt{n} \, \mathtt{T} \, \concretesymbol{\sigma}'  .~  \concretesymbol{\sigma}  \ottsym{=}  \ottnt{n}  \mathord{\scriptstyle @}  \mathtt{T}  \mathord{,}\,  \concretesymbol{\sigma}' \, \wedge \, \ottsym{(}  \ottnt{n}  \neq  \ottsym{0}  \Longrightarrow  \ottnt{P_{{\mathrm{1}}}}  \ottsym{(}  \kappa  \mathord{,}\,  \concretesymbol{\sigma}'  \ottsym{)}  \ottsym{)} \, \wedge \, \ottsym{(}  \ottnt{n}  \ottsym{=}  \ottsym{0}  \Longrightarrow  \ottnt{P_{{\mathrm{2}}}}  \ottsym{(}  \kappa  \mathord{,}\,  \concretesymbol{\sigma}'  \ottsym{)}  \ottsym{)}}%
\ottpremise{ \{  \ottnt{P_{{\mathrm{1}}}} \} \  \ottnt{c_{{\mathrm{1}}}}  \ \{ \ottnt{Q} \}  \, \qquad \,  \{  \ottnt{P_{{\mathrm{2}}}} \} \  \ottnt{c_{{\mathrm{2}}}}  \ \{ \ottnt{Q} \} }%
}{
 \{  \ottnt{P} \} \  \ottkw{genIf} \, \ottnt{c_{{\mathrm{1}}}} \, \ottnt{c_{{\mathrm{2}}}}  \ \{ \ottnt{Q} \} }{%
{\ottdrulename{If}}{}%
}}

\newcommand{\ottdruleRepeat}[1]{\ottdrule[#1]{%
\ottpremise{P_n  \ottsym{(}  \kappa  \mathord{,}\,  \concretesymbol{\sigma}  \ottsym{)}  :=  \exists \, \mathtt{T} \, \concretesymbol{\sigma}'  .~  \concretesymbol{\sigma}  \ottsym{=}  \ottnt{n}  \mathord{\scriptstyle @}  \mathtt{T}  \mathord{,}\,  \concretesymbol{\sigma}' \, \wedge \, \ottnt{Inv}  \ottsym{(}  \kappa  \mathord{,}\,  \concretesymbol{\sigma}  \ottsym{)}}%
\ottpremise{Q_n  \ottsym{(}  \kappa  \mathord{,}\,  \concretesymbol{\sigma}  \ottsym{)}  :=  \exists \, \mathtt{T} \, \concretesymbol{\sigma}'  .~  \concretesymbol{\sigma}  \ottsym{=}  \ottnt{n}  \mathord{\scriptstyle @}  \mathtt{T}  \mathord{,}\,  \concretesymbol{\sigma}' \, \wedge \, \forall \, \mathtt{T}'  .~  \ottnt{Inv}  \ottsym{(}  \kappa  \mathord{,}\,  \ottsym{(}  \ottsym{(}  \ottnt{n}  \ottsym{-}  \ottsym{1}  \ottsym{)}  \mathord{\scriptstyle @}  \mathtt{T}'  \mathord{,}\,  \concretesymbol{\sigma}'  \ottsym{)}  \ottsym{)}}%
\ottpremise{\forall \, \ottnt{n}  .~  \ottsym{0}  \ottsym{<}  \ottnt{n}  \Longrightarrow   \{  P_n \} \  \ottnt{c}  \ \{ Q_n \} }%
\ottpremise{\ottnt{P}  \ottsym{(}  \kappa  \mathord{,}\,  \concretesymbol{\sigma}  \ottsym{)}  :=  \exists \, \ottnt{n} \, \mathtt{T} \, \concretesymbol{\sigma}'  .~  \ottsym{0}  \le  \ottnt{n} \, \wedge \, \concretesymbol{\sigma}  \ottsym{=}  \ottnt{n}  \mathord{\scriptstyle @}  \mathtt{T}  \mathord{,}\,  \concretesymbol{\sigma}' \, \wedge \, \ottnt{Inv}  \ottsym{(}  \kappa  \mathord{,}\,  \concretesymbol{\sigma}  \ottsym{)}}%
\ottpremise{\ottnt{Q}  \ottsym{(}  \kappa  \mathord{,}\,  \concretesymbol{\sigma}  \ottsym{)}  :=  \exists \, \mathtt{T} \, \concretesymbol{\sigma}'  .~  \concretesymbol{\sigma}  \ottsym{=}  \ottsym{0}  \mathord{\scriptstyle @}  \mathtt{T}  \mathord{,}\,  \concretesymbol{\sigma}' \, \wedge \, \ottnt{Inv}  \ottsym{(}  \kappa  \mathord{,}\,  \concretesymbol{\sigma}  \ottsym{)}}%
}{
 \{  \ottnt{P} \} \  \ottkw{genFor} \, \ottnt{c}  \ \{ \ottnt{Q} \} }{%
{\ottdrulename{Repeat}}{}%
}}

\newcommand{\ottdruleELab}[1]{\ottdrule[#1]{%
\ottpremise{ \rho  \vdash   \ottnt{LE}  \downarrow  \ottnt{L} }%
\ottpremise{\ottnt{P}  \ottsym{(}  \kappa  \mathord{,}\,  \concretesymbol{\sigma}  \ottsym{)}  :=  \kappa  \ottsym{=}  \kappa_{{\mathrm{0}}} \, \wedge \, \concretesymbol{\sigma}  \ottsym{=}  \concretesymbol{\sigma}_{{\mathrm{0}}}}%
\ottpremise{\ottnt{Q}  \ottsym{(}  \kappa  \mathord{,}\,  \concretesymbol{\sigma}  \ottsym{)}  :=  \kappa  \ottsym{=}  \kappa_{{\mathrm{0}}} \, \wedge \, \concretesymbol{\sigma}  \ottsym{=}  \mathsf{Tag} \, \ottsym{(}  \ottnt{L}  \ottsym{)}  \mathord{\scriptstyle @}  \mathtt{T}_\mathtt{D}  \mathord{,}\,  \concretesymbol{\sigma}_{{\mathrm{0}}}}%
}{
 \{  \ottnt{P} \} \  \ottkw{genELab} \, \ottnt{LE}  \ \{ \ottnt{Q} \} }{%
{\ottdrulename{ELab}}{}%
}}

\newcommand{\ottdruleELabJCS}[1]{\ottdrule[#1]{%
\ottpremise{ \ottsym{(}  L_{pc}  \mathord{,}\,  \ottnt{L_{{\mathrm{1}}}}  \mathord{,}\,  \ottnt{L_{{\mathrm{2}}}}  \mathord{,}\,  \ottnt{L_{{\mathrm{3}}}}  \ottsym{)}  \vdash   \ottnt{LE}  \downarrow  \ottnt{L} }%
\ottpremise{\kappa_{{\mathrm{0}}}  =   \begin{array}{|@{\;}l@{\;}|@{\;}l@{\;}|@{\;}l@{\;}|@{\;}l@{\;}|@{\;}l@{\;}||@{\;}l@{\;}|@{\;}l@{\;}|}
                       \hline
                           \ottnt{op}  &   \mathsf{Tag} ( L_{pc} )   &   \mathsf{Tag} ( \ottnt{L_{{\mathrm{1}}}} )   &   \mathsf{Tag} ( \ottnt{L_{{\mathrm{2}}}} )   &   \mathsf{Tag} ( \ottnt{L_{{\mathrm{3}}}} )   &   \dummytag   &   \dummytag   \\
                       \hline
                       \end{array} }%
\ottpremise{\ottnt{P}  \ottsym{(}  \kappa  \mathord{,}\,  \concretesymbol{\sigma}  \ottsym{)}  :=  \kappa  \ottsym{=}  \kappa_{{\mathrm{0}}} \, \wedge \, \ottnt{Q}  \ottsym{(}  \kappa  \mathord{,}\,  \ottsym{(}  \mathsf{Tag} \, \ottsym{(}  \ottnt{L}  \ottsym{)}  \mathord{\scriptstyle @}  \mathtt{T}_\mathtt{D}  \mathord{,}\,  \concretesymbol{\sigma}  \ottsym{)}  \ottsym{)}}%
}{
 \{  \ottnt{P} \} \  \ottkw{genELab} \, \ottnt{LE}  \ \{ \ottnt{Q} \} }{%
{\ottdrulename{ELabJCS}}{}%
}}

\newcommand{\ottdruleBot}[1]{\ottdrule[#1]{%
\ottpremise{\ottnt{P}  \ottsym{(}  \kappa  \mathord{,}\,  \concretesymbol{\sigma}  \ottsym{)}  :=  \ottnt{Q}  \ottsym{(}  \kappa  \mathord{,}\,  \ottsym{(}  \mathsf{Tag} \, \ottsym{(}  \bot  \ottsym{)}  \mathord{\scriptstyle @}  \mathtt{T}_\mathtt{D}  \mathord{,}\,  \concretesymbol{\sigma}  \ottsym{)}  \ottsym{)}}%
}{
 \{  \ottnt{P} \} \  \ottkw{genBot}  \ \{ \ottnt{Q} \} }{%
{\ottdrulename{Bot}}{}%
}}

\newcommand{\ottdruleJoin}[1]{\ottdrule[#1]{%
\ottpremise{\ottnt{P}  \ottsym{(}  \kappa  \mathord{,}\,  \concretesymbol{\sigma}  \ottsym{)}  :=  \exists \, \ottnt{L} \, \ottnt{L'} \, \concretesymbol{\sigma}'  .~  \concretesymbol{\sigma}  \ottsym{=}  \mathsf{Tag} \, \ottsym{(}  \ottnt{L}  \ottsym{)}  \mathord{\scriptstyle @}  \mathtt{T}_\mathtt{D}  \mathord{,}\,  \mathsf{Tag} \, \ottsym{(}  \ottnt{L'}  \ottsym{)}  \mathord{\scriptstyle @}  \mathtt{T}_\mathtt{D}  \mathord{,}\,  \concretesymbol{\sigma}' \, \wedge \, \ottnt{Q}  \ottsym{(}  \kappa  \mathord{,}\,  \mathsf{Tag} \, \ottsym{(}  \ottnt{L}  \mathord{\vee}  \ottnt{L'}  \ottsym{)}  \mathord{\scriptstyle @}  \mathtt{T}_\mathtt{D}  \mathord{,}\,  \concretesymbol{\sigma}'  \ottsym{)}}%
}{
 \{  \ottnt{P} \} \  \ottkw{genJoin}  \ \{ \ottnt{Q} \} }{%
{\ottdrulename{Join}}{}%
}}

\newcommand{\ottdruleFlows}[1]{\ottdrule[#1]{%
\ottpremise{\ottnt{P}  \ottsym{(}  \kappa  \mathord{,}\,  \concretesymbol{\sigma}  \ottsym{)}  :=  \exists \, \ottnt{L} \, \ottnt{L'} \, \concretesymbol{\sigma}'  .~  \concretesymbol{\sigma}  \ottsym{=}  \mathsf{Tag} \, \ottsym{(}  \ottnt{L}  \ottsym{)}  \mathord{\scriptstyle @}  \mathtt{T}_\mathtt{D}  \mathord{,}\,  \mathsf{Tag} \, \ottsym{(}  \ottnt{L'}  \ottsym{)}  \mathord{\scriptstyle @}  \mathtt{T}_\mathtt{D}  \mathord{,}\,  \concretesymbol{\sigma}' \, \wedge \, \ottnt{Q}  \ottsym{(}  \kappa  \mathord{,}\,  \ottsym{(}   \text{if $ \ottnt{L}  \le  \ottnt{L'} $ then $ \ottsym{1} $ else $ \ottsym{0} $}   \ottsym{)}  \mathord{\scriptstyle @}  \mathtt{T}_\mathtt{D}  \mathord{,}\,  \concretesymbol{\sigma}'  \ottsym{)}}%
}{
 \{  \ottnt{P} \} \  \ottkw{genFlows}  \ \{ \ottnt{Q} \} }{%
{\ottdrulename{Flows}}{}%
}}

\newcommand{\ottdruleAppRXXOKJCS}[1]{\ottdrule[#1]{%
\ottpremise{ ${\it Rule}$_ \mathcal{R} ( \ottnt{op} )   \ottsym{=}   \langle  \ottnt{allow} ,  e_{rpc} ,  e_{r}  \rangle }%
\ottpremise{\kappa_{{\mathrm{0}}}  =   \begin{array}{|@{\;}l@{\;}|@{\;}l@{\;}|@{\;}l@{\;}|@{\;}l@{\;}|@{\;}l@{\;}||@{\;}l@{\;}|@{\;}l@{\;}|}
                       \hline
                           \ottnt{op}  &   \mathsf{Tag} ( L_{pc} )   &   \mathsf{Tag} ( \ottnt{L_{{\mathrm{1}}}} )   &   \mathsf{Tag} ( \ottnt{L_{{\mathrm{2}}}} )   &   \mathsf{Tag} ( \ottnt{L_{{\mathrm{3}}}} )   &   \dummytag   &   \dummytag   \\
                       \hline
                       \end{array} }%
\ottpremise{ \vdash_{ \mathcal{R} } \ruleeval{ \ottsym{(}  L_{pc}  \mathord{,}\,  \ottnt{L_{{\mathrm{1}}}}  \mathord{,}\,  \ottnt{L_{{\mathrm{2}}}}  \mathord{,}\,  \ottnt{L_{{\mathrm{3}}}}  \ottsym{)} }{ \ottnt{op} }{ L_{rpc} }{ L_r } }%
\ottpremise{\ottnt{P}  \ottsym{(}  \kappa  \mathord{,}\,  \concretesymbol{\sigma}  \ottsym{)}  :=  \kappa  \ottsym{=}  \kappa_{{\mathrm{0}}} \, \wedge \, \ottnt{Q}  \ottsym{(}  \kappa  \mathord{,}\,  \ottsym{(}  \ottsym{1}  \mathord{\scriptstyle @}  \mathtt{T}_\mathtt{D}  \mathord{,}\,  \mathsf{Tag} \, \ottsym{(}  L_r  \ottsym{)}  \mathord{\scriptstyle @}  \mathtt{T}_\mathtt{D}  \mathord{,}\,  \mathsf{Tag} \, \ottsym{(}  L_{rpc}  \ottsym{)}  \mathord{\scriptstyle @}  \mathtt{T}_\mathtt{D}  \mathord{,}\,  \concretesymbol{\sigma}  \ottsym{)}  \ottsym{)}}%
}{
 \{  \ottnt{P} \} \  \ottkw{genApplyRule} \,  \langle  \ottnt{allow} ,  e_{rpc} ,  e_{r}  \rangle   \ \{ \ottnt{Q} \} }{%
{\ottdrulename{AppR\_OKJCS}}{}%
}}

\newcommand{\ottdruleAppRXXOK}[1]{\ottdrule[#1]{%
\ottpremise{ \vdash_{ \mathcal{R} } \ruleeval{ \ottsym{(}  L_{pc}  \mathord{,}\,  \ottnt{L_{{\mathrm{1}}}}  \mathord{,}\,  \ottnt{L_{{\mathrm{2}}}}  \mathord{,}\,  \ottnt{L_{{\mathrm{3}}}}  \ottsym{)} }{ \ottnt{op} }{ L_{rpc} }{ L_r } }%
\ottpremise{\kappa_{{\mathrm{0}}}  =   \begin{array}{|@{\;}l@{\;}|@{\;}l@{\;}|@{\;}l@{\;}|@{\;}l@{\;}|@{\;}l@{\;}||@{\;}l@{\;}|@{\;}l@{\;}|}
                       \hline
                           \ottnt{op}  &   \mathsf{Tag} ( L_{pc} )   &   \mathsf{Tag} ( \ottnt{L_{{\mathrm{1}}}} )   &   \mathsf{Tag} ( \ottnt{L_{{\mathrm{2}}}} )   &   \mathsf{Tag} ( \ottnt{L_{{\mathrm{3}}}} )   &   \dummytag   &   \dummytag   \\
                       \hline
                       \end{array} }%
\ottpremise{\ottnt{P}  \ottsym{(}  \kappa  \mathord{,}\,  \concretesymbol{\sigma}  \ottsym{)}  :=  \kappa  \ottsym{=}  \kappa_{{\mathrm{0}}} \, \wedge \, \concretesymbol{\sigma}  \ottsym{=}  \concretesymbol{\sigma}_{{\mathrm{0}}}}%
\ottpremise{\ottnt{Q}  \ottsym{(}  \kappa  \mathord{,}\,  \concretesymbol{\sigma}  \ottsym{)}  :=  \ottsym{(}  \kappa  \ottsym{=}  \kappa_{{\mathrm{0}}}  \ottsym{)} \, \wedge \, \ottsym{(}  \concretesymbol{\sigma}  \ottsym{=}  \ottsym{1}  \mathord{\scriptstyle @}  \mathtt{T}_\mathtt{D}  \mathord{,}\,  \mathsf{Tag} \, \ottsym{(}  L_r  \ottsym{)}  \mathord{\scriptstyle @}  \mathtt{T}_\mathtt{D}  \mathord{,}\,  \mathsf{Tag} \, \ottsym{(}  L_{rpc}  \ottsym{)}  \mathord{\scriptstyle @}  \mathtt{T}_\mathtt{D}  \mathord{,}\,  \concretesymbol{\sigma}_{{\mathrm{0}}}  \ottsym{)}}%
}{
 \{  \ottnt{P} \} \  \ottkw{genApplyRule} \,  \langle  \ottnt{allow} ,  e_{rpc} ,  e_{r}  \rangle   \ \{ \ottnt{Q} \} }{%
{\ottdrulename{AppR\_OK}}{}%
}}

\newcommand{\ottdruleOXXCompose}[1]{\ottdrule[#1]{%
\ottpremise{ \{  \ottnt{P_{{\mathrm{1}}}} \} \  \ottnt{c_{{\mathrm{1}}}}  \ \{ \ottnt{P_{{\mathrm{2}}}} \}  \, \qquad \,  \{  \ottnt{P_{{\mathrm{2}}}}  \} \  \ottnt{c_{{\mathrm{2}}}}   \ \{  \ottnt{P_{{\mathrm{3}}}}  \}^{ O }_{ \concretefont{pc} } }%
}{
 \{  \ottnt{P_{{\mathrm{1}}}}  \} \  \ottnt{c_{{\mathrm{1}}}}  {\scriptstyle \mathord{+\!+} }  \ottnt{c_{{\mathrm{2}}}}   \ \{  \ottnt{P_{{\mathrm{3}}}}  \}^{ O }_{ \concretefont{pc} } }{%
{\ottdrulename{O\_Compose}}{}%
}}

\newcommand{\ottdruleOXXAppend}[1]{\ottdrule[#1]{%
\ottpremise{ \{  \ottnt{P}  \} \  \ottnt{c_{{\mathrm{1}}}}   \ \{  \ottnt{Q}  \}^{ O }_{ \concretefont{pc} } }%
}{
 \{  \ottnt{P}  \} \  \ottnt{c_{{\mathrm{1}}}}  {\scriptstyle \mathord{+\!+} }  \ottnt{c_{{\mathrm{2}}}}   \ \{  \ottnt{Q}  \}^{ O }_{ \concretefont{pc} } }{%
{\ottdrulename{O\_Append}}{}%
}}

\newcommand{\ottdruleOXXGenRet}[1]{\ottdrule[#1]{%
\ottpremise{\ottnt{P}  \ottsym{(}  \kappa  \mathord{,}\,  \concretesymbol{\sigma}  \ottsym{)}  :=  \exists \, \concretesymbol{\sigma}'  .~  \ottnt{Q}  \ottsym{(}  \kappa  \mathord{,}\,  \concretesymbol{\sigma}'  \ottsym{)} \, \wedge \, \concretesymbol{\sigma}  \ottsym{=}  \ottsym{(}  \concretefont{pc}  \mathord{,}\,  \text{\sf u}  \ottsym{)}  \mathord{;}\,  \concretesymbol{\sigma}'}%
\ottpremise{O  \ottsym{(}  \kappa  \mathord{,}\,  \concretesymbol{\sigma}  \ottsym{)}  :=   \mathtt{Success} }%
}{
 \{  \ottnt{P}  \} \  \ottsym{[}  \text{\sf Ret}  \ottsym{]}   \ \{  \ottnt{Q}  \}^{ O }_{ \concretefont{pc} } }{%
{\ottdrulename{O\_GenRet}}{}%
}}

\newcommand{\ottdruleOXXGenJump}[1]{\ottdrule[#1]{%
\ottpremise{\ottnt{P}  \ottsym{(}  \kappa  \mathord{,}\,  \concretesymbol{\sigma}  \ottsym{)}  :=  \exists \, \concretesymbol{\sigma}'  .~  \ottnt{Q}  \ottsym{(}  \kappa  \mathord{,}\,  \concretesymbol{\sigma}'  \ottsym{)} \, \wedge \, \concretesymbol{\sigma}  \ottsym{=}  \ottsym{(}  \ottsym{-}  \ottsym{1}  \ottsym{)}  \mathord{\scriptstyle @}   \dummytag   \mathord{,}\,  \concretesymbol{\sigma}'}%
\ottpremise{O  \ottsym{(}  \kappa  \mathord{,}\,  \concretesymbol{\sigma}  \ottsym{)}  :=   \mathtt{Failure} }%
}{
 \{  \ottnt{P}  \} \  \ottsym{[}  \text{\sf Jump}  \ottsym{]}   \ \{  \ottnt{Q}  \}^{ O }_{ \concretefont{pc} } }{%
{\ottdrulename{O\_GenJump}}{}%
}}